\documentclass[acmsmall]{acmart}

\AtBeginDocument{%
  }

\setcopyright{cc}
\setcctype{by}
\acmJournal{PACMMOD}
\acmYear{2026} \acmVolume{4} \acmNumber{3 (SIGMOD)} \acmArticle{129}
\acmMonth{6} \acmDOI{10.1145/3802006}

\received{October 2025}
\received[revised]{January 2026}
\received[accepted]{February 2026}

\ccsdesc[500]{Security and privacy~Database and storage security}
\ccsdesc[500]{Information systems~Data management systems}

\keywords{Differential privacy, Motif counting}

\usepackage{graphicx}
\usepackage{amsthm}

\usepackage{amssymb} 
\usepackage{amsmath}
\usepackage{color}
\usepackage{array}
\usepackage{mathtools}
\usepackage{enumitem}
\usepackage{changepage}
\usepackage{multirow}
\usepackage{appendix}
\usepackage{tabularx}
\usepackage[normalem]{ulem}
\usepackage{svg} 
\usepackage{placeins} 
\usepackage{subcaption}
\usepackage{colortbl}
\usepackage{threeparttable}
\usepackage{marginnote}

\newif\ifincludeappendix

\usepackage[linesnumbered,ruled,vlined]{algorithm2e}

\includeappendixtrue

\begin{document}

\title{Acyclic Graph Pattern Counting under Local Differential Privacy}

\author{Yihua Hu}
\affiliation{%
  \institution{Nanyang Technological University}
  \city{Singapore}
  \country{Singapore}
}
\email{yihua001@e.ntu.edu.sg}

\author{Kuncan Wang}
\affiliation{%
  \institution{Nanyang Technological University}
  \city{Singapore}
  \country{Singapore}
}
\email{kuncan001@e.ntu.edu.sg}

\author{Wei Dong}
\authornote{Wei Dong is the corresponding author.}
\affiliation{%
  \institution{Nanyang Technological University}
  \city{Singapore}
  \country{Singapore}
}
\email{wei_dong@ntu.edu.sg}

\renewcommand{\shortauthors}{Yihua Hu, Kuncan Wang, and Wei Dong}

\begin{abstract}
Graph pattern counting serves as a cornerstone of network analysis with extensive real-world applications.
Its integration with local differential privacy (LDP) has gained growing attention for protecting sensitive graph information in decentralized settings.
However, existing LDP frameworks are largely ad hoc, offering solutions only for specific patterns such as triangles and stars. 
A general mechanism for counting arbitrary graph patterns, even for the subclass of acyclic patterns, has remained an open problem.
To fill this gap, we present the first general solution for counting arbitrary acyclic patterns under LDP. 
We identify and tackle two fundamental challenges: generalizing pattern construction from distributed data and eliminating node duplication during the construction. 
To address the first challenge, we propose an LDP-tailored recursive subpattern counting framework that incrementally builds patterns across multiple communication rounds.
For the second challenge, we apply a random marking technique that restricts each node to a unique position in the pattern during computation.
Our mechanism achieves strong utility guarantees: for any acyclic graph pattern with $k$ edges, we achieve an additive error of $\tilde{O}(\sqrt{N}d(G)^k)$, where $N$ is the number of nodes and $d(G)$ is the maximum degree of the input graph $G$.
Experiments on real-world graph datasets across multiple types of acyclic patterns demonstrate that our mechanisms achieve up to ${46}$-${2600}\times$ improvement in utility and ${300}$-$650\times$ reduction in communication cost compared to the baseline methods.

\end{abstract}

\maketitle

\section{Introduction}


Graph pattern counting is fundamental for extracting meaningful insights from graph data, with applications in community detection~\cite{palla2005uncovering}, fraud detection~\cite{akoglu2013anomaly}, and network characterization~\cite{milo2002network}.
However, directly releasing such analytical results can pose significant privacy risks, as adversaries may infer sensitive information about nodes or edges in the graph~\cite{narayanan2009anonymizing, wu2021adapting, zhang2022inference}.
To address this privacy concern, \emph{differential privacy} (DP)~\cite{dwork2006calibrating} has emerged as a gold-standard framework for privacy-preserving graph data analysis. 
By injecting carefully calibrated noise into the computation, DP prevents the reliable inference of any single node or edge's presence, thereby providing node- or edge-level privacy guarantees.


Most existing studies on DP in graphs adopt the \emph{central differential privacy} (central-DP) model~\cite{nissim2007smooth,karwa2011private,zhang2015private,dong2023continual}, where a trusted curator accesses the entire dataset and applies the DP mechanism globally.
While effective in utility (i.e., low error), this model assumes full trust in the curator, making it vulnerable to data breaches and privacy attacks. 
To mitigate this risk, the \emph{local differential privacy} (LDP)~\cite{kasiviswanathan2011can,cormode2018privacy} model has gained traction, wherein each data holder privatizes their data locally before sharing it with any external party, including the analyzer and other data holders. 
In the graph setting, each node functions as a data holder, with access only to its incident edges.
Following prior work on LDP-based graph pattern counting~\cite{sun2019analyzing,ye2020towards,ye2020lf,Imola2021,Imola2022,eden2023triangle,he2025robust}, we adopt the same edge-level privacy setting, referred to as edge-LDP.
Consequently, any released information is well protected, offering robust privacy guarantees even under adversarial conditions.



However, realizing graph pattern counting under LDP is fundamentally challenging due to the decentralized nature of data collection. 
More specifically, constructing a graph pattern typically requires aggregating information distributed across multiple nodes, which is inherently complex under the LDP privacy requirement. 
As a result, existing works have primarily focused on simple, localized structures such as stars~\cite{Imola2021} and triangles~\cite{ye2020towards, ye2020lf,Imola2021, Imola2022,eden2023triangle, he2025robust}, where the required information can be gathered with minimal coordination between nodes.
Despite this progress, no universal solution exists, even for the subclass of acyclic pattern counting queries, which serve as fundamental primitives in network analysis~\cite{milo2002network, itzkovitz2005subgraphs,katz1953new,liben2003link} and relational database queries~\cite{yannakakis1981algorithms, gottlob2001complexity} in the non-private setting.
Furthermore, realizing acyclic pattern counting under LDP is natural and practically important when the graph is inherently distributed across multiple parties (e.g., across social platforms).
In this work, we identify the following two unresolved challenges in enabling acyclic graph pattern counting under LDP.

\begin{table}[t]
\centering
\renewcommand{\arraystretch}{1.2}
\resizebox{\textwidth}{!}{
\begin{tabular}{l c *{3}{c} *{3}{c}}
\toprule
 &  & \multicolumn{3}{c}{\textbf{Our Results}} & \multicolumn{3}{c}{\textbf{SOTA}} \\
\cmidrule(lr){3-5}\cmidrule(lr){6-8}
\textbf{Task} & \textbf{Query Result}
& \textbf{Error} & \textbf{Comm.} & \textbf{Time}
& \textbf{Error} & \textbf{Comm.} & \textbf{Time} \\

\midrule
$k$-line walk counting
& $\leq N d(G)^{k}$
& $\tilde{O}(\sqrt{N}d(G)^{k-1})$ & $O(M + N)$ & $O(N d(G))$
& 
{-}
& {$O(N^2)$} & {$O(N d(G))$} \\

$k$-line path counting
& $\leq N d(G)^{k}$
& $\tilde{O}(\sqrt{N} d(G)^{k})$ & $O(M + N)$ & $O(Nd(G))$
& {$O(N^k)$} & $O(N^2)$ & $O(N^{k+1})$ \\

$k$-edge acyclic pattern counting
& $\leq N d(G)^{k}$
& $\tilde{O}(\sqrt{N} d(G)^{k})$ & $O(M + N)$ & $O(Nd(G))$
& {$O(N^k)$} & $O(N^2)$ & $O(N^{k+1})$ \\

\bottomrule
\end{tabular} 
}
\caption{
Summary of results for three pattern counting queries under edge-LDP.\protect\footnotemark\protect\footnotemark
Here, $N$ is the number of nodes, $M$ is the number of edges, $d(G)$ is the maximum degree of graph $G$, and $k$ is the number of edges in the pattern.
SOTA denotes the state-of-the-art method: {\cite{betzer2024publishing} for $k$-line walk counting, and \cite{suppakitpaisarn2025counting} for both $k$-line path counting and $k$-edge acyclic pattern counting.}
{The error bound in \cite{betzer2024publishing} depends on ad hoc parameters and is therefore not included.
}
}
\label{tab:task-comparison}
\end{table}

\footnotetext[1]{
We use $\Omega(\cdot)$ to denote asymptotic lower bounds, $O(\cdot)$ for asymptotic upper bounds, and $\tilde{O}(\cdot)$ to suppress logarithmic factors. 
The number of pattern edges $k$, privacy budget $\varepsilon$, and error probability $\beta$ are treated as constants and omitted from the asymptotic notation.
}

\footnotetext[2]{
All results in the introduction hold with constant probability.
}

\textbf{Challenge 1. Generalizing acyclic pattern construction.} 
Existing LDP pattern counting solutions typically exploit structure-specific properties to guide computation, making them hard to extend to more complex or varied patterns.
For example, in star counting, each node independently counts the number of stars centered at itself and sends a perturbed count to the analyzer for aggregation.
In triangle counting, a common approach is for each node to locally enumerate pairs of its neighbors and complete the triangle using the noisy information of the last edge.
In contrast, general acyclic pattern counting should support arbitrary sizes and structures without relying on any pattern-specific assumptions.

\textbf{Challenge 2. Eliminating node duplication.}  
Acyclic patterns require that all nodes within a pattern instance be distinct, since repeated nodes would form cycles.
While this constraint also applies to stars and triangles, duplication is easy to prevent there because all involved nodes lie within a local neighborhood, which allows a single node to ensure that no repeated nodes appear in the pattern. 
In contrast, general acyclic patterns may span distant nodes, making it difficult for any single node to detect and eliminate duplication.

A straightforward approach to addressing both challenges is to process data centrally on the analyzer using the standard LDP technique of \emph{Warner's Randomized Response} (RR)~\cite{warner1965randomized}, {as recently adopted by~\cite{suppakitpaisarn2025counting}}.
Specifically, each node privatizes its incident edge list using RR and sends it to the analyzer, which then reconstructs a global noisy graph and performs analytics on it.
However, this method yields poor utility with an error bound of {$O(N^k)$}, where $N$ is the number of nodes and $k$ is the number of edges in the target pattern. 
In contrast, the true pattern count is at most $N d(G)^k$, where $d(G)$ is the maximum degree. 
{For many real-world networks with relatively uniform degree distributions, the count typically scales as $O(N d(G)^k)$.
This gap highlights the inefficiency of the RR-based method.
In the central-DP literature, algorithms often achieve error bounds of the form $\tilde{O}(\mathrm{poly}(d(G)))$~\cite{johnson2018towards}.
Since LDP typically incurs an additional $O(\sqrt{N})$ overhead~\cite{bonawitz2017practical} compared to central-DP, this raises the central question of our work:}




\emph{{Is there a general solution for counting arbitrary acyclic graph patterns under edge-LDP with an error of $\tilde{O}(\sqrt{N}\,\mathrm{poly}(d(G)))$?}}


This paper answers the question affirmatively. 
To address \textbf{Challenge~1}, we 
recursively reduce acyclic pattern counting to subpattern counting, where each subpattern is also acyclic, and solve it step by step.
Consider path counting as an example, where the goal is to count the number of paths with $k$ edges, namely $k$-line paths. 
In each round, a node collects from its neighbors the counts of paths of a certain length ending at those neighbors, starting from the base case of $0$-line paths, each initialized with a count of $1$. 
The node then aggregates these counts to compute the number of paths that are one edge longer and end at itself.
After $k$ rounds, each node holds the count of $k$-line paths ending at itself, and the global count is obtained by summing over all nodes. 
To address \textbf{Challenge~2}, {we develop a \emph{random marking} technique inspired by color-coding~\cite{alon1995color}, which was originally proposed to reduce the runtime of non-private path and cycle detection.}
With random marking, each node is restricted to appear only at a designated position in the pattern during computation. 
This technique can be seamlessly combined with our first solution by allowing each node to participate in only one specific round. 
By integrating both strategies, we develop the first general solution for acyclic pattern counting under edge-LDP. 
Our main contributions are as follows:

\begin{enumerate}
\item 
As our first result, we address $k$-line walk counting under edge-LDP, which serves as the foundation of our mechanisms for acyclic patterns. 
Our mechanism runs in $k-1$ rounds and produces an unbiased estimate of the true number of $k$-line walks, with an error of $\tilde{O}(\sqrt{N} d(G)^{k-1})$.
Summing over all rounds, the mechanism incurs an overall communication cost of $O(M+N)$, a computation cost of $O(N)$ on the analyzer, and $O(d(G))$ on each node.

\item
We next address $k$-line path counting by applying random marking to our $k$-line walk counting solution.
The resulting $k$-round mechanism also produces an unbiased estimate, with a slightly increased error of $\tilde{O}(\sqrt{N}d(G)^k)$. 
Moreover, it maintains the same overall communication cost of $O(M+N)$ and computation costs of $O(N)$ on the analyzer and $O(d(G))$ on each node.

\item 
We further generalize our $k$-line path counting solution to handle arbitrary acyclic patterns with $k$ edges. 
The resulting mechanism runs in at most $k$ rounds and produces an unbiased estimate, with asymptotically the same error of $\tilde{O}(\sqrt{N} d(G)^k)$.
The communication and computation costs remain unchanged, with a total communication cost of $O(M+N)$, a computation cost of $O(N)$ on the analyzer, and $O(d(G))$ on each node.

\end{enumerate}

Experimental results show that our mechanism achieves a relative error below ${8}\%$ on $k$-line walk counting queries. 
For acyclic pattern counting tasks, including $k$-line path counting and $k$-edge acyclic pattern counting, our mechanisms yield up to ${46}$–${2600}\times$ improvement in utility and ${300}$–$650\times$ reduction in communication cost over the RR-based approach. 


\section{Related Work}


{
In the literature of graph pattern counting, acyclic patterns play a central role due to their wide applicability and favorable computational properties. 
They are commonly used in motif analysis~\cite{milo2002network, itzkovitz2005subgraphs}, where counting subtrees helps characterize branching or hierarchical structures in networks. 
Moreover, counts of short paths serve as important features in diffusion modeling and link prediction~\cite{katz1953new, liben2003link}. 
From a complexity standpoint, there is a natural gap between cyclic and acyclic pattern counting queries in the relational database setting~\cite{ngo2018worst,abo2016faq}. 
As a result, many works focus specifically on acyclic queries, where specialized algorithms offer stronger theoretical runtime guarantees~\cite{yannakakis1981algorithms, gottlob2001complexity}. 
Our work follows this line of reasoning by focusing exclusively on acyclic patterns.
}

Moving towards the privacy-preserving setting, graph pattern counting has been extensively explored under various DP models, with the central-DP model receiving the most attention. 
Under edge-DP, \citet{nissim2007smooth} and \citet{karwa2011private} leveraged \emph{sensitivity-based mechanisms}; the former focuses on triangle counting, while the latter extends the approach to $k$-star and $k$-triangle counting. \citet{zhang2015private} proposed the \emph{ladder function}, which supports triangle counting, $k$-star counting, $k$-triangle counting, and $k$-clique counting. 
\citet{johnson2018towards} introduced \emph{elastic sensitivity}, enabling, for the first time, the counting of arbitrary graph patterns under edge-DP. \citet{dong21:residual} later proposed \emph{residual sensitivity}, which also handles arbitrary graph pattern counting and achieves provably near-optimal error~\cite{dong2021nearly}.
Additionally, \citet{fichtenberger2021differentially} and \citet{dong2024continual} addressed edge-DP graph pattern counting in dynamic graph settings.
In the context of node-DP, \citet{kasiviswanathan2013analyzing} and \citet{blocki2013differentially} addressed arbitrary pattern counting but under assumptions on node degrees. 
\citet{chen2013recursive} introduced the \emph{recursive mechanism}, which significantly improves utility but suffers from high computational cost.
More recently, \citet{dong2022r2t,dong2024instance} achieved advances in both utility and efficiency for general pattern counting queries.
\citet{fang2022shifted} proposed a user-DP framework for monotonic functions that subsumes the node-DP pattern counting problem.
\citet{hu2025n2e} extended arbitrary edge-DP pattern counting mechanisms to node-DP with a multiplicative $\tilde{O}(d(G))$ error overhead.


In contrast to work under the central-DP model, which studies both edge-DP and node-DP, research under the LDP model has primarily focused on edge-LDP.
\citet{sun2019analyzing} studied multi-round counting of triangles, $3$-line paths, and $k$-cliques under an extended local view of nodes. 
\citet{ye2020towards,ye2020lf} addressed triangle counting performed in a single round. 
\citet{Imola2021} proposed an order-optimal method for $k$-star counting and a higher-utility solution for triangle counting with single- and two-round variants. 
Building on this, \citet{Imola2022} further reduced the communication cost for triangle counting in the two-round setting. 
\citet{eden2023triangle} established lower bounds for triangle counting in both single- and multi-round settings.
Recently, \citet{he2024butterfly} presented the first multi-round LDP solution for butterfly counting on bipartite graphs, and \citet{he2025robust} improved the utility and efficiency of triangle counting using a multi-round scheme.

{Beyond the above works, \citet{betzer2024publishing} proposed the first multi-round solution for walk counting, a fundamental problem also addressed in our work as a step toward general acyclic pattern counting. 
While their approach shares our strategy of walk-based decomposition, it relies on a clipping threshold, which results in a biased estimator. 
They also briefly mention an idea similar to our non-clipping approach, but without providing any theoretical utility guarantees.}
{\citet{suppakitpaisarn2025counting} proposed a one-round RR-based solution for general pattern counting queries, but both the error and runtime scale exponentially with respect to $N$, making the approach impractical for large real-world graphs.
}

Another DP model in the decentralized setting is the shuffle model, where only a few works have investigated graph pattern counting. \citet{imola2022differentially} proposed the first one-round edge-DP solutions under the shuffle model for triangle and 4-cycle counting.

\section{Preliminaries}

\subsection{Notations}

We denote an unweighted, undirected graph as $G = (V, E)$, where $V = \{v_1, v_2, \ldots, v_N\}$ is the set of nodes, and $E \subseteq V \times V$ is the set of edges.
Self-loops, i.e., edges of the form $(u,u)$, are not allowed.
Let $N = |V|$ and $M = |E|$ denote the number of nodes and edges in $G$, respectively. 
For a node $v_i \in V$, 
let $\mathcal{E}(i)$ be the set of incident edges of node $v_i$, and
let $\mathcal{N}(i)$ denote the indices of its neighbors, where $\mathcal{N}(i) = \{j \mid (v_i, v_j) \in \mathcal{E}(i)\}$. 
The degree of $v_i$ is defined as $d(v_{i}) = |\mathcal{N}(i)|$, and $d(G)$ denotes the maximum degree in $G$. 
If two graphs $G = (V, E)$ and $G' = (V', E')$ differ by exactly one edge $e = (v_i, v_j)$, 
we write $G \sim_e G'$, or simply $G \sim G'$.

\subsection{Graph Pattern Counting}
\label{subsec:prelim_pattern_counting}

Given any acyclic graph pattern, we use $k$ to denote the number of edges in the pattern. 
Let $Q: \mathcal{G} \to \mathbb{Z}_{\geq 0}$ be the pattern counting query that returns the number of subgraphs in an input graph $G$.
A common type of acyclic pattern is a path, which is a sequence of distinct nodes connected by edges. 
In contrast, walks may include repeated nodes, making them computationally easier to count in the non-private setting: walk counting is solvable in polynomial time, whereas path counting is provably \#W[1]-complete~\cite{flum2004parameterized}.
We refer to a path with $k$ edges as a $k$-line path, and a walk with $k$ edges as a $k$-line walk.
Note that walk counting is not considered part of acyclic pattern counting, as walks can contain cycles, e.g., a $3$-line walk can form a triangle.










\subsection{Differential Privacy on Graphs}







Most works on DP in graph analytics operate under the central-DP setting, where a trusted data curator has access to the entire graph dataset. 
Upon receiving a query, the curator performs analytics and ensures that the output satisfies DP. 
This guarantees that one cannot distinguish between two neighboring graphs based on the output, thereby preserving privacy.
Formally, a randomized mechanism is said to satisfy DP if the following holds:

\begin{definition}[Differential Privacy]
\label{def:dp}
A randomized mechanism $\mathcal{M}: \mathcal{G} \to \mathcal{K}$ satisfies $\varepsilon$-DP if, for any pair of neighboring graphs $G, G'$ and any measurable subset $\mathcal{S} \subseteq \mathcal{K}$, it holds that
\begin{equation*}
    \Pr[\mathcal{M}(G) \in \mathcal{S}] \leq e^{\varepsilon} \Pr[\mathcal{M}(G') \in \mathcal{S}].
\end{equation*}
\end{definition}

Two graphs $G$ and $G'$ are said to be \textit{node-neighboring} if they differ by exactly one node and all of its incident edges. 
Similarly, they are \textit{edge-neighboring} if they differ by exactly one edge. 
Applying these notions to the DP definition yields the concepts of \textit{node-differential privacy} (node-DP) and \textit{edge-differential privacy} (edge-DP)~\cite{hay2009accurate}, respectively. 

DP satisfies several desirable properties, as outlined below.

\begin{lemma}[Post-Processing]
\label{lem:pp}
Let $\mathcal{M}: \mathcal{G} \to \mathcal{K}$ be a randomized mechanism that satisfies $\varepsilon$-DP. 
Then for any (possibly randomized) function $\mathcal{F}: \mathcal{K} \to \mathcal{Z}$, the mechanism $\mathcal{F} \circ \mathcal{M}: \mathcal{G} \to \mathcal{Z}$ also satisfies $\varepsilon$-DP. 
\end{lemma}

\begin{lemma}[Basic Composition] 
\label{lem:bc}
Let $\mathcal{M}_1, \mathcal{M}_2, \dots, \mathcal{M}_n$ be a series of randomized mechanisms, 
where for each $i \in \{1, 2, \dots, n\}$,
$\mathcal{M}_i: \mathcal{G} \to \mathcal{K}_i$ satisfies $\varepsilon$-DP. 
The composition mechanism
$
\mathcal{M}(G) = (\mathcal{M}_1(G), \dots, \mathcal{M}_n(G))
$
satisfies $n\varepsilon$-DP.
\end{lemma}

\begin{lemma}[Parallel Composition]
\label{lem:pc}
Let $\mathcal{M}_1, \mathcal{M}_2, \ldots, \mathcal{M}_n$ be mechanisms, each satisfying $\varepsilon$-DP. 
If for any pair of neighboring graphs $G$ and $G'$, there exists at most one mechanism $\mathcal{M}_i$ such that $\mathcal{M}_i(G) \neq \mathcal{M}_i(G')$, then the composition mechanism
$
\mathcal{M}(G) = (\mathcal{M}_1(G), \dots, \mathcal{M}_n(G))
$
satisfies $\varepsilon$-DP.
\end{lemma}

A common DP mechanism is the \textit{Laplace mechanism}. 
Here, for any arbitrary graph query $Q: \mathcal{G} \to \mathcal{R}$, its \textit{global sensitivity} is defined as
\begin{equation*}
    \mathrm{GS}_Q := \max_{G, G'} \| Q(G) - Q(G') \|_1,
\end{equation*}
where $G$ and $G'$ are any pair of neighboring graphs. Then, the Laplace mechanism is defined as follows:
\begin{definition}[Laplace Mechanism~\cite{dwork2006calibrating}]
\label{def:lap_mech}
Given a graph query $Q: \mathcal{G} \to \mathcal{R}$, the mechanism
\begin{equation*}
    \mathcal{M}(G) = Q(G) + \operatorname{Lap}\left(\frac{1}{\varepsilon}\mathrm{GS}_Q\right)
\end{equation*}
satisfies $\varepsilon$-DP, where $\operatorname{Lap}(\mathrm{GS}_Q/\varepsilon)$ denotes a Laplace random variable with mean $0$ and scale $\mathrm{GS}_Q/\varepsilon$.
\end{definition}
Specifically, computing global sensitivity under node-neighboring or edge-neighboring graphs leads to node-DP or edge-DP guarantees, respectively.

\subsection{Local Differential Privacy}
\label{subsec:ldp}

In contrast to the central-DP setting, where data analytics is performed centrally by a trusted curator, 
LDP studies a distributed setting,
in which each party corresponds to a node with access only to its own information. 
When performing data analytics tasks, each node first privatizes its data locally, then sends the noisy output to an analyzer for post-processing.
To guarantee LDP, we must ensure that the collection of messages received by the analyzer satisfies DP, as defined in Definition~\ref{def:dp}.
More precisely, each node $v_i$ has access to a subgraph $G_i = \left(\{v_i\} \cup \{v_j \mid j \in \mathcal{N}(i)\}, \mathcal{E}(i)\right)$.
To perform graph analytics,
each $v_i$ invokes 
a local randomizer $\mathcal{R}$ on $G_i$, which produces a privatized output $\mathcal{X}_i = \mathcal{R}(G_i)$.
The analyzer $\mathcal{A}$ then aggregates the local outputs $(\mathcal{X}_1, \mathcal{X}_2, \dots, \mathcal{X}_N)$ to produce the final result.
A mechanism is said to satisfy $\varepsilon$-LDP if the joint output $(\mathcal{X}_1, \mathcal{X}_2, \dots, \mathcal{X}_N)$ satisfies $\varepsilon$-DP.
Similar to the central-DP setting, applying node- or edge-neighboring definitions yields notions of \textit{node-LDP} and \textit{edge-LDP}, respectively.
To date, all existing LDP graph analytics literature~\cite{sun2019analyzing,ye2020towards, ye2020lf,Imola2021, Imola2022,eden2023triangle, he2025robust} has focused on the edge-LDP setting, which is also the focus of our work.

A commonly used mechanism under LDP is RR, in which each node $v_i \in V$ outputs a randomized neighbor list. 
The analyzer $\mathcal{A}$ can then construct a noisy graph based on the received neighboring lists and perform graph analytics on it.
Since each edge is shared between two incident nodes, to avoid duplication, the information of edge $(v_i,v_j)$ is reported only by node $v_i$ when $i < j$.
Specifically, for each node $v_i$, RR takes as input a binary neighbor list $\textbf{a}_i = (a_{i,1}, \dots, a_{i,i-1}) \in \{0, 1\}^{i-1}$, where $a_{i,j}$ is $1$ if $j \in \mathcal{N}(i)$ and $0$ otherwise. 
For each entry $a_{i,j}$, RR keeps the true bit with probability $p ={e^\varepsilon}/({e^\varepsilon + 1})$ and flips it with probability $q = {1}/({e^\varepsilon + 1})$. 
This mechanism satisfies $\varepsilon$-edge-LDP.

\paragraph{Multi-round LDP mechanism}

The above basic LDP model is sufficient for simple queries such as edge counting and star counting~\cite{Imola2021}, where the final result can be obtained in a single round. 
However, more complex graph analytics tasks like triangle counting~\cite{Imola2021, Imola2022,he2025robust} often require collaboration between nodes and the analyzer across multiple rounds of interaction.
Here, we consider a $k$-round LDP mechanism.
In round $\ell \in \{1, \dots, k\}$, each node $v_i$ invokes a round-specific randomizer $\mathcal{R}^{(\ell)}$ on its local subgraph $G_i$, possibly incorporating messages released by the other parties in the previous $\ell-1$ rounds.
Let $\mathcal{X}^{(\ell)}_i$ denote the privatized output of node $v_i$ in round $\ell$, and let the collection of all node messages in round $\ell$ be
$\mathcal{X}^{(\ell)} = 
    (\mathcal{X}^{(\ell)}_1, \ldots, \mathcal{X}^{(\ell)}_N)$.
The mechanism satisfies $\varepsilon$-LDP if the joint messages from all rounds, i.e., $(\mathcal{X}^{(1)}, \dots, \mathcal{X}^{(k)})$, satisfy $\varepsilon$-DP.
To ensure this, we can allocate the privacy budget across rounds such that each $\mathcal{X}^{(\ell)}$ satisfies $\varepsilon / k$-DP. 
By applying basic composition (Lemma~\ref{lem:bc}), the entire $k$-round process satisfies $\varepsilon$-DP as required.
To quantify the communication cost incurred across rounds, we sum the sizes (in bits) of all messages transmitted. 
For computational cost, we report the total cost incurred across all rounds by the analyzer and by a single node.

\subsection{Useful Concentration Bounds}
The following three lemmas provide high-probability concentration bounds that will be used in the subsequent analysis.
 
\begin{lemma}[Concentration Bound of Laplace Distributions~\cite{chan11continual}]
\label{lem:con_lap}
Let $\eta_1, \eta_2, \dots, \eta_k$ be independent random variables with $\eta_i \sim \mathrm{Lap}(b_i)$. Then, for any $\beta > 0$,
\begin{equation*}
\label{eq:lap_bound}
\Pr\left[\bigg| \sum_i \eta_i \bigg| \geq \sqrt{8 \sum_i b_i^2 \cdot \log\frac{2}{\beta}}\right] \leq \beta.
\end{equation*}

\end{lemma}

\begin{lemma}[Multiplicative Chernoff Bound~\cite{doerr2019probabilistic}]
\label{lem:con_cher}
Let \(X_1,\dots,X_n\) be independent random variables taking values in \([0,1]\).
Let \(X=\sum_{i=1}^n X_i\).
For any \(0 < \delta \le 1\),
\begin{equation*}
    \Pr[X \ge (1+\delta)\,\mathbb{E}[X]] 
  \leq \exp(-\frac{1}{3}\delta^{2}\,\mathbb{E}[X]).
\end{equation*}
    
\end{lemma}

\begin{lemma}[Chebyshev’s inequality~\cite{chebyshev1867valeurs}]
\label{lem:che_neq}
Let $X$ be a random variable with $\mathbb{E}[X]=\mu$ and $\operatorname{Var}(X)=\sigma^{2}$.  
For any $k>0$,
\begin{equation*}
  \Pr[|X-\mu|\ge k\sigma]\le\frac{1}{k^{2}} .
\end{equation*}
\end{lemma}

\section{LDP Walk Counting}
\label{sec:k-wk}

We begin with a special case of acyclic pattern counting: $k$-line path counting. 
As a preliminary step, we develop a mechanism for the related task of $k$-line walk counting.
As noted in Section~\ref{subsec:prelim_pattern_counting}, walk counting can be viewed as a simpler variant of path counting.
Although walk counting does not fall under acyclic graph pattern counting, it forms the foundation for our mechanisms for both path counting and general acyclic graph pattern counting.
{In this section, we first provide a brief overview of the existing one-round RR-based solution for walk counting, which incurs an error of $O(N^k)$.}
We then introduce our multi-round LDP approach, which achieves an error bound of $\tilde{O}(Nd(G)^{k-1})$ with constant probability, while improving communication and computation costs by an order of magnitude.

\subsection{{Overview of the RR-Based Method}}
\label{subsec:strawman}


Similar to many pattern counting tasks under LDP, such as triangle counting, the key challenge in answering $k$-line walk queries lies in capturing correlations across multiple edges. 
{\cite{suppakitpaisarn2025counting} adopts a straightforward one-round approach that first constructs a noisy graph $\tilde{G}$ using RR, as described in Section~\ref{subsec:ldp}.
Given $\tilde{G}$, for each node pair $(v_i, v_j)$, the analyzer $\mathcal{A}$ then computes an unbiased estimator $\hat{a}_{i,j}$ for the binary existence of edge $(v_i, v_j)$ in the original graph $G$.
Since the existence of any walk can be unbiasedly estimated by the product of the estimators of its constituent edges, $\mathcal{A}$ enumerates all $(k+1)$-node sequences and sums the corresponding walk estimators to compute the total count.
This method guarantees an unbiased estimate with an additive error of $O(N^k)$ for $k$-line walk counting, and it extends naturally to general $k$-edge acyclic pattern counting with the same asymptotic utility guarantee.
}

\paragraph{Redundant count removal}
The above solution treats walk orientations distinctly; that is, the sequences $(v_1, \dots, v_{k+1})$ and $(v_{k+1}, \dots, v_1)$ are considered different walks. 
{To eliminate redundant counts due to orientation, ~\cite{suppakitpaisarn2025counting} divides the total count by the number of automorphisms~\cite{godsil2013algebraic} of the queried pattern.}
{Due to space constraints, we defer detailed discussions on redundant count handling across pattern counting queries to Appendix~{A.1} of the full version~\cite{full_version}.}

\paragraph{Communication and computation analysis}
The total communication cost is $O(N^2)$, as each node $v_i$ sends a message of size $O(N)$ to construct the noisy graph. 
For computation, the analyzer $\mathcal{A}$ requires $O(N^{k+1})$ time to enumerate
all possible sequences,
while each node incurs $O(N)$ time to perturb its neighbor list.

\subsection{Our Approach: Multi-Round Aggregation}
\label{subsec:k-wk}

The problem with the strawman solution is that we publish ``too much'' information 
of nodes
in an overly ``fine-grained'' manner. 
Fundamentally, the more information is released with the same privacy level, the more difficult it is to achieve higher utility based on the released data.
Therefore, we would like to collect the data in a more ``local'' and more ``aggregated'' manner. 
Given edge counting (i.e., $1$-line walk counting) as an example, instead of letting each node publish its perturbed neighbor list, we release only the degree information.
The analyzer then aggregates the perturbed degrees to produce the query result.
With this approach, the error can be reduced from typically $\Omega(N)$
to $O(\sqrt{N})$, by simply adding Laplace noises of scale $2/\varepsilon$ to each node's degree count.

\begin{algorithm}
\small
\caption{Local randomizer $\mathcal{R}^{(\ell)}$ of $\mathcal{M}_{\mathrm{k\text{-}wk}}$}
\label{alg:k-wk_randomizer}
\SetNoFillComment

\tcc{Initialization} 
\ForEach{$i \in \{1, \dots, N\}$}{
    $\mathcal{X}^{(0)}_i \leftarrow 1$
}

\SetKwFunction{r}{$\mathcal{R}^{(\ell)}$}
\SetKwProg{Fn}{Function}{:}{}
\Fn{\r{$G_i$, $\varepsilon$}}{

Received $\mathcal{X}^{(\ell -1)}_i$ for each $i \in \{1, \dots, N\}$

    $ \|\mathcal{X}^{(\ell-1)}\|_\infty := \max_i|\mathcal{X}_i^{(\ell-1)}|$ \\
    $\mathcal{X}^{(\ell)}_i \leftarrow 
    \sum_{j \in \mathcal{N}(i)}\mathcal{X}^{(\ell-1)}_j +
    \operatorname{Lap}(\frac{2k}{\varepsilon}\|\mathcal{X}^{(\ell-1)}\|_\infty)$\\
    
    \textbf{Publish} $\mathcal{X}^{(\ell)}_i$
   
}
\end{algorithm}

This idea, {also explored in~\cite{betzer2024publishing}}, can be extended to support arbitrary $k$-line walk counting queries with multiple rounds of communication: 
At the beginning, we initialize $\mathcal{X}^{(0)}_i =1$ for any $v_i\in V$, representing that the number of $0$-line walk ending at each node is $1$.
Then, in round $\ell \in \{1, \dots, k\}$, each node collects the number of $(\ell-1)$-line walks ending at its neighbors.
By aggregating these values, each node can compute the number of $\ell$-line walks that end at itself.
To ensure the reported counts satisfy edge-LDP, each node must add noise to obscure the presence or absence of edges between its neighbors.
Specifically, {unlike~\cite{betzer2024publishing} that limits the global sensitivity using a clipping mechanism,}
our approach adds Laplace noise with a scale proportional to the maximum count from the previous round.
\uline{Notably, for each node, the noise scale is determined by the maximum count across all nodes rather than just its own neighborhood, thereby preventing leakage of local structural information.}
Besides, since each edge contributes to the computation in every round, the total privacy budget is divided across the $k$ rounds, allocating $\varepsilon/k$ to each.
Finally, the counts produced in round $k$ are sent to the analyzer $\mathcal{A}$, which aggregates all node reports and returns 
\begin{equation*}
\label{e:k_owk_output}
     \mathcal{M}_{\mathrm{k\text{-}wk}}(G,\varepsilon) =  \sum_{i=1}^N \mathcal{X}^{(k)}_i
\end{equation*}
as the output.
The detailed local randomizer $\mathcal{R}^{(\ell)}$ applied at round~$\ell$ of this $k$-round local-DP mechanism $\mathcal{M}_{\mathrm{k\text{-}wk}}$ is described in Algorithm~\ref{alg:k-wk_randomizer}.

Figure~\ref{fig:2-wk} illustrates the execution of the $2$-round mechanism $\mathcal{M}_{\mathrm{2\text{-}wk}}$ on graph $G$ with a total privacy budget of $\varepsilon = 4$.
Each node is assigned a per-round privacy budget of $\varepsilon / 2k = 1$.
In both rounds $1$ and $2$, each node first aggregates the outputs from its neighbors in the previous round (shown in blue), and then adds Laplace noise to privatize the result (shown in red).
The final output of the mechanism, $5.7$, is computed by summing the privatized values of all nodes in round $2$.

\begin{figure}[h]
  \centering
  \includegraphics[width=0.58\linewidth]{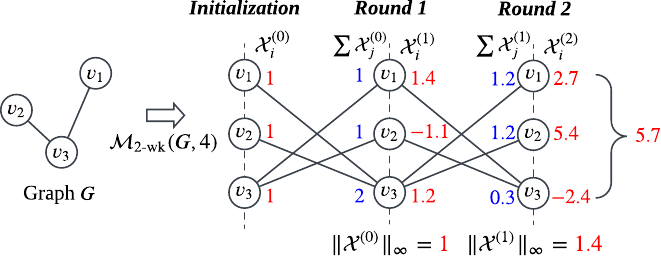} 
\caption{Illustration of the 2-line walk counting mechanism $\mathcal{M}_{\mathrm{2\text{-}wk}}$ applied to graph $G$ with privacy budget $\varepsilon = 4$.
For simplicity, $\sum$ denotes $\sum_{j \in \mathcal{N}(i)}$.}
\label{fig:2-wk}
\end{figure}

For privacy and utility, we establish the following guarantees:
\begin{theorem}
\label{the:k-owk-privacy}
For any graph $G$, walk length $k$, and privacy budget $\varepsilon$, $\mathcal{M}_{\mathrm{k\text{-}wk}}$ satisfies $\varepsilon$-edge-LDP.
\end{theorem}

\begin{theorem}
\label{the:o_k-wk_utility}
For any graph $G$, walk length $k$, and parameters $\varepsilon$ and $\beta$, the $k$-round mechanism $\mathcal{M}_{\mathrm{k\text{-}wk}}$ satisfies
$\mathbb{E}[\mathcal{M}_{\mathrm{k\text{-}wk}}(G,\varepsilon)] = {W}_k$,
and with probability at least $1 - \beta$,
\begin{equation*}
|\mathcal{M}_{\mathrm{k\text{-}wk}}(G,\varepsilon) - {W}_k|
  \leq  k\gamma \sqrt{N} (d(G) + \gamma)^{k-1} = \tilde{O}(\sqrt{N}d(G)^{k-1}),
\end{equation*}
where $\gamma = 2k\sqrt{8\log(2kN/\beta)}/\varepsilon = \tilde{O}(1)$.  
\end{theorem}

{Proofs of all lemmas and theorems in this paper are provided in Appendix~A.3 of the full version~\cite{full_version}.}
By Theorem~\ref{the:o_k-wk_utility}, $\mathcal{M}_\mathrm{k\text{-}wk}$ successfully reduces the error from typically {$O(N^k)$} to $\tilde{O}(\sqrt{N}d(G)^{k-1})$ for general $k$-line walk counting queries.

\paragraph{Communication and computation analysis}
The total communication cost of the $k$-round mechanism $\mathcal{M}_{\mathrm{k\text{-}wk}}$ is $O(N^2)$. 
Specifically, in the first $k-1$ rounds, each node $v_i$ broadcasts its count of $O(1)$ size to all other nodes, incurring a total cost of $O(N^2)$, considering that $k$ is a constant. 
In the final round, each node sends its count to the analyzer, contributing an additional $O(N)$ cost.
For the computational cost, the analyzer $\mathcal{A}$ runs in $O(N)$ time to aggregate the values received from all $N$ nodes in round $k$. 
Each node performs $O(N)$ computation across all $k$ rounds: in each round, it spends $O(N)$ time computing a maximum and another $O(N)$ time aggregating values from neighbors.

\subsection{Optimization in Efficiency}
\label{subsec:opt_wk}

So far, $\mathcal{M}_{\mathrm{k\text{-}wk}}$ has achieved improved computational complexity with state-of-the-art utility. 
Next, we further optimize both communication and computation.

\paragraph{Round reduction:}

The number of rounds in $\mathcal{M}_{\mathrm{k\text{-}wk}}$ can be reduced from $k$ to $k-1$. 
The key observation is that, for each $k$-line walk, the final edge is shared between the second-to-last node and the last node. 
Therefore, once the second-to-last node obtains the noisy count of $(k-1)$-line walks ending at itself, it can compute the number of $k$-line walks where it serves as the second-to-last node, eliminating the need for round $k$.
More precisely, in round $k-1$, each node $v_i$ computes $\mathcal{X}^{(k-1)}_i$ as in Algorithm~\ref{alg:k-wk_randomizer} and multiplies it by $d(v_i) + \mathrm{Lap}(2k/\varepsilon)$. 
The analyzer $\mathcal{A}$ then aggregates these values to produce the final output. 
We will show that this modification preserves both the privacy and utility guarantees stated in Theorem~\ref{the:k-owk-privacy} and Theorem~\ref{the:o_k-wk_utility}.

\paragraph{Computational cost reduction:}

Another observation is that, in each round, each node scans through all $N$ results from the previous round to compute a maximum, which incurs significant computational redundancy. 
To address this, we can shift the maximum computation to the analyzer. 
More precisely, 
at the end of each round $\ell \in \{1, \dots, k-2\}$, each node $v_i$ additionally sends its value $\mathcal{X}^{(\ell)}_i$ to the analyzer $\mathcal{A}$, which computes the global round maximum and broadcasts it to all nodes. 
This reduces the computation cost at each node from $O(N)$ to $O(d(G))$.

\paragraph{Communication cost reduction:}

The communication cost between nodes can be reduced from $O(N^2)$ to $O(M)$ under the assumption that nodes connected by an edge have direct communication channels. 
In this setting, at the end of each of the first $k-2$ rounds, each node sends its result only to its neighbors, while the global maximum is computed by the analyzer as discussed above. 
This change does not affect the node-level computation, as each node only requires the results from its neighbors and the global maximum. 
Consequently, in each of the first $k-2$ rounds, each edge facilitates bidirectional communication between its endpoints, leading to a total communication cost of $O(M)$ between nodes.

\begin{algorithm}
\small
\caption{Local randomizer $\mathcal{R}^{(\ell)}$ of optimized $\mathcal{M}_{\mathrm{k\text{-}wk}}$}
\label{alg:k-wk-opt_randomizer}
\SetNoFillComment

\ForEach{$i \in \{1, \dots, N\}$}{
    $\mathcal{X}^{(0)}_i \leftarrow 1$
}

\SetKwFunction{r}{$\mathcal{R}^{(\ell)}$}
\SetKwProg{Fn}{Function}{:}{}
\Fn{\r{$G_i$, $\varepsilon$}}{

Received $\|\mathcal{X}^{(\ell-1)}\|_\infty$ from $\mathcal{A}$

   Received $\mathcal{X}^{(\ell-1)}_j$ for each $j \in \mathcal{N}(i)$

    $\mathcal{X}^{(\ell)}_i \leftarrow 
    \sum_{j \in \mathcal{N}(i)}\mathcal{X}^{(\ell-1)}_j +
    \operatorname{Lap}(\frac{2k}{\varepsilon}\|\mathcal{X}^{(\ell-1)}\|_\infty)$\\

  \If{$\ell = k-1$}{
         ${\mathcal{X}}^{(\ell)}_i 
        \leftarrow 
        {\mathcal{X}}^{(\ell)}_i \cdot 
        (d(v_i) + \operatorname{Lap}(\frac{2k}{\varepsilon}))$\\
    }

    \textbf{Send out} $\mathcal{X}^{(\ell)}_i$ to $\mathcal{N}(i)$ and $\mathcal{A}$
   
}
\end{algorithm}

The revised local randomizer is shown in Algorithm~\ref{alg:k-wk-opt_randomizer}.
These refinements collectively reduce the total number of rounds from $k$ to $k-1$, lower the overall communication cost from $O(N^2)$ to $O(M+N)$, and achieve an
improvement from $O(N)$ to $O(d(G))$ in the per-node computational cost.
Below, we show that these refinements will not affect the privacy and utility guarantee:
\begin{theorem}
For any graph $G$, walk length $k$, and privacy budget $\varepsilon$, the optimized mechanism $\mathcal{M}_{\mathrm{k\text{-}wk}}$ satisfies $\varepsilon$-edge-LDP.
\end{theorem}

\begin{theorem}
For any graph $G$, walk length $k$, and parameters $\varepsilon$ and $\beta$, the optimized $(k-1)$-round mechanism $\mathcal{M}_{\mathrm{k\text{-}wk}}$ satisfies
$
     \mathbb{E}[\mathcal{M}_{\mathrm{k\text{-}wk}}(G,\varepsilon)] = {W}_k,
$
and with probability at least $1 - \beta$,
\begin{equation*}
|\mathcal{M}_{\mathrm{k\text{-}wk}}(G,\varepsilon) - {W}_k|
  \leq  k\gamma \sqrt{N} (d(G) + \gamma)^{k-1} = \tilde{O}(\sqrt{N}d(G)^{k-1}),
\end{equation*}
where $\gamma = 2k\sqrt{8\log(2kN/\beta)}/\varepsilon = \tilde{O}(1)$.  
\end{theorem}

\section{LDP Path Counting}
\label{sec:k-lp}

Now, we extend our LDP solution for walk counting to support path counting.  
As discussed in Section~\ref{subsec:prelim_pattern_counting}, path counting imposes an additional constraint compared with walk counting: each node must appear at most once along the path. 
Even in the non-private, centralized setting, this constraint introduces significant challenges.
As shown in the prior work~\cite{jokic2022number}, in theory, walk counting can be solved asymptotically more efficiently than path counting.
Moving towards the LDP setting, the problem becomes even more challenging, as enforcing the uniqueness constraint for a path would inherently require information sharing across multiple nodes, which is fundamentally difficult to achieve under the LDP model.

To address this issue, in this section, we introduce a protocol named random marking, which 
can be seamlessly integrated into our walk counting mechanisms to enable private $k$-line path counts estimation under LDP. 
The resulting mechanism achieves a utility bound of $\tilde{O}(\sqrt{N}d(G)^{k})$, representing only a multiplicative overhead of $\tilde{O}(d(G))$ compared to the corresponding $k$-line walk counting utility of $\tilde{O}(\sqrt{N}d(G)^{k-1})$.

\subsection{Random Marking: Transforming Path Counting into Walk Counting 
}
\label{subsec:k-lp_randommarking}


As highlighted above, detecting node repetition in a path under the LDP setting is inherently difficult due to limited global visibility. 
To address this challenge, {we draw inspiration from color-coding~\cite{alon1995color}, which colors nodes and searches for patterns with distinct node colors}, and adopt a different strategy: rather than detecting repetitions retrospectively, we proactively prevent them by imposing a stricter constraint from the outset.
Specifically, for each node, we require that it not only appear at most once in a walk, but also occupy a predefined position within it. 
{Interestingly, this idea was originally introduced to improve runtime for non-private path and cycle detection, whereas we leverage it here to address the challenges in the LDP setting, where only local views are available.}

To illustrate the high-level idea, we first consider the non-private setting and defer the DP-relevant operations to the next subsection.
We begin with an extra marking round, in which each node $v_i$ samples a random mark $r_i$ uniformly from $\{0, 1, \dots, k\}$ and sends it to its neighbors. 
This mark enforces that the node may appear only as the $(r_i + 1)$-th node in any $k$-line path being counted. 
For example, if a node $v_j$ draws $r_j = 0$, it can only serve as the starting node of a valid $k$-line path. 
We then apply the $(k-1)$-round walk counting mechanism $\mathcal{M}_{\mathrm{k\text{-}wk}}$, with an added constraint: each node $v_i$ participates only in the $r_i$-th round of computation.
At all other rounds $\ell \ne r_i$, the node simply outputs $0$.
More precisely, in each round $\ell \in \{1, \dots, k-1\}$, nodes $v_i \in V$ with $r_i = \ell$ compute the number of $\ell$-line walks ending at themselves with node marks $(0, 1, \dots, r_i)$.
This is achieved by aggregating counts from active neighbors that participated in the previous round.
Additionally, nodes with mark $k - 1$ complete the last line by multiplying their count by the number of their neighbors marked with $k$. 
Notably, since $\mathcal{M}_{\mathrm{k\text{-}wk}}$ runs for only $k-1$ rounds, nodes marked with $0$ or $k$ do not participate in subsequent computation. 
To distinguish this from the DP mechanism, we denote the resulting non-private algorithm as $\widehat{\mathcal{M}}_{\mathrm{k\text{-}pt}}$, where each local client executes the function $\widehat{\mathcal{R}}^{(\ell)}_i$ at round $\ell$. 
The detailed algorithm is provided in Algorithm~\ref{alg:k-lp_nondp}.
Here, we use $\widehat{\mathcal{X}}_i^{(\ell)}$ to denote the non-private output of node $v_i$ in round $\ell$, in distinction from the DP-compliant result ${\mathcal{X}}_i^{(\ell)}$.

\begin{algorithm}
\small
\caption{Local client function $\widehat{\mathcal{R}}^{(\ell)}$ of $\widehat{\mathcal{M}}_{\mathrm{k\text{-}pt}}$}
\label{alg:k-lp_nondp}
\SetKwFunction{rnew}{$\widehat{\mathcal{R}}^{(\ell)}$}
\SetNoFillComment

\ForEach{$i \in \{1, \dots, N\}$}{
    $\widehat{\mathcal{X}}^{(0)}_i \leftarrow 1$
}

\SetKwProg{Fn}{Function}{:}{}
\Fn{\rnew{$G_i$}}{

\If{$\ell = \text{mark}$}{
    \tcc{Randomly sample a mark}
    $r_i \gets \textsf{Unif}(0,k)$

     \textbf{Send out} $r_i$ to $\mathcal{N}(i)$
   
    }

\If{$\ell = r_i$}{

Received $\widehat{\mathcal{X}}^{(\ell-1)}_j$ for each $j \in \mathcal{N}(i)$

        $\widehat{\mathcal{X}}^{(\ell)}_i 
        \leftarrow 
        \sum_{j \in \mathcal{N}(i)} 
        \widehat{\mathcal{X}}^{(\ell-1)}_j \, \mathbf{1}\{r_j = {\ell -1}\}$\\

    \If{$\ell = k-1$}{
        $\widehat{\mathcal{X}}^{(\ell)}_i 
        \leftarrow 
        \widehat{\mathcal{X}}^{(\ell)}_i \cdot \sum_{j \in \mathcal{N}(i)} \mathbf{1}\{r_j = k\}$\\
    }

}
\Else{
     $\widehat{\mathcal{X}}^{(\ell)}_i \leftarrow 0$
}

 \textbf{Send out} $\widehat{\mathcal{X}}^{(\ell)}_i$ to $\mathcal{N}(i)$ and $\mathcal{A}$
}
\end{algorithm}

With this design, it is straightforward to ensure that each node $v_i \in V$ appears at most once in a walk, since its contribution is limited to being the $(r_i+1)$-th node in any walk. 
This enforces the requirement that nodes do not repeat in a path.
However, it also introduces the problem of underestimation. 
Specifically, a $k$-line path in the graph is counted only if all nodes along the path receive marks exactly from $0$ to $k$. 
Since node marks are sampled independently, the probability that a $k$-line path is sampled is $(k+1)^{-(k+1)}$. 
Consequently, at the analyzer $\mathcal{A}$, the result is rescaled by $(k+1)^{(k+1)}$:
\begin{equation*}
     \widehat{\mathcal{M}}_{\mathrm{k\text{-}pt}}(G) = (k+1)^{k+1} \, 
     \sum_{i=1}^N \widehat{\mathcal{X}}^{(k-1)}_i.
\end{equation*}
{Random marking introduces variance in  $\widehat{\mathcal{M}}_{\mathrm{k\text{-}pt}}(G)$ relative to the true $k$-line path count $P_k$. 
For utility analysis, we capture this effect as an additive sampling error, which we bound tightly below. 
Importantly, this sampling error is accounted for in the overall error of our LDP mechanism introduced in Section~\ref{subsec:k-lp_dp}.
}

\begin{lemma}
\label{lem:k-lp_sampling_bound}
For any graph $G$, path length $k$, and parameter $\beta$, 
the $k$-round non-private algorithm $\widehat{\mathcal{M}}_{\mathrm{k\text{-}pt}}$ satisfies $
    \mathbb{E}[\widehat{\mathcal{M}}_{\mathrm{k\text{-}pt}}(G)] = P_k $,
and with probability at least $1 - \beta$,
\begin{equation*}
    |\widehat{\mathcal{M}}_{\mathrm{k\text{-}pt}}(G)-P_k| \leq
    (k+1) \sqrt{N\beta^{-1}}(d(G)+k+1)^k
    = O(\sqrt{N} d(G)^k).
\end{equation*}
\end{lemma}


\subsection{Integrating DP and Efficiency Optimization}
\label{subsec:k-lp_dp}


With Algorithm~\ref{alg:k-lp_nondp}, we successfully reduce the problem of $k$-line path counting to that of $k$-line walk counting. 
We can then utilize the idea of the edge-LDP mechanism $\mathcal{M}_{\mathrm{k\text{-}wk}}$ to privatize node outputs from Algorithm~\ref{alg:k-lp_nondp}.
To begin with, publishing node marks incurs no privacy loss, as the marks are independent of any sensitive data.
Before each round $\ell$, the analyzer collects the outputs from active nodes, computes the global maximum $\|\mathcal{X}^{(\ell-1)}\|_\infty$, and broadcasts it to all nodes.
Then in round $\ell$, after aggregating the result from its neighbors in round $\ell-1$, each active node adds Laplace noise scaled to $\|\mathcal{X}^{(\ell-1)}\|_\infty$.
A key distinction from walk counting is that, under random marking, each edge affects at most one node’s computation in a single round. 
Hence, by parallel composition (Lemma~\ref{lem:pp}), the privacy budget $\varepsilon$ does not need to be divided across nodes or rounds.

\begin{figure*}[htbp!]
  \centering
  \includegraphics[width=0.96\linewidth]{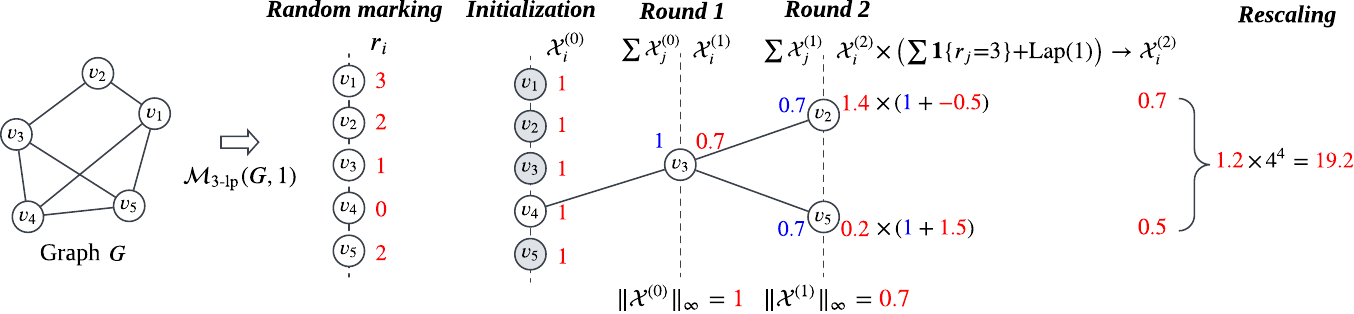} 
  \caption{Illustration of the $3$-line path counting mechanism $\mathcal{M}_{\mathrm{3\text{-}pt}}$ applied to graph $G$ with privacy budget $\varepsilon =1$.
For simplicity, $\sum$ denotes $\sum_{j \in \mathcal{N}(i)}$.
The initialization step occurs before random marking, but is shown afterward for clarity. 
}
\label{fig:3-lp}
\end{figure*}

With the above idea, we establish the $k$-round edge-LDP mechanism $\mathcal{M}_{\mathrm{k\text{-}pt}}$ for $k$-line path counting. 
Compared with the non-private local client $\widehat{\mathcal{R}}^{(\ell)}$ in Algorithm~\ref{alg:k-lp_nondp}, here the local randomizer $\mathcal{R}^{(\ell)}$ adds noise $\mathrm{Lap}(\|\mathcal{X}^{(\ell-1)}\|_\infty/\varepsilon)$ in Line~9 and noise $\mathrm{Lap}(1/\varepsilon)$ to the multiplicative factor 
$\sum_{j \in \mathcal{N}(i)} \mathbf{1}\{r_j = k\}$ in Line~11. 
{When evaluating the privacy of $\mathcal{M}_{\mathrm{k\text{-}pt}}$, astute readers may consider privacy amplification, whereby privacy guarantees can be strengthened when the DP result is derived from an unknown subsample of the input dataset.
However, privacy amplification does not apply here. 
In $\mathcal{M}_{\mathrm{k\text{-}pt}}$, node markings are publicly revealed, making it deterministically known whether an edge will be sampled and thus violating the sampling uncertainty required for privacy amplification.
Moreover, privatizing node markings to enable privacy amplification introduces new challenges, as each mark must be revealed to neighboring nodes during communication.
As a result, we establish the following privacy guarantee:
}
\begin{theorem}
\label{the:k-lp_privacy}
    For any graph $G$, path length $k$, and privacy budget $\varepsilon$, the mechanism $\mathcal{M}_{\mathrm{k\text{-}pt}}$ satisfies $\varepsilon$-edge-LDP.
\end{theorem}
For utility analysis, the total error $|\mathcal{M}_{\mathrm{k\text{-}pt}}(G, \varepsilon)- P_k|$ consists of two parts.
The first is the sampling error incurred by random marking, as shown in Lemma~\ref{lem:k-lp_sampling_bound}.
For the second part, the {DP noise} introduced by the Laplace mechanism, we show that it slightly differs from that of $k$-line walk counting but admits the same asymptotic upper bound.
\begin{lemma}
    \label{lem:k-lp_additive_bound}
    For any graph $G$, path length $k$, parameters $\varepsilon$ and $\beta$, 
$\mathcal{M}_{\mathrm{k\text{-}pt}}$ satisfies
$
    \mathbb{E}[\mathcal{M}_{\mathrm{k\text{-}pt}}(G, \varepsilon)]
    =  
    \widehat{\mathcal{M}}_{\mathrm{k\text{-}pt}}(G)$
under the same node marking $(r_1, \dots, r_N)$, 
and with probability at least $1 - \beta$,
\begin{align*}
    |\mathcal{M}_{\mathrm{k\text{-}pt}}(G, \varepsilon) - 
    \widehat{\mathcal{M}}_{\mathrm{k\text{-}pt}}(G)| 
    &\leq 
    {k(k+1)} \,
    \gamma
    \sqrt{N}
    (\hat{d}(G) + \gamma \sqrt{\hat{d}(G)} + \gamma)^{k-1} 
    \\
    &=  \tilde{O}(\sqrt{N} d(G)^{k-1}),
\end{align*}
where $\gamma = (k+1)\sqrt{8\log(6N/\beta)}/\varepsilon = \tilde{O}(1) $ and $\hat{d}(G) = \max(d(G), \lceil\gamma^2\rceil)$.
\end{lemma}
Combining both parts of the error, we establish the utility bound:
\begin{theorem}
\label{the:k-lp_bound}
    For any graph $G$, path length $k$, parameters $\varepsilon$ and $\beta$, the $k$-round mechanism $\mathcal{M}_{\mathrm{k\text{-}pt}}$ satisfies $
    \mathbb{E}[\mathcal{M}_{\mathrm{k\text{-}pt}}(G, \varepsilon)] =  
    P_k$,
and with probability at least $1 - \beta$,
\begin{align*}
    |\mathcal{M}_{\mathrm{k\text{-}pt}}(G, \varepsilon) - 
    P_k|  \leq  (k+1) \sqrt{2N\beta^{-1}}(d(G)+k+1)^k + \\ 
          {k(k+1)} \, 
        \gamma
        \sqrt{N}
        (\hat{d}(G) + \gamma \sqrt{\hat{d}(G)} + \gamma)^{k-1} 
          =  \tilde{O}(\sqrt{N} d(G)^{k}),
\end{align*}
where $\gamma = (k+1)\sqrt{8\log(12N/\beta)}/\varepsilon = \tilde{O}(1) $ and $\hat{d}(G) = \max(d(G), \lceil\gamma^2\rceil)$.
\end{theorem}

In summary, $\mathcal{M}_{\mathrm{k\text{-}pt}}$ still produces an unbiased estimator of the number of $k$-line paths $P_k$, while incurring a $\tilde{O}(d(G))$ amplification in error upper bound compared to the $k$-line walk counting solution.
Specifically, the overall error is asymptotically dominated by the sampling error, while the {DP noise} matches that of walk counting.

\paragraph{{Remark}}
{
Notably, unbiased estimators are common in LDP settings, especially for RR-based mechanisms, where the final result can be regarded as a sum of independent random variables, and unbiasedness follows directly.
However, our mechanism induces correlations among different noise sources, which significantly complicates proving unbiasedness.}

{
Besides, in our mechanism, the sampling error also arises from the privacy aspect, even though it is independent of the privacy budget $\varepsilon$. 
This is because random marking is not specific to path counting; instead, it addresses an LDP-specific challenge, i.e., eliminating duplicate nodes when each node observes only limited local information.
While our theoretical analysis suggests that the sampling error asymptotically dominates the DP noise by a factor of $\tilde{O}(d(G))$, our experiments show that, in practice, the gap is much smaller. 
{
This observation indicates that applying random marking offers a near ``free lunch'' in our mechanism. 
}
Moreover, we can optimize the overall error by mitigating this gap.
Since both mechanisms $\widehat{\mathcal{M}}_{\mathrm{k\text{-}pt}}$ and ${\mathcal{M}}_{\mathrm{k\text{-}pt}}$ produce unbiased estimators, we run the LDP mechanism
$n_{\text{rep}}$ times with per run privacy budget $\varepsilon/n_{\text{rep}}$ and average the outputs.
This reduces the sampling error, and while allocating a smaller per-run privacy budget increases the DP noise, averaging across runs mitigates its impact. 
We demonstrate the effectiveness of this optimization strategy in
Section~\ref{subsec:sample_dp}.
}


\paragraph{Optimization in efficiency}

Earlier in this section, we required that nodes in rounds inconsistent with their marks report $0$. 
In fact, this step can be omitted, since the computation results depend only on non-zero values. 
As a result, we instead require that nodes remain silent in such rounds.
Moreover, since only nodes with mark $\ell$ are active in round $\ell$, communication can be optimized by having the analyzer $\mathcal{A}$ send $\|\mathcal{X}^{(\ell-1)}\|_\infty$ exclusively to those nodes; their outputs are then forwarded only to neighbors whose mark is $\ell+1$.
This requires the analyzer to know each node’s mark, which can be achieved by modifying Line~6 of Algorithm~\ref{alg:k-lp_nondp} so that nodes also send their marks to $\mathcal{A}$. 

Figure~\ref{fig:3-lp} shows an example of executing the optimized 3-round mechanism $\mathcal{M}_{\mathrm{3\text{-}pt}}$ on graph $G$ with $\varepsilon = 1$, showing only the active nodes at each step.
The mechanism first initializes the number of $0$-line paths ending at each node to $1$.
After that, each node samples its mark and becomes active only in the corresponding round. 
Since initialization corresponds to round $0$,  
only outputs of nodes with mark $0$ are used; outputs from all other nodes are unused and shown in gray.
In the final round, the analyzer aggregates the outputs of all active nodes and rescales the sum to produce the final result $19.2$.

\paragraph{Communication and computation analysis}  
The total communication cost of $\mathcal{M}_{\mathrm{k\text{-}pt}}$ is $O(M+N)$. Communication between nodes contributes $O(M)$, as each node sends its mark and noisy result once to all its neighbors. 
The communication between nodes and the analyzer adds an additional $O(N)$, as each node sends its mark and participates in at most one round of receiving and sending values with the analyzer.
In terms of computational cost, the analyzer runs in $O(N)$ time across $k-2$ rounds of computing the round maximum and one round of aggregating the final result. 
On the node side, the runtime is $O(d(G))$, as each node participates in one round of aggregating $O(d(G))$ values from its neighbors.

\section{LDP Acyclic Pattern Counting}
\label{sec:pattern}

We now move from path counting to the general problem of counting arbitrary acyclic patterns under LDP. 
To avoid ambiguity, from now on, we refer to the nodes in the graph pattern as vertices, while continuing to use nodes to refer to those in $G$.
As discussed earlier, a path inherently imposes an ordering on vertices, progressing naturally from the source to the destination. 
Our path counting mechanisms leverage this property by enumerating paths in a sequential, node-by-node manner. 
In contrast, acyclic patterns lack such directional structure and thus do not admit a natural vertex order, making it difficult to extend our previously proposed mechanism to this more general problem.
To overcome this challenge, we first introduce a method to impose a vertex ordering on arbitrary acyclic patterns. 
Based on this ordering, {we then extend the count aggregation in our LDP path counting mechanism with the idea of recursive subtree counting~\cite{yannakakis1981algorithms} to support general acyclic patterns.}
Our results show that, for any acyclic pattern with $k$ edges, the mechanism still produces an unbiased estimation.

\subsection{Ordering Vertices in Acyclic Patterns}
\label{subsec:preprocess}

Although general acyclic patterns lack a fixed natural orientation, they can be formulated as trees, which admit a conceptual bottom-up orientation, progressing from the leaves to the root.
Based on this idea, given an acyclic pattern, we first formulate it as a tree by specifying a root, identifying the leaves, and defining parent-child relationships among vertices.
Notably, the same pattern may admit multiple valid tree formulations. 
For example, a $k$-line path can be represented as (i) a linear tree with one endpoint as the root and the other as the sole leaf, or (ii) a tree folded at a midpoint with an internal root and two branches terminating at the endpoints.

After transforming the acyclic pattern into a tree $\mathcal{T}$, we define an ordering over the vertices in $\mathcal{T}$.
To distinguish them from nodes in the graph, we use $u$ to denote vertices in the tree. 
The ordering can be determined by a standard depth-first search algorithm, in which each parent vertex is visited after all of its children, and children are visited in lexicographic order.
We denote the ordered vertices in the tree as $(u_0, u_1, \dots, u_k)$, where the subscript indicates the order of each vertex, starting from $0$.
Furthermore, we use $\mathcal{L}$ to denote the set of subscripts of leaf vertices of $\mathcal{T}$. 
For a vertex $u_\ell$, let $\mathcal{T}(u_\ell)$ denote the subtree of $\mathcal{T}$ rooted at $u_\ell$.
\uline{Since there is a correspondence between the tree and the pattern, any subtree $\mathcal{T}(u_\ell)$ is itself an acyclic pattern.}
Besides, to define the parent-child relationships, let $\mathcal{P}(\ell)$ denote the subscript of $u_\ell$'s parent and let $\mathcal{C}(\ell)$ denote the set of subscripts of immediate children of $u_\ell$. 
Notably, $\mathcal{T}$ is not necessarily a binary tree; therefore, a vertex may have more than $2$ children.

\begin{figure}[h]
  \centering
  \includegraphics[width=0.60\linewidth]{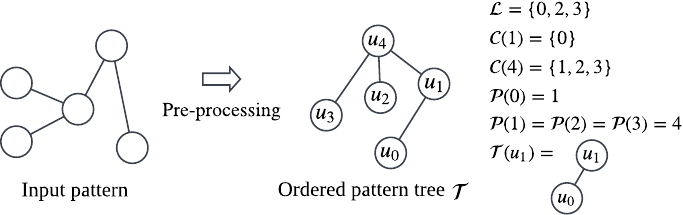} 
  \caption{
  An example of pre-processing a $4$-edge pattern.}
  \label{fig:pattern}
\end{figure}

Our mechanism can operate on any valid tree formulation and vertex ordering, as the asymptotic error bounds, communication cost, and computational cost depend only on $k$ and are independent of the specific tree structure. 
Furthermore, since the tree construction depends only on the pattern and not on the graph $G$, it can be performed entirely as a pre-processing step by the analyzer~$\mathcal{A}$, which then publishes the relevant structural information to all nodes before any local computation.
Figure~\ref{fig:pattern} illustrates an example of this pre-processing step and the resulting structural information.



\subsection{Non-Private Acyclic Pattern Counting}
\label{subsec:non_dp_pattern}

With the vertex ordering of the pattern established, we are now ready to introduce our LDP protocol for acyclic pattern counting. 
Similar to the case of path counting, we begin by studying the non-private solution.
As mentioned earlier, the acyclic pattern is now represented as a tree, and our goal becomes counting the number of occurrences of this tree within the given graph.
To achieve this, we adopt a recursive strategy{~\cite{yannakakis1981algorithms}}: the count for the entire tree is computed by first counting its subtrees.
{To align with our LDP setting, we need to adapt this idea to distributed settings, 
optimize information aggregation, 
calibrate noise for privacy, and analyze error correlations to derive the target bound.}
Specifically, we traverse the vertices of the tree in the previously defined bottom-up order, 
and at each vertex, we aggregate the counting information from the subtrees rooted at its children that are already computed in earlier steps.
In the distributed setting, each node in $G$ is responsible for counting the number of subtrees rooted at itself. 
To do so, it should collect counting information of smaller subtrees rooted at its neighbors.
Building on these steps, we further employ the random marking technique to ensure that no node appears more than once in any single pattern instance.

We now present the details of the non-private algorithm. 
Let $\widehat{\mathcal{X}}_i^{(\ell)}$ denote the output count of node $v_i \in V$ at round~$\ell$.  
Initially, for all $\ell \in \mathcal{L}$, the count for leaf vertex $u_\ell$ is set to $1$ at every node $v_i$, and each node samples a random mark $r_i \in \{0, 1, \dots, k\}$. 
In each round $\ell \in \{0, 1, \dots, k\}$, we focus on counting the subtree $\mathcal{T}(u_\ell)$ at nodes with mark $\ell$.
At this point, the operations performed depend on whether $u_\ell$ is a leaf vertex in $\mathcal{T}$:
\begin{enumerate}
    \item If $u_\ell$ is a leaf vertex, no operation is needed, as each node's count is already initialized to $1$.
    \item If $u_\ell$ is an internal vertex with children, node $v_i$ first collects the subtree counts for each child of $u_\ell$ from its neighbors.
    More precisely, for each $c \in \mathcal{C}(\ell)$, node $v_i$ computes the number of the subtree $\mathcal{T}(c)$ by summing the counts $\widehat{\mathcal{X}}_j^{(c)}$ from its neighbors $v_j$ with mark $c$. 
    We denote this sum as $\widehat{\mathcal{Y}}^{(c)}_i$. 
    After collecting all subtree counts, $v_i$ computes their product to obtain $\widehat{\mathcal{X}}^{(\ell)}_i = \prod_{c \in \mathcal{C}(\ell)} 
        \widehat{\mathcal{Y}}_i^{(c)}$, representing the count of $\mathcal{T}(u_\ell)$ rooted at $v_i$.  
\end{enumerate}
At the end of each round $\ell \notin \mathcal{L}$, nodes with mark $\ell$ send their counts only to neighbors with mark $\mathcal{P}(\ell)$. 
In the final round $k$, active nodes send their counts to the analyzer.
The resulting mechanism runs in $k+2-|\mathcal{L}|$ rounds.
The detailed algorithm for the local client at each node is provided in Algorithm~\ref{alg:pattern_nondp}.

\begin{algorithm}
\small
\caption{Local client $\widehat{\mathcal{R}}^{(\ell)}$ of $\widehat{\mathcal{M}}_\mathcal{T}$}
\label{alg:pattern_nondp}
\SetKwFunction{rnew}{$\widehat{\mathcal{R}}^{(\ell)}$}
\SetNoFillComment

\ForEach{$i \in \{1, \dots, N\}$}{
    \ForEach{$\ell \in \mathcal{L}$}{  
    $\widehat{\mathcal{X}}^{(\ell)}_i \leftarrow 1$
}
}

\SetKwProg{Fn}{Function}{:}{}
\Fn{\rnew{$G_i$, $\mathcal{T}$}}{

\If{$\ell = \text{init}$}{
    $r_i \gets \textsf{Unif}(0,k)$

     \textbf{Send out} $r_i$ to $\mathcal{N}(i)$ and $\mathcal{A}$
    }
  \If{$\ell \in \mathcal{L}$ or $\ell \neq r_i$}{
    \textbf{Skip}
    }  
    \If{$\ell = r_i$}
    {
      
    Received $\widehat{\mathcal{X}}^{(c)}_j$ for each $j \in \mathcal{N}(i)$ and $c \in \mathcal{C}(\ell)$
      
      \ForEach{$c \in \mathcal{C}(\ell)$}{
           $\widehat{\mathcal{Y}}_i^{(c)} \gets \sum_{j \in \mathcal{N}(i)} 
           \widehat{\mathcal{X}}_j^{(c)} \,
           \mathbf{1}\{r_j = c\}
           $
        }

        $\widehat{\mathcal{X}}^{(\ell)}_i \gets \prod_{c \in \mathcal{C}(\ell)} 
        \widehat{\mathcal{Y}}_i^{(c)}$
    }
 \textbf{Send out} $\widehat{\mathcal{X}}^{(\ell)}_i$ to $\mathcal{N}(i)$ with mark $\mathcal{P}(\ell)$ and $\mathcal{A}$
   
}
\end{algorithm}

With the above design, we proactively ensure that each counted acyclic pattern contains no repeated nodes. 
However, similar to the case of path counting, directly summing the node counts at the final round $k$ leads to an underestimation of the true number of pattern instances $Q_\mathcal{T}$.  
Specifically, a subgraph $H \subseteq G$ that matches the pattern $\mathcal{T}$ is counted only if each of its $k+1$ nodes receives a mark corresponding to the order of its mapped vertex in $\mathcal{T}$.  
Thus, the probability that a graph pattern instance is sampled is $(k+1)^{-(k+1)}$.  
Therefore, the analyzer $\mathcal{A}$ rescales the aggregated counts by a factor of $(k+1)^{k+1}$ and outputs:
\begin{equation*}
    \widehat{\mathcal{M}}_\mathcal{T}(G) = (k+1)^{(k+1)} \,
    \sum_{i=1}^N \widehat{\mathcal{X}}_i^{(k)} \, \mathbf{1}\{r_i =k\}.
\end{equation*}

\begin{lemma}
\label{lem:pattern_sampling_error}
    For any graph $G$, $k$-edge pattern $\mathcal{T}$, and parameter $\beta$, 
the $(k + 2 - |\mathcal{L}|)$-round non-private algorithm $\widehat{\mathcal{M}}_{\mathcal{T}}$ satisfies
$
    \mathbb{E}[\widehat{\mathcal{M}}_\mathcal{T}(G)] = Q_\mathcal{T}
$
, and with probability at least $1 - \beta$,
\begin{equation*}
    |  \widehat{\mathcal{M}}_\mathcal{T}(G)- Q_\mathcal{T}| \leq
    (k+1) \sqrt{N\beta^{-1}}(d(G)+k+1)^k
    = O(\sqrt{N} d(G)^k).
\end{equation*}

\end{lemma}

\subsection{Integrating DP and Efficiency Optimization}
\label{subsec:pattern_dp}


We now discuss how to privatize $\widehat{\mathcal{M}}_\mathcal{T}(G)$ by applying the Laplace mechanism to each node's computation.
As previously mentioned, each node sends its result to its neighbors in at most one round, and a neighbor uses the received count only if it operates in a later round. 
Hence, the presence of any edge can affect the computation of at most one node. 
By parallel composition (Lemma~\ref{lem:pc}), the privacy budget $\varepsilon$ can be assigned to each node $v_i$ at its active round without division.
The problem then reduces to ensuring each node's output satisfies $\varepsilon$-edge-DP. 
Since nodes with marks in $\mathcal{L}$ have a fixed output of $1$, we focus only on nodes with marks outside $\mathcal{L}$.
As shown in Algorithm~\ref{alg:pattern_nondp}, an active node $v_i$ at round $\ell$ aggregates outputs from neighbors in each child round $c$ to compute the sum $\widehat{\mathcal{Y}}^{(c)}_i$. 
Similar to path counting, any edge from a node with mark $c$ can influence this sum by at most $\| \mathcal{X}^{(c)} \|_\infty$, a public value computed by the analyzer $\mathcal{A}$.
Since each edge contributes to at most one such sum across all child rounds $c \in \mathcal{C}(\ell)$, the privacy budget $\varepsilon$ need not be divided across them. 
Accordingly, we modify Line~13 of Algorithm~\ref{alg:pattern_nondp} by adding Laplace noise of scale $\| \mathcal{X}^{(c)} \|_\infty / \varepsilon$, resulting in the private local randomizer $\mathcal{R}^{(\ell)}$. 
{As in the path counting case, sampling amplification does not apply here either.}
The resulting edge-LDP mechanism, denoted as $\mathcal{M}_\mathcal{T}$, satisfies the following privacy and utility guarantees:

\begin{theorem}
\label{the:pattern_privacy}
For any graph $G$, $k$-edge pattern $\mathcal{T}$, and privacy budget $\varepsilon$, 
the mechanism $\mathcal{M}_\mathcal{T}$ satisfies $\varepsilon$-edge-LDP.
\end{theorem}
For utility, the total error consists of the sampling error shown in Lemma~\ref{lem:pattern_sampling_error}, and the {DP noise} established below.

\begin{lemma}
    \label{lem:pattern_additive_bound}
    For any graph $G$, $k$-edge pattern $\mathcal{T}$, parameters $\varepsilon$ and $\beta$, 
$\mathcal{M}_{\mathcal{T}}$ satisfies
$
    \mathbb{E}[\mathcal{M}_{\mathcal{T}}(G, \varepsilon)] = 
    \widehat{\mathcal{M}}_{\mathcal{T}}(G)
$ 
under the same node marking $(r_1, \dots, r_N)$, 
and with probability at least $1 - \beta$,
\begin{align*}
    |\mathcal{M}_{\mathcal{T}}(G, \varepsilon) - 
    \widehat{\mathcal{M}}_{\mathcal{T}}(G)| 
    &\leq 
    {k(k+1)} \,
    \gamma
    \sqrt{N}
    (\hat{d}(G) + \gamma \sqrt{\hat{d}(G)} + \gamma)^{k-1} 
    \\
    &=  \tilde{O}(\sqrt{N} d(G)^{k-1}),
\end{align*}
where $\gamma = (k+1)\sqrt{8\log(6kN/\beta)}/\varepsilon = \tilde{O}(1)$ and $\hat{d}(G) = \max(d(G), \lceil \gamma^2 \rceil)$.
\end{lemma}

By combining Lemma~\ref{lem:pattern_sampling_error} and Lemma~\ref{lem:pattern_additive_bound}, we obtain the overall utility guarantee:

\begin{theorem}
\label{the:pattern_bound}
For any graph $G$, $k$-edge pattern $\mathcal{T}$, parameters $\varepsilon$ and $\beta$,
the $(k + 2 - |\mathcal{L}|)$-round mechanism $\mathcal{M}_\mathcal{T}$ satisfies 
$\mathbb{E}[\mathcal{M}_\mathcal{T}(G, \varepsilon)] = Q_\mathcal{T}$,
and with probability at least $1 - \beta$,
\begin{align*}
    |\mathcal{M}_\mathcal{T}(G, \varepsilon) - Q_\mathcal{T}| 
    \leq
    (k+1) \sqrt{2N\beta^{-1}}(d(G)+k+1)^k +\\
      {k(k+1)} \,
    \gamma
    \sqrt{N}
    (\hat{d}(G) + \gamma \sqrt{\hat{d}(G)} + \gamma)^{k-1} 
    = \tilde{O}(\sqrt{N} d(G)^{k}),
\end{align*}
where $\gamma = (k{+}1)\sqrt{8\log(12kN/\beta)}/\varepsilon = \tilde{O}(1) $ and $\hat{d}(G) = \max(d(G), \lceil \gamma^2 \rceil)$.
\end{theorem}

To conclude, $\mathcal{M}_\mathcal{T}$ achieves the same asymptotic utility guarantee as $\mathcal{M}_{\mathrm{k\text{-}pt}}$ while preserving unbiasedness of the output. 
{Consequently, the same LDP error optimization strategy, i.e., averaging the LDP results across multiple runs, also applies to our general acyclic pattern counting mechanism.
}

\paragraph{Communication and computation analysis}

The total communication cost of $\mathcal{M}_\mathcal{T}$ is $O(M +N)$, where $O(M)$ arises from inter-node communication and $O(N)$ from communication between the analyzer and nodes. 
The analyzer requires $O(N)$ computation time, where pattern pre-processing takes constant time. 
Moreover, each node performs $O(d(G))$ local computation.

\paragraph{Special design for $k$-star counting} 

In $k$-star counting, node duplication can be handled locally by each star center, allowing us to skip the random marking process. 
For improved performance, we propose a one-round mechanism $\mathcal{M}_{k\star}$ tailored to this case. 
Specifically, each node $v_i$ first privatizes its degree by adding Laplace noise with scale $2/\varepsilon$, yielding $\tilde{d}(v_i) = d(v_i) + \mathrm{Lap}(2/\varepsilon)$.  
It then computes the product $\tilde{d}(v_i)(\tilde{d}(v_i) - 1)\cdots(\tilde{d}(v_i) - k + 1)$ and sends the result to the analyzer $\mathcal{A}$,  
which aggregates all received values and outputs their sum.  
{To conclude, $\mathcal{M}_{k\star}$ satisfies $\varepsilon$-LDP while achieving the SOTA utility of $\tilde{O}(\sqrt{N}d(G)^{k-1})$~\cite{Imola2021}, with the utility proof deferred to Appendix~A.3 of the full version~\cite{full_version}.
}

\section{Experiment}
\label{sec:exp}

We conduct experiments on four tasks under edge-LDP: 
$k$-line walk counting, 
$k$-line path counting,
$k$-edge acyclic pattern counting, 
and $k$-star counting. 
{For the first three tasks, we compare with the RR-based method from~\cite{suppakitpaisarn2025counting}, denoted as $RR_{\text{enum}}$.
}
{In particular, for $k$-line walk counting, we additionally compare against the method introduced in~\cite{betzer2024publishing}, denoted as $\mathrm{WalkClip}_k$.
}
For $k$-star counting, we compare against $\mathrm{LocalLap}_{k\star}$~\cite{Imola2021}.

\subsection{Setup}
\label{subsec:exp_setup}

\paragraph{Datasets}

\begin{table}[!htbp]
\centering
\small
\begin{tabular}{lcccc}
\toprule
\textbf{Dataset} & \textbf{Nodes} & \textbf{Edges} & \textbf{Max Degree} & \textbf{Ave. Degree} \\
\midrule
\textbf{AstroPh}   & 18,772  & 198,110  & 504   & 21.1 \\
\textbf{Enron}     & 36,692  & 138,831  & 1,383 & 7.6  \\
\textbf{Epinions1} & 75,879  & 405,740  & 3,044 & 10.7 \\
\textbf{Slashdot}  & 77,360  & 546,487  & 2,541 & 14.1 \\
\textbf{Twitter}   & 81,306  & 1,342,310 & 3,383 & 33.0 \\
\textbf{Epinions2} & 131,828 & 711,783  & 3,558 & 10.8 \\
\bottomrule
\end{tabular}
\caption{Summary of datasets used in the experiments.}
\label{tab:dataset}
\end{table}

We evaluate our mechanisms on six real-world graph datasets from SNAP~\cite{leskovec2014snap}, summarized in Table~\ref{tab:dataset}. 
\textbf{AstroPh} is a co-authorship network in astrophysics, 
\textbf{Enron} is an email communication network, 
\textbf{Twitter} and \textbf{Slashdot} are social networks derived from their respective platforms, and 
\textbf{Epinions1} and \textbf{Epinions2} represent a trust network and a social network from Epinions. 
All graphs were preprocessed to be undirected and free of self-loops.

\paragraph{Queries}

For $k$-line walk and path counting tasks, we evaluate cases with $k \in \{4, 5, {6}\}$.
The corresponding walk counting queries are denoted as $Q_{\mathrm{4\text{-}wk}}$, $Q_{\mathrm{5\text{-}wk}}$, {and $Q_{\mathrm{6\text{-}wk}}$},
and the path counting queries as $Q_{\mathrm{4\text{-}pt}}$, $Q_{\mathrm{5\text{-}pt}}$,
{and $Q_{\mathrm{6\text{-}pt}}$}.
For $k$-edge acyclic pattern counting, we consider {three} representative patterns, denoted $\pi_1$, $\pi_2$, {and $\pi_3$}, with edge counts $k = 4$, $k = 5$, {and $k=6$}, respectively. 
The associated queries are denoted by $Q_{\pi_1}$, $Q_{\pi_2}$, {and $Q_{\pi_3}$}.
For the first two patterns, we explore two distinct tree structures as pre-processing results: $\mathcal{T}_1$ and $\mathcal{T}_1'$ for $\pi_1$, and $\mathcal{T}_2$ and $\mathcal{T}_2'$ for $\pi_2$.
{For the third pattern, we consider a single tree structure $\mathcal{T}_3$.}
The structural details of the patterns and their corresponding trees are illustrated in Figure~\ref{fig:tree_form}.
Furthermore, we evaluated {three} cases of $k$-star counting with $k \in \{3, 4, {5}\}$, denoted as queries $Q_{3\star}$, $Q_{4\star}$, {and $Q_{5\star}$}.
All queries are evaluated without redundant counts in the results.

\begin{figure}[h]
  \centering
  \includegraphics[width=0.65\linewidth]{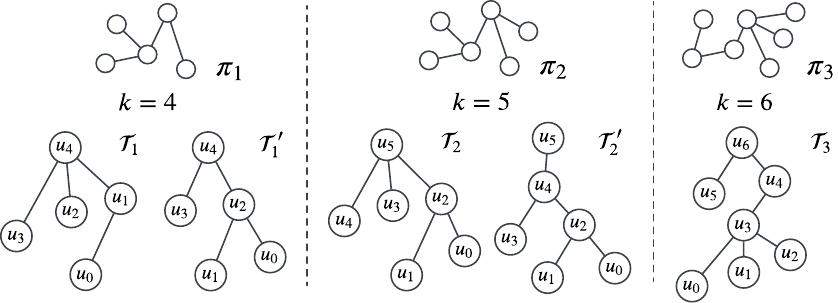} 
  \caption{
  {Selected input patterns and their tree formulations with vertex orderings.}}
  \label{fig:tree_form}
\end{figure}

\paragraph{Experimental parameters}
All experiments are conducted on a Linux server with dual 64-core AMD EPYC 9555 CPUs and 2~TB RAM.
Each experiment is repeated $10$ times.
To evaluate utility, we discard the highest and lowest $2$ results and report the average of the remaining $6$ runs.
Communication cost is reported as the average over all $10$ runs.
We set the default privacy budget to $\varepsilon = 1$.
{All evaluated methods, including our parameter-free mechanisms, take only the privacy budget $\varepsilon$ and the query pattern as inputs, with no additional hyperparameters or tuning. 
For each query, all methods are run under identical settings to ensure fair comparison.}
{As $d(G)$ is assumed to be unknown in the experiments, both $\mathrm{WalkClip}_k$ and $\mathrm{LocalLap}_{k\star}$ require an additional round to estimate the noisy maximum degree $\tilde{d}(G)$. 
For these two methods, we allocate a privacy budget of $\varepsilon / 10$ to the estimation step, and the remaining $9\varepsilon / 10$ to the main counting algorithm. 
Following~\cite{betzer2024publishing}, we set the clipping factor\footnote{{The clipping threshold at round $\ell$ is defined as this factor raised to the power of $\ell$.}} in $\mathrm{WalkClip}_k$ as $\sqrt{\tilde{d}(G)}$.
}

\begin{table}[hbtp!]
\centering
\begin{threeparttable}
\resizebox{0.95\columnwidth}{!}{%
\begin{tabular}{clccccccc}
\toprule
\textbf{Query} & \textbf{Method} & \textbf{Metric} &
\textbf{AstroPh} & \textbf{Enron} & \textbf{Epinions1} &
\textbf{Slashdot} & \textbf{Twitter} & \textbf{Epinions2} \\
\midrule
\multirow[c]{6}{*}[-2pt]{$Q_{\mathrm{4\text{-}wk}}$}
& \multirow[c]{2}{*}{$\mathcal{M}_{\mathrm{4\text{-}wk}}$}
& Rel. err (\%) & \cellcolor{gray!35}1.48 
& \cellcolor{gray!35}1.82
& \cellcolor{gray!35}1.18 
& \cellcolor{gray!35}1.75 
& \cellcolor{gray!35}1.22 
& \cellcolor{gray!35}0.91 \\
& & Comm. & \cellcolor{gray!12}6.76~MB
& \cellcolor{gray!12}7.01~MB 
& \cellcolor{gray!12}15.28~MB 
& \cellcolor{gray!12}17.27~MB 
& \cellcolor{gray!12}44.07~MB 
& \cellcolor{gray!12}26.72~MB \\

\cmidrule(lr){2-9}
& {\multirow[c]{2}{*}{$\mathrm{WalkClip}_4$}}
& {Rel. err (\%)} 
& {98.97} 
& {98.68} 
& {98.70} 
& {97.17} 
& {97.92} 
& {99.34} \\
& & {Comm.} & {10.50~GB} & {40.13~GB} & {171.60~GB} & {178.36~GB} & {197.02~GB} & {515.98~GB} \\

\cmidrule(lr){2-9}
& \multirow[c]{2}{*}{$RR_{\text{enum}}$}
& Rel. err (\%) & 572.86 & 929.78 & 733.06 & 1719.88 & 308.38 & 687.89 \\
& & Comm.$^\dagger$ & 168.00~MB & 641.95~MB & 2745.41~MB & 2853.63~MB & 3152.17~MB & 8255.56~MB \\
\midrule

\multirow[c]{5}{*}[-3pt]{$Q_{\mathrm{5\text{-}wk}}$}
& \multirow[c]{2}{*}{$\mathcal{M}_{\mathrm{5\text{-}wk}}$}
& Rel. err (\%) & \cellcolor{gray!35}2.27 
& \cellcolor{gray!35}2.30 
& \cellcolor{gray!35}1.72 
& \cellcolor{gray!35}2.80 
& \cellcolor{gray!35}1.49 
& \cellcolor{gray!35}0.88 \\
& & Comm. & \cellcolor{gray!12}10.07~MB 
& \cellcolor{gray!12}10.37~MB 
& \cellcolor{gray!12}22.63~MB 
& \cellcolor{gray!12}25.61~MB 
& \cellcolor{gray!12}65.79~MB 
& \cellcolor{gray!12}39.58~MB \\
\cmidrule(lr){2-9}

& {\multirow[c]{2}{*}{$\mathrm{WalkClip}_5$}}
& {Rel. err (\%)} 
& {99.76} 
& {99.59} 
& {99.61} 
& {98.88} 
& {99.35} 
& {99.84} \\
& & {Comm.} & {13.13~GB} & {50.16~GB} & {214.49~GB} & {222.95~GB} & {246.27~GB} & {644.98~GB} \\
\cmidrule(lr){2-9}

& \multirow[c]{1}{*}{$RR_{\text{enum}}$}
& Rel. err (\%) & 1854.25 & 3269.04 & 2346.32 & 8122.89 & 1027.45 & 2208.09 \\
\midrule

\multirow[c]{5}{*}[-4pt]{
{$Q_{\mathrm{6\text{-}wk}}$}}
& \multirow[c]{2}{*}{{$\mathcal{M}_{\mathrm{6\text{-}wk}}$}}
& {Rel. err (\%)} 
& \cellcolor{gray!35} {5.12} 
& \cellcolor{gray!35} {7.15} 
& \cellcolor{gray!35} {2.22} 
& \cellcolor{gray!35} {5.37} 
& \cellcolor{gray!35} {2.54} 
& \cellcolor{gray!35} {1.63} \\
& & {Comm.} 
& \cellcolor{gray!12} {13.38~MB} 
& \cellcolor{gray!12} {13.74~MB} 
& \cellcolor{gray!12} {29.97~MB} 
& \cellcolor{gray!12} {33.95~MB} 
& \cellcolor{gray!12} {87.51~MB} 
& \cellcolor{gray!12} {52.44~MB} \\
\cmidrule(lr){2-9}

& {\multirow[c]{2}{*}{$\mathrm{WalkClip}_6$}}
& {Rel. err (\%)} 
& {99.95} 
& {99.88} 
& {99.89} 
& {99.58} 
& {99.81} 
& {99.96} \\
& & {Comm.} & {15.75~GB} & {60.19~GB} & {257.39~GB} & {267.54~GB} & {295.53~GB} & {773.98~GB} \\

\cmidrule(lr){2-9}
& \multirow[c]{1}{*}{{$RR_{\text{enum}}$}}
& {Rel. err (\%)} 
& {6870.55} & {13021.01} & {8660.78} & {42530.83} & {3771.45} & {8179.33} \\
\midrule

\multirow[c]{3}{*}{$Q_{\mathrm{4\text{-}pt}}$}
& \multirow[c]{2}{*}{${\mathcal{M}}_{\mathrm{4\text{-}pt}}$}
& Rel. err (\%) & \cellcolor{gray!35}5.67
& \cellcolor{gray!35}11.47 
& \cellcolor{gray!35}5.95 
& \cellcolor{gray!35}14.02 
& \cellcolor{gray!35}8.91 
& \cellcolor{gray!35}8.79 \\
& & Comm. & \cellcolor{gray!12}3.55~MB 
& \cellcolor{gray!12}3.59~MB 
& \cellcolor{gray!12}7.84~MB
& \cellcolor{gray!12}8.92~MB 
& \cellcolor{gray!12}23.35~MB 
& \cellcolor{gray!12}13.76~MB \\
\cmidrule(lr){2-9}
& \multirow[c]{1}{*}{$RR_{\text{enum}}$}
& Rel. err (\%) & 477.46 & 791.34 & 600.68 & 1448.47 & 254.75 & 558.30 \\
\midrule

\multirow[c]{3}{*}{$Q_{\mathrm{5\text{-}pt}}$}
& \multirow[c]{2}{*}{${\mathcal{M}}_{\mathrm{5\text{-}pt}}$}
& Rel. err (\%) & \cellcolor{gray!35}11.77 
& \cellcolor{gray!35}5.95 
& \cellcolor{gray!35}11.55 
& \cellcolor{gray!35}10.23 
& \cellcolor{gray!35}4.13 
& \cellcolor{gray!35}9.55 \\
& & Comm. & \cellcolor{gray!12}3.58~MB 
& \cellcolor{gray!12}3.64~MB 
& \cellcolor{gray!12}7.96~MB 
& \cellcolor{gray!12}9.02~MB 
& \cellcolor{gray!12}23.54~MB
& \cellcolor{gray!12}13.92~MB \\
\cmidrule(lr){2-9}
& \multirow[c]{1}{*}{$RR_{\text{enum}}$}
& Rel. err (\%) & 1350.34 & 2400.53 & 1655.97 & 5877.49 & 733.03 & 1538.58 \\
\midrule

\multirow[c]{3}{*}{{$Q_{\mathrm{6\text{-}pt}}$}}
& \multirow[c]{2}{*}{{${\mathcal{M}}_{\mathrm{6\text{-}pt}}$}}
& {Rel. err (\%)} 
& \cellcolor{gray!35} {6.81} 
& \cellcolor{gray!35} {19.29} 
& \cellcolor{gray!35} {15.62} 
& \cellcolor{gray!35} {9.44} 
& \cellcolor{gray!35} {6.70} 
& \cellcolor{gray!35} {6.84} \\
& & {Comm.} 
& \cellcolor{gray!12} {3.60~MB} 
& \cellcolor{gray!12} {3.68~MB} 
& \cellcolor{gray!12} {8.02~MB} 
& \cellcolor{gray!12} {9.09~MB} 
& \cellcolor{gray!12} {23.59~MB}
& \cellcolor{gray!12} {14.03~MB} \\
\cmidrule(lr){2-9}
& \multirow[c]{1}{*}{{$RR_{\text{enum}}$}}
& {Rel. err (\%)} 
& {4115.66} & {7936.65} & {4977.07} & {25456.53} & {2209.38} & {4634.36} \\
\midrule

\multirow[c]{3}{*}{$Q_{\pi_1}$}
& \multirow[c]{2}{*}{$\mathcal{M}_{\mathcal{T}_1}$}
& Rel. err (\%) & \cellcolor{gray!35}12.10 & \cellcolor{gray!35}21.90 & \cellcolor{gray!35}11.71 & \cellcolor{gray!35}27.54 & \cellcolor{gray!35}15.16 & \cellcolor{gray!35}17.38 \\
& & Comm. & \cellcolor{gray!12}3.37~MB & \cellcolor{gray!12}3.37~MB & \cellcolor{gray!12}7.36~MB & \cellcolor{gray!12}8.39~MB & \cellcolor{gray!12}22.29~MB & \cellcolor{gray!12}12.89~MB \\
\cmidrule(lr){2-9}
& \multirow[c]{1}{*}{$RR_{\text{enum}}$}
& Rel. err (\%) & 256.42 & 198.42 & 178.08 & 368.72 & 65.25 & 166.18 \\
\midrule

\multirow[c]{3}{*}{$Q_{\pi_2}$}
& \multirow[c]{2}{*}{$\mathcal{M}_{\mathcal{T}_2}$}
& Rel. err (\%) & \cellcolor{gray!35}11.53 & \cellcolor{gray!35}23.27 & \cellcolor{gray!35}37.30 & \cellcolor{gray!35}48.14 & \cellcolor{gray!35}18.64 & \cellcolor{gray!35}17.24 \\
& & Comm. & \cellcolor{gray!12}3.32~MB & \cellcolor{gray!12}3.30~MB & \cellcolor{gray!12}7.23~MB & \cellcolor{gray!12}8.25~MB & \cellcolor{gray!12}21.96~MB & \cellcolor{gray!12}12.67~MB \\
\cmidrule(lr){2-9}
& \multirow[c]{1}{*}{$RR_{\text{enum}}$}
& Rel. err (\%) & 448.55 & 200.03 & 209.82 & 619.40 & 70.76 & 178.98 \\

\midrule 
\multirow[c]{3}{*}{{$Q_{\pi_3}$}}
& \multirow[c]{2}{*}{{$\mathcal{M}_{\mathcal{T}_3}$}}
& {Rel. err (\%)} 
& \cellcolor{gray!35} {18.74} 
& \cellcolor{gray!35} {48.23} 
& \cellcolor{gray!35} {23.81} 
& \cellcolor{gray!35} {77.08} 
& \cellcolor{gray!35} {28.20} 
& \cellcolor{gray!35} {43.45} \\
& & {Comm.} 
& \cellcolor{gray!12} {3.39~MB} 
& \cellcolor{gray!12} {3.40~MB} 
& \cellcolor{gray!12} {7.44~MB} 
& \cellcolor{gray!12} {8.46~MB} 
& \cellcolor{gray!12} {22.39~MB} 
& \cellcolor{gray!12} {13.00~MB} \\
\cmidrule(lr){2-9}
& \multirow[c]{1}{*}{{$RR_{\text{enum}}$}}
& {Rel. err (\%)} 
& {877.11}
& {273.41} 
& {205.49} 
& {547.55} 
& {54.35} 
& {232.20} \\

\bottomrule
\end{tabular}
}
\end{threeparttable}
\caption{
Relative error (\%) and communication cost (MB{/GB}) for $k$-line walk, $k$-line path, and $k$-edge acyclic pattern counting queries across six datasets.  
$^\dagger$Communication cost of $RR_{\mathrm{enum}}$ remains constant across queries.  
}
\label{tab:3queries}
\end{table}

\begin{table}[hbtp!]
\centering
\resizebox{0.90\columnwidth}{!}{%
\begin{tabular}{clccccccc}
\toprule
\textbf{Query} & \textbf{Method} & \textbf{Metric} & \textbf{AstroPh} & \textbf{Enron} & \textbf{Epinions1} & \textbf{Slashdot} & \textbf{Twitter} & \textbf{Epinions2} \\
\midrule

\multirow[c]{4}{*}{$Q_{3\star}$}
& \multirow{2}{*}{$\mathcal{M}_{3\star}$} 
& Rel. err (\%) 
& \cellcolor{gray!35} 0.17 
& \cellcolor{gray!35} 0.09 
& \cellcolor{gray!35} 0.05 
& \cellcolor{gray!35} 0.06
& \cellcolor{gray!35} 0.04 
& \cellcolor{gray!35} 0.04 \\
& & Comm.$^{\dagger}$ 
& \cellcolor{gray!12} 0.14~MB 
& \cellcolor{gray!12} 0.28~MB 
& \cellcolor{gray!12} 0.58~MB
& \cellcolor{gray!12} 0.59~MB 
& \cellcolor{gray!12} 0.62~MB 
& \cellcolor{gray!12} 1.00~MB \\
\cmidrule(lr){2-9}

& \multirow{2}{*}{$\mathrm{LocalLap}_{3\star}$} & Rel. err (\%) & 6.79 & 8.16 & 19.52 & 16.47 & 5.64 & 10.62 \\
& & Comm.$^{\dagger}$ & 0.43~MB & 0.84~MB & 1.74~MB & 1.77~MB & 1.86~MB & 3.01~MB

\\
\midrule

\multirow{2}{*}{$Q_{4\star}$}
& $\mathcal{M}_{4\star}$ & Rel. err (\%) & 
\cellcolor{gray!35}0.49 & 
\cellcolor{gray!35}0.19 & 
\cellcolor{gray!35}0.15 & 
\cellcolor{gray!35}0.06 & 
\cellcolor{gray!35}0.09 & 
\cellcolor{gray!35}0.09 \\
\cmidrule(lr){2-9}
& $\mathrm{LocalLap}_{4\star}$ & Rel. err (\%) & 
33.88 & 16.69 & 43.04 & 42.78 & 16.55 & 35.67 \\

\midrule

\multirow{2}{*}{{$Q_{5\star}$}}
& {$\mathcal{M}_{5\star}$} & {Rel. err (\%)} & \cellcolor{gray!35}
{0.50} & \cellcolor{gray!35}{0.22} & \cellcolor{gray!35}{0.18} & \cellcolor{gray!35}{0.14} & \cellcolor{gray!35}{0.11} & \cellcolor{gray!35}{0.09} \\
\cmidrule(lr){2-9}
& {$\mathrm{LocalLap}_{5\star}$} & {Rel. err (\%)} & 
{38.58} & {24.81} & {82.03} & {53.66} & {26.92} & {37.21} \\

\bottomrule
\end{tabular}
}
\caption{Relative error (\%) and communication cost (MB) for $k$-star queries on six datasets. $^{\dagger}$Communication cost of both methods remains constant across star counting queries.}
\label{tab:star}
\end{table}

\subsection{Experiment Results}
\label{subsec:exp_results}

\paragraph{Utility and communication cost}

The experimental results, including relative errors and communication costs for all four tasks, are summarized in Table~\ref{tab:3queries} and Table~\ref{tab:star}. 
Across all {72} query-dataset combinations, our mechanism consistently outperforms the baseline methods in both utility and communication cost.

For $k$-line walk counting queries, our mechanism achieves a relative error of less than {$8\%$} on all datasets, 
with generally better performance for {smaller $k$.}
{In comparison, the $\mathrm{WalkClip}_k$ method consistently yields high relative errors due to significant clipping bias, as the clipping factor leads to the truncation of the majority of walks in large graphs.}
{Regarding communication, results confirm that our optimization delivers substantial savings over the baselines.}
{In summary, across all tested walk counting queries, our mechanisms achieve up to a $110\times$ reduction in error and a $300\times$ reduction in communication cost compared with the baseline methods.
}


For $k$-line path counting queries, the utility of our mechanism shows a slight drop compared to the $k$-line walk queries of the same $k$, due to the involvement of sampling error. 
Nonetheless, it offers lower communication cost and maintains a significant advantage over the baseline: 
{we achieve up to a $2600\times$ error reduction and a $600\times$ communication cost improvement across all queries.}

For $k$-edge acyclic pattern counting queries, our mechanism outperforms the baseline method by up to {$46\times$ in error and $650\times$ in communication cost.
In general, our mechanism achieves higher utility for smaller acyclic patterns.
Notably, the communication cost remains stable across different values of $k$, demonstrating the effectiveness of our optimization strategy.
}

For $k$-star counting queries, our mechanism achieves a $3\times$ reduction in communication cost, as it does not depend on information about the maximum degree. 
This enables each node to complete its interaction with the analyzer in a single round of communication.
In addition to its communication efficiency, our mechanism yields up to {a $710\times$ error reduction,}
while consistently maintaining a relative error below $0.5\%$ across all datasets. 
This demonstrates the robustness and utility of our method.


\paragraph{Privacy budget $\varepsilon$}

To evaluate the impact of the privacy budget $\varepsilon$ on performance, we assess queries $Q_{\mathrm{5\text{-}wk}}$, $Q_{\mathrm{5\text{-}pt}}$, $Q_{\pi_1}$, and $Q_{4\star}$ on the \textbf{AstroPh} dataset, varying $\varepsilon$ uniformly from $0.2$ to {$4.0$} in steps of $0.2$.
Relative errors are shown in Figure~\ref{fig:privacy_budget_analysis}.
Notably, communication cost remains independent of $\varepsilon$.
As illustrated, all methods generally improve in utility as $\varepsilon$ increases, with our mechanism consistently achieving order-of-magnitude lower error across the full range. 
{$\mathrm{WalkClip_5}$, however, maintains a stable error, as its performance is bottlenecked by the clipping bias rather than the DP noise.}
For $Q_{\mathrm{5\text{-}pt}}$ and $Q_{\pi_1}$, the error of our mechanism stabilizes when $\varepsilon \geq 1.0$, as the sampling error, which is independent of $\varepsilon$, becomes dominant over the diminishing DP error.

\begin{figure}[htbp]
    \centering
    \includegraphics[scale=0.54]{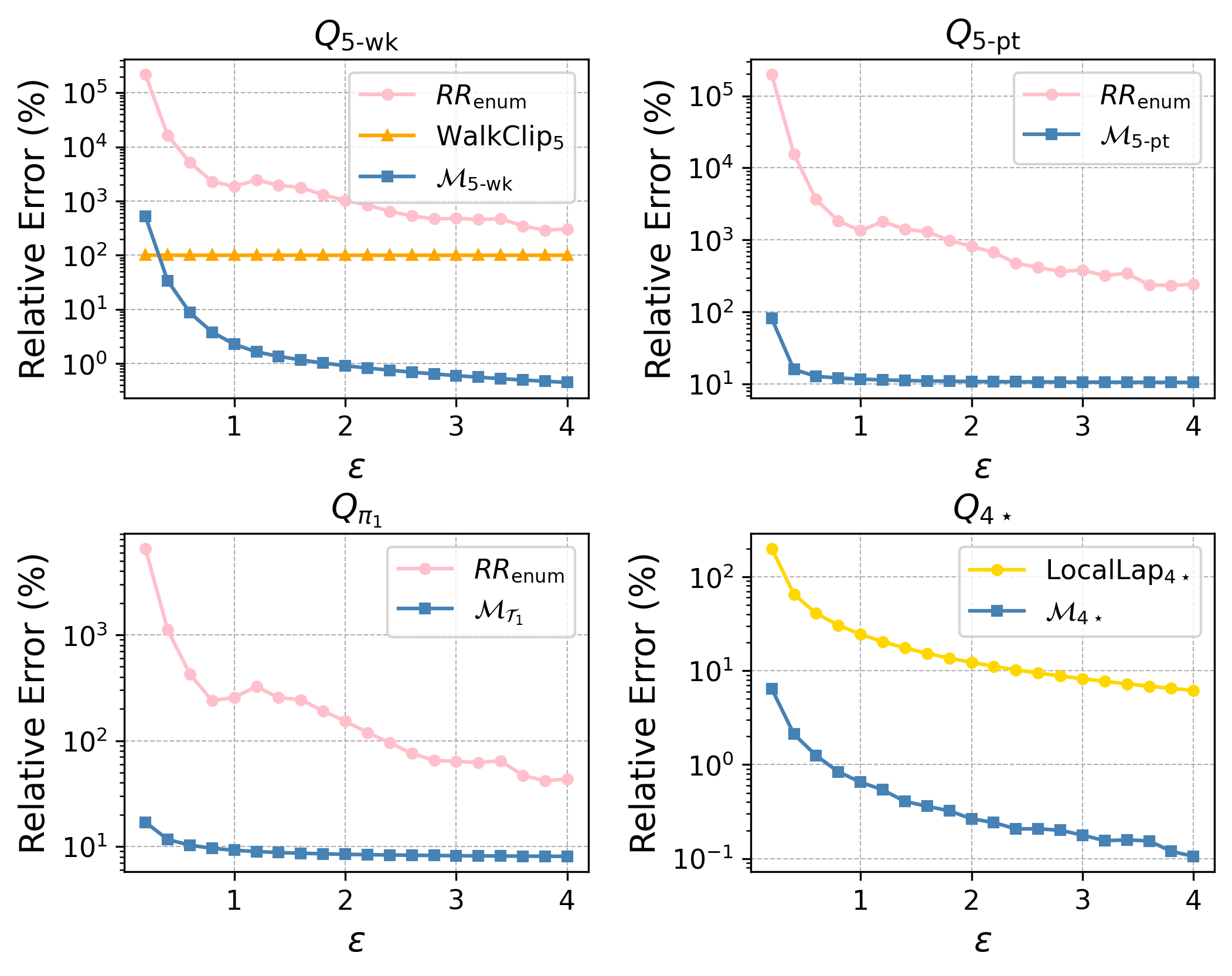}
    \Description{Comparison of relative errors under different epsilon values for four queries.}    
    \caption{
    {Relative error (\%) for $Q_{5\text{-wk}}$, $Q_{5\text{-pt}}$, $Q_{\pi_1}$, and $Q_{4\star}$ on the \textbf{AstroPh} dataset under varying $\varepsilon$ values.}}
    \label{fig:privacy_budget_analysis}
\end{figure}

\paragraph{Tree structure $\mathcal{T}$}
In Section~\ref{sec:pattern}, we establish that the pre-processing tree structure of the input pattern does not affect the asymptotic error bounds or communication cost for $k$-edge acyclic pattern counting queries.
To empirically validate this, we evaluate alternative tree structures $\mathcal{T}_1'$ and $\mathcal{T}_2'$ for patterns $\pi_1$ and $\pi_2$ (Figure~\ref{fig:tree_form}), with resulting mechanisms $\mathcal{M}_{\mathcal{T}_1'}$ and $\mathcal{M}_{\mathcal{T}_2'}$ that execute in $3$ and $4$ rounds.
Evaluation results on the \textbf{Enron}, \textbf{Epinions1}, and \textbf{Twitter} datasets are presented in Table~\ref{tab:pattern_alt_trees}.

\begin{table}[!hbtp]
\centering
\small
\resizebox{0.6\columnwidth}{!}{
\begin{tabular}{c l c c c c}
\toprule
\textbf{Query} & \textbf{Method} & \textbf{Metric} & \textbf{Enron} & \textbf{Epinions1} & \textbf{Twitter} \\
\midrule

\multirow{4}{*}{$Q_{\pi_1}$}
& \multirow{2}{*}{$\mathcal{M}_{\mathcal{T}_1}$}
& \centering Rel. err (\%) & \cellcolor{gray!35}21.90 & \cellcolor{gray!35}11.71 & 15.16 \\
& & \centering Comm. & 3.37~MB & \cellcolor{gray!12}7.36~MB & \cellcolor{gray!12}22.29~MB \\
\cmidrule(lr){2-6}

& \multirow{2}{*}{$\mathcal{M}_{\mathcal{T}_1'}$}
& \centering Rel. err (\%) & 30.66 & 21.40 & \cellcolor{gray!35}14.69 \\
& & \centering Comm. & \cellcolor{gray!12}3.36~MB & \cellcolor{gray!12}7.36~MB & 22.30~MB \\
\midrule

\multirow{4}{*}{$Q_{\pi_2}$}
& \multirow{2}{*}{$\mathcal{M}_{\mathcal{T}_2}$}
& \centering Rel. err (\%) & \cellcolor{gray!35}23.27 & 37.30 & 18.64 \\
& & \centering Comm. & \cellcolor{gray!12}3.32~MB & \cellcolor{gray!12}7.23~MB & \cellcolor{gray!12}21.96~MB \\
\cmidrule(lr){2-6}

& \multirow{2}{*}{$\mathcal{M}_{\mathcal{T}_2'}$}
& \centering Rel. err (\%) & 23.42 & \cellcolor{gray!35}28.70 & \cellcolor{gray!35}18.37 \\
& & \centering Comm. & 3.47~MB & 7.60~MB & 22.71~MB \\
\bottomrule
\end{tabular}
}
\caption{
Relative error (\%) and communication cost (MB) for acyclic pattern queries $Q_{\pi_1}$ and $Q_{\pi_2}$ using different tree structures on three datasets.
}
\label{tab:pattern_alt_trees}
\end{table}

Comparing the two tree structures for each pattern counting query, {while utility differences exist, the variations remain bounded within a constant factor. 
This aligns with our theoretical conclusion that the asymptotic error bound is independent of the tree structure.}
Furthermore, the communication cost remains stable, even for $\pi_2$, where $\mathcal{M}_{\mathcal{T}_2'}$ requires an additional round compared to $\mathcal{M}_{\mathcal{T}_2}$.
This is because the dominant communication overhead stems from publishing the randomized marking results, where all nodes participate, while later-round computations involve only subsets of nodes and contribute comparatively less to the total cost.

\subsection{Sampling Error vs.\ DP {Noise}}
\label{subsec:sample_dp}


In Sections~\ref{sec:k-lp} and~\ref{sec:pattern}, our theoretical analysis demonstrates that, for both $k$-line path and $k$-edge acyclic pattern counting tasks, the sampling error asymptotically dominates the overall error. 
This conclusion is empirically supported by the error decomposition results on all six datasets, presented in Table~6 {in Appendix~{A.2} of the full version~\cite{full_version}.}
{Notably, although our theoretical results imply an asymptotic gap of $\tilde{O}(d(G))$ between the sampling error and the DP noise, this gap is much smaller than $d(G)$ in practice.
}

{Additionally, we empirically evaluate the LDP error optimization strategy in Section~\ref{subsec:k-lp_dp}} on the \textbf{Epinions2} dataset using queries $Q_{\mathrm{4\text{-}pt}}$, $Q_{\mathrm{5\text{-}pt}}$, $Q_{\pi_1}$, and $Q_{\pi_2}$, varying $n_{\text{rep}}$ from $1$ to $10$.
Figure~\ref{fig:soc-epinions-nrep} shows the relative sampling error, DP error, and total error for $Q_{\pi_1}$.
{Results for all four queries are provided in Figure~7 in Appendix~A.2 of the full version~\cite{full_version}.}
As illustrated in Figure~\ref{fig:soc-epinions-nrep}, sampling error decreases while DP error increases with larger $n_{\text{rep}}$, with total error minimized when the two are roughly balanced.  
This strategy yields $2.8\times$, $1.8\times$, $4.6\times$, and $2.0\times$ improvements in total error at the optimal $n_{\text{rep}}$ values for queries $Q_{\mathrm{4\text{-}pt}}$, $Q_{\mathrm{5\text{-}pt}}$, $Q_{\pi_1}$, and $Q_{\pi_2}$, respectively, on this dataset.

\begin{figure}[h]
    \centering
    \includegraphics[width=0.62\textwidth]{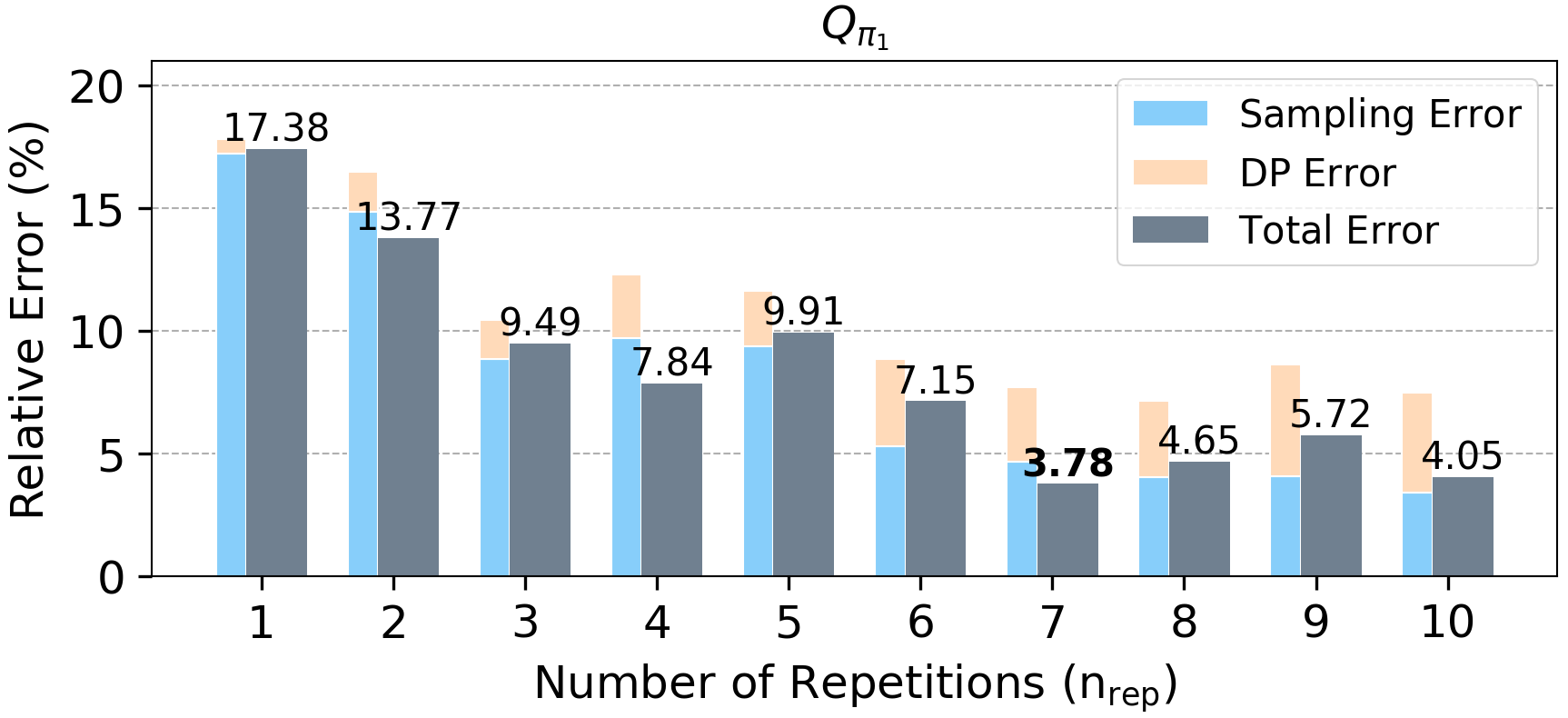}
    \caption{
        Relative error decomposition for $Q_{\pi_1}$ on the \textbf{Epinions2} dataset
        under varying numbers of repetitions ($n_{\text{rep}}$).
    }
    \label{fig:soc-epinions-nrep}
\end{figure}

\section{Future Work}
\label{sec:future_work}

{
A natural direction for future work is to extend our approach to cyclic patterns, which present additional challenges.
Unlike acyclic patterns, cyclic patterns do not allow recursive decomposition, so their counts cannot always be derived by aggregating subpattern counts, making our multi-round aggregation strategy inapplicable. 
A possible workaround is to break cycles and reduce the problem to acyclic patterns with additional node constraints. 
For example, triangle counting can be reduced to counting 2-line paths whose endpoints are connected, but this introduces significant utility challenges, as it requires tracking both endpoints of each counted path while preserving privacy.
}
Another promising extension is to support node-LDP, which offers stronger privacy guarantees than edge-LDP by protecting the presence of each individual node along with all its incident edges.

\section*{Acknowledgments}

This research is supported by the NTU–NAP Startup Grant (024584-00001), the Singapore Ministry of Education Tier 1 Grant (RG19/25), and by the National Research Foundation, Singapore, and the Cyber Security Agency of Singapore under the National Cybersecurity R\&D Programme and the CyberSG R\&D Programme Office (Award CRPO-GC3-NTU-001). Any opinions, findings, conclusions, or recommendations expressed in these materials are those of the author(s) and do not reflect the views of the National Research Foundation, Singapore, the Cyber Security Agency of Singapore, or the CyberSG R\&D Programme Office.
We would also like to thank the anonymous reviewers who have made valuable suggestions to improve this paper.

\newpage
\bibliographystyle{ACM-Reference-Format}
\bibliography{ref}


\begin{thebibliography}{55}


\ifx \showCODEN    \undefined \def \showCODEN     #1{\unskip}     \fi
\ifx \showISBNx    \undefined \def \showISBNx     #1{\unskip}     \fi
\ifx \showISBNxiii \undefined \def \showISBNxiii  #1{\unskip}     \fi
\ifx \showISSN     \undefined \def \showISSN      #1{\unskip}     \fi
\ifx \showLCCN     \undefined \def \showLCCN      #1{\unskip}     \fi
\ifx \shownote     \undefined \def \shownote      #1{#1}          \fi
\ifx \showarticletitle \undefined \def \showarticletitle #1{#1}   \fi
\ifx \showURL      \undefined \def \showURL       {\relax}        \fi
\providecommand\bibfield[2]{#2}
\providecommand\bibinfo[2]{#2}
\providecommand\natexlab[1]{#1}
\providecommand\showeprint[2][]{arXiv:#2}

\bibitem[Abo~Khamis et~al\mbox{.}(2016)]%
        {abo2016faq}
\bibfield{author}{\bibinfo{person}{Mahmoud Abo~Khamis}, \bibinfo{person}{Hung~Q Ngo}, {and} \bibinfo{person}{Atri Rudra}.} \bibinfo{year}{2016}\natexlab{}.
\newblock \showarticletitle{FAQ: questions asked frequently}. In \bibinfo{booktitle}{\emph{Proceedings of the 35th ACM SIGMOD-SIGACT-SIGAI Symposium on Principles of Database Systems}}. \bibinfo{pages}{13--28}.
\newblock


\bibitem[Akoglu and Faloutsos(2013)]%
        {akoglu2013anomaly}
\bibfield{author}{\bibinfo{person}{Leman Akoglu} {and} \bibinfo{person}{Christos Faloutsos}.} \bibinfo{year}{2013}\natexlab{}.
\newblock \showarticletitle{Anomaly, event, and fraud detection in large network datasets}. In \bibinfo{booktitle}{\emph{Proceedings of the sixth ACM international conference on Web search and data mining}}. \bibinfo{pages}{773--774}.
\newblock


\bibitem[Alon et~al\mbox{.}(1995)]%
        {alon1995color}
\bibfield{author}{\bibinfo{person}{Noga Alon}, \bibinfo{person}{Raphael Yuster}, {and} \bibinfo{person}{Uri Zwick}.} \bibinfo{year}{1995}\natexlab{}.
\newblock \showarticletitle{Color-coding}.
\newblock \bibinfo{journal}{\emph{Journal of the ACM (JACM)}} \bibinfo{volume}{42}, \bibinfo{number}{4} (\bibinfo{year}{1995}), \bibinfo{pages}{844--856}.
\newblock


\bibitem[Betzer et~al\mbox{.}(2024)]%
        {betzer2024publishing}
\bibfield{author}{\bibinfo{person}{Louis Betzer}, \bibinfo{person}{Vorapong Suppakitpaisarn}, {and} \bibinfo{person}{Quentin Hillebrand}.} \bibinfo{year}{2024}\natexlab{}.
\newblock \showarticletitle{Publishing number of walks and katz centrality under local differential privacy}. In \bibinfo{booktitle}{\emph{The 40th Conference on Uncertainty in Artificial Intelligence}}.
\newblock


\bibitem[Blocki et~al\mbox{.}(2013)]%
        {blocki2013differentially}
\bibfield{author}{\bibinfo{person}{Jeremiah Blocki}, \bibinfo{person}{Avrim Blum}, \bibinfo{person}{Anupam Datta}, {and} \bibinfo{person}{Or Sheffet}.} \bibinfo{year}{2013}\natexlab{}.
\newblock \showarticletitle{Differentially private data analysis of social networks via restricted sensitivity}. In \bibinfo{booktitle}{\emph{Proceedings of the 4th conference on Innovations in Theoretical Computer Science}}. \bibinfo{pages}{87--96}.
\newblock


\bibitem[Bonawitz et~al\mbox{.}(2017)]%
        {bonawitz2017practical}
\bibfield{author}{\bibinfo{person}{Keith Bonawitz}, \bibinfo{person}{Vladimir Ivanov}, \bibinfo{person}{Ben Kreuter}, \bibinfo{person}{Antonio Marcedone}, \bibinfo{person}{H~Brendan McMahan}, \bibinfo{person}{Sarvar Patel}, \bibinfo{person}{Daniel Ramage}, \bibinfo{person}{Aaron Segal}, {and} \bibinfo{person}{Karn Seth}.} \bibinfo{year}{2017}\natexlab{}.
\newblock \showarticletitle{Practical secure aggregation for privacy-preserving machine learning}. In \bibinfo{booktitle}{\emph{proceedings of the 2017 ACM SIGSAC Conference on Computer and Communications Security}}. \bibinfo{pages}{1175--1191}.
\newblock


\bibitem[Chan et~al\mbox{.}(2011)]%
        {chan11continual}
\bibfield{author}{\bibinfo{person}{T.-H.~Hubert Chan}, \bibinfo{person}{Elaine Shi}, {and} \bibinfo{person}{Dawn Song}.} \bibinfo{year}{2011}\natexlab{}.
\newblock \showarticletitle{Private and Continual Release of Statistics}.
\newblock \bibinfo{journal}{\emph{ACM Transactions on Information and System Security}} (\bibinfo{year}{2011}).
\newblock


\bibitem[Chebyshev(1867)]%
        {chebyshev1867valeurs}
\bibfield{author}{\bibinfo{person}{Pafnutii~Lvovich Chebyshev}.} \bibinfo{year}{1867}\natexlab{}.
\newblock \showarticletitle{Des valeurs moyennes}.
\newblock \bibinfo{journal}{\emph{J. Math. Pures Appl}} \bibinfo{volume}{12}, \bibinfo{number}{2} (\bibinfo{year}{1867}), \bibinfo{pages}{177--184}.
\newblock


\bibitem[Chen and Zhou(2013)]%
        {chen2013recursive}
\bibfield{author}{\bibinfo{person}{Shixi Chen} {and} \bibinfo{person}{Shuigeng Zhou}.} \bibinfo{year}{2013}\natexlab{}.
\newblock \showarticletitle{Recursive mechanism: towards node differential privacy and unrestricted joins}. In \bibinfo{booktitle}{\emph{Proceedings of the 2013 ACM SIGMOD International Conference on Management of Data}}. \bibinfo{pages}{653--664}.
\newblock


\bibitem[Cormode et~al\mbox{.}(2018)]%
        {cormode2018privacy}
\bibfield{author}{\bibinfo{person}{Graham Cormode}, \bibinfo{person}{Somesh Jha}, \bibinfo{person}{Tejas Kulkarni}, \bibinfo{person}{Ninghui Li}, \bibinfo{person}{Divesh Srivastava}, {and} \bibinfo{person}{Tianhao Wang}.} \bibinfo{year}{2018}\natexlab{}.
\newblock \showarticletitle{Privacy at scale: Local differential privacy in practice}. In \bibinfo{booktitle}{\emph{Proceedings of the 2018 international conference on management of data}}. \bibinfo{pages}{1655--1658}.
\newblock


\bibitem[Doerr(2019)]%
        {doerr2019probabilistic}
\bibfield{author}{\bibinfo{person}{Benjamin Doerr}.} \bibinfo{year}{2019}\natexlab{}.
\newblock \showarticletitle{Probabilistic tools for the analysis of randomized optimization heuristics}.
\newblock In \bibinfo{booktitle}{\emph{Theory of evolutionary computation: Recent developments in discrete optimization}}. \bibinfo{publisher}{Springer}, \bibinfo{pages}{1--87}.
\newblock


\bibitem[Dong et~al\mbox{.}(2024a)]%
        {dong2024continual}
\bibfield{author}{\bibinfo{person}{Wei Dong}, \bibinfo{person}{Zijun Chen}, \bibinfo{person}{Qiyao Luo}, \bibinfo{person}{Elaine Shi}, {and} \bibinfo{person}{Ke Yi}.} \bibinfo{year}{2024}\natexlab{a}.
\newblock \showarticletitle{Continual observation of joins under differential privacy}.
\newblock \bibinfo{journal}{\emph{Proceedings of the ACM on Management of Data}} \bibinfo{volume}{2}, \bibinfo{number}{3} (\bibinfo{year}{2024}), \bibinfo{pages}{1--27}.
\newblock


\bibitem[Dong et~al\mbox{.}(2022)]%
        {dong2022r2t}
\bibfield{author}{\bibinfo{person}{Wei Dong}, \bibinfo{person}{Juanru Fang}, \bibinfo{person}{Ke Yi}, \bibinfo{person}{Yuchao Tao}, {and} \bibinfo{person}{Ashwin Machanavajjhala}.} \bibinfo{year}{2022}\natexlab{}.
\newblock \showarticletitle{R2T: Instance-optimal Truncation for Differentially Private Query Evaluation with Foreign Keys}. In \bibinfo{booktitle}{\emph{Proc. ACM SIGMOD International Conference on Management of Data}}.
\newblock


\bibitem[Dong et~al\mbox{.}(2024b)]%
        {dong2024instance}
\bibfield{author}{\bibinfo{person}{Wei Dong}, \bibinfo{person}{Juanru Fang}, \bibinfo{person}{Ke Yi}, \bibinfo{person}{Yuchao Tao}, {and} \bibinfo{person}{Ashwin Machanavajjhala}.} \bibinfo{year}{2024}\natexlab{b}.
\newblock \showarticletitle{Instance-optimal Truncation for Differentially Private Query Evaluation with Foreign Keys}.
\newblock \bibinfo{journal}{\emph{ACM Transactions on Database Systems}} \bibinfo{volume}{49}, \bibinfo{number}{4} (\bibinfo{year}{2024}), \bibinfo{pages}{1--40}.
\newblock


\bibitem[Dong et~al\mbox{.}(2023)]%
        {dong2023continual}
\bibfield{author}{\bibinfo{person}{Wei Dong}, \bibinfo{person}{Qiyao Luo}, {and} \bibinfo{person}{Ke Yi}.} \bibinfo{year}{2023}\natexlab{}.
\newblock \showarticletitle{Continual Observation under User-level Differential Privacy}. In \bibinfo{booktitle}{\emph{2023 IEEE Symposium on Security and Privacy (SP)}}. IEEE Computer Society, \bibinfo{pages}{2190--2207}.
\newblock


\bibitem[Dong and Yi(2021)]%
        {dong21:residual}
\bibfield{author}{\bibinfo{person}{Wei Dong} {and} \bibinfo{person}{Ke Yi}.} \bibinfo{year}{2021}\natexlab{}.
\newblock \showarticletitle{Residual Sensitivity for Differentially Private Multi-Way Joins}. In \bibinfo{booktitle}{\emph{Proc. ACM SIGMOD International Conference on Management of Data}}.
\newblock


\bibitem[Dong and Yi(2022)]%
        {dong2021nearly}
\bibfield{author}{\bibinfo{person}{Wei Dong} {and} \bibinfo{person}{Ke Yi}.} \bibinfo{year}{2022}\natexlab{}.
\newblock \showarticletitle{A Nearly Instance-optimal Differentially Private Mechanism for Conjunctive Queries}. In \bibinfo{booktitle}{\emph{Proc. ACM Symposium on Principles of Database Systems}}.
\newblock


\bibitem[Dwork et~al\mbox{.}(2006)]%
        {dwork2006calibrating}
\bibfield{author}{\bibinfo{person}{Cynthia Dwork}, \bibinfo{person}{Frank McSherry}, \bibinfo{person}{Kobbi Nissim}, {and} \bibinfo{person}{Adam Smith}.} \bibinfo{year}{2006}\natexlab{}.
\newblock \showarticletitle{Calibrating noise to sensitivity in private data analysis}. In \bibinfo{booktitle}{\emph{Theory of cryptography conference}}. Springer, \bibinfo{pages}{265--284}.
\newblock


\bibitem[Eden et~al\mbox{.}(2023)]%
        {eden2023triangle}
\bibfield{author}{\bibinfo{person}{Talya Eden}, \bibinfo{person}{Quanquan~C Liu}, \bibinfo{person}{Sofya Raskhodnikova}, {and} \bibinfo{person}{Adam Smith}.} \bibinfo{year}{2023}\natexlab{}.
\newblock \showarticletitle{Triangle Counting with Local Edge Differential Privacy}. In \bibinfo{booktitle}{\emph{50th International Colloquium on Automata, Languages, and Programming (ICALP 2023)}}. Schloss Dagstuhl--Leibniz-Zentrum f{\"u}r Informatik, \bibinfo{pages}{52--1}.
\newblock


\bibitem[Fang et~al\mbox{.}(2022)]%
        {fang2022shifted}
\bibfield{author}{\bibinfo{person}{Juanru Fang}, \bibinfo{person}{Wei Dong}, {and} \bibinfo{person}{Ke Yi}.} \bibinfo{year}{2022}\natexlab{}.
\newblock \showarticletitle{Shifted Inverse: A General Mechanism for Monotonic Functions under User Differential Privacy}.
\newblock  (\bibinfo{year}{2022}).
\newblock


\bibitem[Fichtenberger et~al\mbox{.}(2021)]%
        {fichtenberger2021differentially}
\bibfield{author}{\bibinfo{person}{Hendrik Fichtenberger}, \bibinfo{person}{Monika Henzinger}, {and} \bibinfo{person}{Wolfgang Ost}.} \bibinfo{year}{2021}\natexlab{}.
\newblock \showarticletitle{Differentially Private Algorithms for Graphs Under Continual Observation}. In \bibinfo{booktitle}{\emph{29th Annual European Symposium on Algorithms (ESA 2021)}}. Schloss Dagstuhl-Leibniz-Zentrum f{\"u}r Informatik.
\newblock


\bibitem[Flum and Grohe(2004)]%
        {flum2004parameterized}
\bibfield{author}{\bibinfo{person}{J{\"o}rg Flum} {and} \bibinfo{person}{Martin Grohe}.} \bibinfo{year}{2004}\natexlab{}.
\newblock \showarticletitle{The parameterized complexity of counting problems}.
\newblock \bibinfo{journal}{\emph{SIAM J. Comput.}} \bibinfo{volume}{33}, \bibinfo{number}{4} (\bibinfo{year}{2004}), \bibinfo{pages}{892--922}.
\newblock


\bibitem[Godsil and Royle(2013)]%
        {godsil2013algebraic}
\bibfield{author}{\bibinfo{person}{Chris Godsil} {and} \bibinfo{person}{Gordon~F Royle}.} \bibinfo{year}{2013}\natexlab{}.
\newblock \bibinfo{booktitle}{\emph{Algebraic graph theory}}. Vol.~\bibinfo{volume}{207}.
\newblock \bibinfo{publisher}{Springer Science \& Business Media}.
\newblock


\bibitem[Gottlob et~al\mbox{.}(2001)]%
        {gottlob2001complexity}
\bibfield{author}{\bibinfo{person}{Georg Gottlob}, \bibinfo{person}{Nicola Leone}, {and} \bibinfo{person}{Francesco Scarcello}.} \bibinfo{year}{2001}\natexlab{}.
\newblock \showarticletitle{The complexity of acyclic conjunctive queries}.
\newblock \bibinfo{journal}{\emph{Journal of the ACM (JACM)}} \bibinfo{volume}{48}, \bibinfo{number}{3} (\bibinfo{year}{2001}), \bibinfo{pages}{431--498}.
\newblock


\bibitem[Hay et~al\mbox{.}(2009)]%
        {hay2009accurate}
\bibfield{author}{\bibinfo{person}{Michael Hay}, \bibinfo{person}{Chao Li}, \bibinfo{person}{Gerome Miklau}, {and} \bibinfo{person}{David Jensen}.} \bibinfo{year}{2009}\natexlab{}.
\newblock \showarticletitle{Accurate estimation of the degree distribution of private networks}. In \bibinfo{booktitle}{\emph{2009 Ninth IEEE International Conference on Data Mining}}. IEEE, \bibinfo{pages}{169--178}.
\newblock


\bibitem[He et~al\mbox{.}(2024)]%
        {he2024butterfly}
\bibfield{author}{\bibinfo{person}{Yizhang He}, \bibinfo{person}{Kai Wang}, \bibinfo{person}{Wenjie Zhang}, \bibinfo{person}{Xuemin Lin}, \bibinfo{person}{Wei Ni}, {and} \bibinfo{person}{Ying Zhang}.} \bibinfo{year}{2024}\natexlab{}.
\newblock \showarticletitle{Butterfly Counting over Bipartite Graphs with Local Differential Privacy}. In \bibinfo{booktitle}{\emph{2024 IEEE 40th International Conference on Data Engineering (ICDE)}}. IEEE, \bibinfo{pages}{2351--2364}.
\newblock


\bibitem[He et~al\mbox{.}(2025)]%
        {he2025robust}
\bibfield{author}{\bibinfo{person}{Yizhang He}, \bibinfo{person}{Kai Wang}, \bibinfo{person}{Wenjie Zhang}, \bibinfo{person}{Xuemin Lin}, \bibinfo{person}{Ying Zhang}, {and} \bibinfo{person}{Wei Ni}.} \bibinfo{year}{2025}\natexlab{}.
\newblock \showarticletitle{Robust Privacy-Preserving Triangle Counting under Edge Local Differential Privacy}.
\newblock \bibinfo{journal}{\emph{Proceedings of the ACM on Management of Data}} \bibinfo{volume}{3}, \bibinfo{number}{3} (\bibinfo{year}{2025}), \bibinfo{pages}{1--26}.
\newblock


\bibitem[Hu et~al\mbox{.}(2025)]%
        {hu2025n2e}
\bibfield{author}{\bibinfo{person}{Yihua Hu}, \bibinfo{person}{Hao Ding}, {and} \bibinfo{person}{Wei Dong}.} \bibinfo{year}{2025}\natexlab{}.
\newblock \showarticletitle{N2E: A General Framework to Reduce Node-Differential Privacy to Edge-Differential Privacy for Graph Analytics}.
\newblock \bibinfo{journal}{\emph{Proceedings of the ACM on Management of Data}} \bibinfo{volume}{3}, \bibinfo{number}{6} (\bibinfo{year}{2025}), \bibinfo{pages}{1--26}.
\newblock


\bibitem[Hu et~al\mbox{.}(2026)]%
        {full_version}
\bibfield{author}{\bibinfo{person}{Yihua Hu}, \bibinfo{person}{Kuncan Wang}, {and} \bibinfo{person}{Wei Dong}.} \bibinfo{year}{2026}\natexlab{}.
\newblock \showarticletitle{Acyclic Graph Pattern Counting under Local Differential Privacy [Full Version]}.
\newblock  (\bibinfo{year}{2026}).
\newblock
\urldef\tempurl%
\url{https://drive.google.com/drive/folders/1L_4OA9tuXksSGel6iqcZxpzarLNc4-k8?usp=sharing}
\showURL{%
\tempurl}


\bibitem[Imola et~al\mbox{.}(2021)]%
        {Imola2021}
\bibfield{author}{\bibinfo{person}{Jacob Imola}, \bibinfo{person}{Takao Murakami}, {and} \bibinfo{person}{Kamalika Chaudhuri}.} \bibinfo{year}{2021}\natexlab{}.
\newblock \showarticletitle{Locally Differentially Private Analysis of Graph Statistics}. In \bibinfo{booktitle}{\emph{30th USENIX Security Symposium (USENIX Security 21)}}. \bibinfo{publisher}{USENIX Association}, \bibinfo{pages}{983--1000}.
\newblock


\bibitem[Imola et~al\mbox{.}(2022a)]%
        {Imola2022}
\bibfield{author}{\bibinfo{person}{Jacob Imola}, \bibinfo{person}{Takao Murakami}, {and} \bibinfo{person}{Kamalika Chaudhuri}.} \bibinfo{year}{2022}\natexlab{a}.
\newblock \showarticletitle{Communication-Efficient Triangle Counting under Local Differential Privacy}. In \bibinfo{booktitle}{\emph{31st USENIX Security Symposium (USENIX Security 22)}}. \bibinfo{publisher}{USENIX Association}, \bibinfo{address}{Boston, MA}, \bibinfo{pages}{537--554}.
\newblock


\bibitem[Imola et~al\mbox{.}(2022b)]%
        {imola2022differentially}
\bibfield{author}{\bibinfo{person}{Jacob Imola}, \bibinfo{person}{Takao Murakami}, {and} \bibinfo{person}{Kamalika Chaudhuri}.} \bibinfo{year}{2022}\natexlab{b}.
\newblock \showarticletitle{Differentially private triangle and 4-cycle counting in the shuffle model}. In \bibinfo{booktitle}{\emph{Proceedings of the 2022 ACM SIGSAC Conference on Computer and Communications Security}}. \bibinfo{pages}{1505--1519}.
\newblock


\bibitem[Itzkovitz and Alon(2005)]%
        {itzkovitz2005subgraphs}
\bibfield{author}{\bibinfo{person}{Shalev Itzkovitz} {and} \bibinfo{person}{Uri Alon}.} \bibinfo{year}{2005}\natexlab{}.
\newblock \showarticletitle{Subgraphs and network motifs in geometric networks}.
\newblock \bibinfo{journal}{\emph{Physical Review E—Statistical, Nonlinear, and Soft Matter Physics}} \bibinfo{volume}{71}, \bibinfo{number}{2} (\bibinfo{year}{2005}), \bibinfo{pages}{026117}.
\newblock


\bibitem[Johnson et~al\mbox{.}(2018)]%
        {johnson2018towards}
\bibfield{author}{\bibinfo{person}{Noah Johnson}, \bibinfo{person}{Joseph~P Near}, {and} \bibinfo{person}{Dawn Song}.} \bibinfo{year}{2018}\natexlab{}.
\newblock \showarticletitle{Towards practical differential privacy for SQL queries}.
\newblock \bibinfo{journal}{\emph{Proceedings of the VLDB Endowment}} \bibinfo{volume}{11}, \bibinfo{number}{5} (\bibinfo{year}{2018}), \bibinfo{pages}{526--539}.
\newblock


\bibitem[Joki{\'c} and Van~Mieghem(2022)]%
        {jokic2022number}
\bibfield{author}{\bibinfo{person}{Ivan Joki{\'c}} {and} \bibinfo{person}{Piet Van~Mieghem}.} \bibinfo{year}{2022}\natexlab{}.
\newblock \showarticletitle{Number of paths in a graph}.
\newblock \bibinfo{journal}{\emph{arXiv preprint arXiv:2209.08840}} (\bibinfo{year}{2022}).
\newblock


\bibitem[Karwa et~al\mbox{.}(2011)]%
        {karwa2011private}
\bibfield{author}{\bibinfo{person}{Vishesh Karwa}, \bibinfo{person}{Sofya Raskhodnikova}, \bibinfo{person}{Adam Smith}, {and} \bibinfo{person}{Grigory Yaroslavtsev}.} \bibinfo{year}{2011}\natexlab{}.
\newblock \showarticletitle{Private analysis of graph structure}.
\newblock \bibinfo{journal}{\emph{Proceedings of the VLDB Endowment}} \bibinfo{volume}{4}, \bibinfo{number}{11} (\bibinfo{year}{2011}), \bibinfo{pages}{1146--1157}.
\newblock


\bibitem[Kasiviswanathan et~al\mbox{.}(2011)]%
        {kasiviswanathan2011can}
\bibfield{author}{\bibinfo{person}{Shiva~Prasad Kasiviswanathan}, \bibinfo{person}{Homin~K Lee}, \bibinfo{person}{Kobbi Nissim}, \bibinfo{person}{Sofya Raskhodnikova}, {and} \bibinfo{person}{Adam Smith}.} \bibinfo{year}{2011}\natexlab{}.
\newblock \showarticletitle{What can we learn privately?}
\newblock \bibinfo{journal}{\emph{SIAM J. Comput.}} \bibinfo{volume}{40}, \bibinfo{number}{3} (\bibinfo{year}{2011}), \bibinfo{pages}{793--826}.
\newblock


\bibitem[Kasiviswanathan et~al\mbox{.}(2013)]%
        {kasiviswanathan2013analyzing}
\bibfield{author}{\bibinfo{person}{Shiva~Prasad Kasiviswanathan}, \bibinfo{person}{Kobbi Nissim}, \bibinfo{person}{Sofya Raskhodnikova}, {and} \bibinfo{person}{Adam Smith}.} \bibinfo{year}{2013}\natexlab{}.
\newblock \showarticletitle{Analyzing graphs with node differential privacy}. In \bibinfo{booktitle}{\emph{Theory of Cryptography Conference}}. Springer, \bibinfo{pages}{457--476}.
\newblock


\bibitem[Katz(1953)]%
        {katz1953new}
\bibfield{author}{\bibinfo{person}{Leo Katz}.} \bibinfo{year}{1953}\natexlab{}.
\newblock \showarticletitle{A new status index derived from sociometric analysis}.
\newblock \bibinfo{journal}{\emph{Psychometrika}} \bibinfo{volume}{18}, \bibinfo{number}{1} (\bibinfo{year}{1953}), \bibinfo{pages}{39--43}.
\newblock


\bibitem[Leskovec and Krevl(2014)]%
        {leskovec2014snap}
\bibfield{author}{\bibinfo{person}{Jure Leskovec} {and} \bibinfo{person}{Andrej Krevl}.} \bibinfo{year}{2014}\natexlab{}.
\newblock \bibinfo{title}{SNAP: Stanford network analysis project}.
\newblock


\bibitem[Liben-Nowell and Kleinberg(2003)]%
        {liben2003link}
\bibfield{author}{\bibinfo{person}{David Liben-Nowell} {and} \bibinfo{person}{Jon Kleinberg}.} \bibinfo{year}{2003}\natexlab{}.
\newblock \showarticletitle{The link prediction problem for social networks}. In \bibinfo{booktitle}{\emph{Proceedings of the twelfth international conference on Information and knowledge management}}. \bibinfo{pages}{556--559}.
\newblock


\bibitem[Milo et~al\mbox{.}(2002)]%
        {milo2002network}
\bibfield{author}{\bibinfo{person}{Ron Milo}, \bibinfo{person}{Shai Shen-Orr}, \bibinfo{person}{Shalev Itzkovitz}, \bibinfo{person}{Nadav Kashtan}, \bibinfo{person}{Dmitri Chklovskii}, {and} \bibinfo{person}{Uri Alon}.} \bibinfo{year}{2002}\natexlab{}.
\newblock \showarticletitle{Network motifs: simple building blocks of complex networks}.
\newblock \bibinfo{journal}{\emph{Science}} \bibinfo{volume}{298}, \bibinfo{number}{5594} (\bibinfo{year}{2002}), \bibinfo{pages}{824--827}.
\newblock


\bibitem[Narayanan and Shmatikov(2009)]%
        {narayanan2009anonymizing}
\bibfield{author}{\bibinfo{person}{Arvind Narayanan} {and} \bibinfo{person}{Vitaly Shmatikov}.} \bibinfo{year}{2009}\natexlab{}.
\newblock \showarticletitle{De-anonymizing social networks}. In \bibinfo{booktitle}{\emph{2009 30th IEEE symposium on security and privacy}}. IEEE, \bibinfo{pages}{173--187}.
\newblock


\bibitem[Ngo et~al\mbox{.}(2018)]%
        {ngo2018worst}
\bibfield{author}{\bibinfo{person}{Hung~Q Ngo}, \bibinfo{person}{Ely Porat}, \bibinfo{person}{Christopher R{\'e}}, {and} \bibinfo{person}{Atri Rudra}.} \bibinfo{year}{2018}\natexlab{}.
\newblock \showarticletitle{Worst-case optimal join algorithms}.
\newblock \bibinfo{journal}{\emph{Journal of the ACM (JACM)}} \bibinfo{volume}{65}, \bibinfo{number}{3} (\bibinfo{year}{2018}), \bibinfo{pages}{1--40}.
\newblock


\bibitem[Nissim et~al\mbox{.}(2007)]%
        {nissim2007smooth}
\bibfield{author}{\bibinfo{person}{Kobbi Nissim}, \bibinfo{person}{Sofya Raskhodnikova}, {and} \bibinfo{person}{Adam Smith}.} \bibinfo{year}{2007}\natexlab{}.
\newblock \showarticletitle{Smooth sensitivity and sampling in private data analysis}. In \bibinfo{booktitle}{\emph{Proceedings of the thirty-ninth annual ACM symposium on Theory of computing}}. \bibinfo{pages}{75--84}.
\newblock


\bibitem[Palla et~al\mbox{.}(2005)]%
        {palla2005uncovering}
\bibfield{author}{\bibinfo{person}{Gergely Palla}, \bibinfo{person}{Imre Der{\'e}nyi}, \bibinfo{person}{Ill{\'e}s Farkas}, {and} \bibinfo{person}{Tam{\'a}s Vicsek}.} \bibinfo{year}{2005}\natexlab{}.
\newblock \showarticletitle{Uncovering the overlapping community structure of complex networks in nature and society}.
\newblock \bibinfo{journal}{\emph{nature}} \bibinfo{volume}{435}, \bibinfo{number}{7043} (\bibinfo{year}{2005}), \bibinfo{pages}{814--818}.
\newblock


\bibitem[Sun et~al\mbox{.}(2019)]%
        {sun2019analyzing}
\bibfield{author}{\bibinfo{person}{Haipei Sun}, \bibinfo{person}{Xiaokui Xiao}, \bibinfo{person}{Issa Khalil}, \bibinfo{person}{Yin Yang}, \bibinfo{person}{Zhan Qin}, \bibinfo{person}{Hui Wang}, {and} \bibinfo{person}{Ting Yu}.} \bibinfo{year}{2019}\natexlab{}.
\newblock \showarticletitle{Analyzing subgraph statistics from extended local views with decentralized differential privacy}. In \bibinfo{booktitle}{\emph{Proceedings of the 2019 ACM SIGSAC conference on computer and communications security}}. \bibinfo{pages}{703--717}.
\newblock


\bibitem[Suppakitpaisarn et~al\mbox{.}(2025)]%
        {suppakitpaisarn2025counting}
\bibfield{author}{\bibinfo{person}{Vorapong Suppakitpaisarn}, \bibinfo{person}{Donlapark Ponnoprat}, \bibinfo{person}{Nicha Hirankarn}, {and} \bibinfo{person}{Quentin Hillebrand}.} \bibinfo{year}{2025}\natexlab{}.
\newblock \showarticletitle{Counting Graphlets of Size $ k $ under Local Differential Privacy}.
\newblock \bibinfo{journal}{\emph{arXiv preprint arXiv:2505.12954}} (\bibinfo{year}{2025}).
\newblock


\bibitem[Warner(1965)]%
        {warner1965randomized}
\bibfield{author}{\bibinfo{person}{Stanley~L Warner}.} \bibinfo{year}{1965}\natexlab{}.
\newblock \showarticletitle{Randomized response: A survey technique for eliminating evasive answer bias}.
\newblock \bibinfo{journal}{\emph{Journal of the American statistical association}} \bibinfo{volume}{60}, \bibinfo{number}{309} (\bibinfo{year}{1965}), \bibinfo{pages}{63--69}.
\newblock


\bibitem[Wu et~al\mbox{.}(2021)]%
        {wu2021adapting}
\bibfield{author}{\bibinfo{person}{Bang Wu}, \bibinfo{person}{Xiangwen Yang}, \bibinfo{person}{Shirui Pan}, {and} \bibinfo{person}{Xingliang Yuan}.} \bibinfo{year}{2021}\natexlab{}.
\newblock \showarticletitle{Adapting membership inference attacks to GNN for graph classification: Approaches and implications}. In \bibinfo{booktitle}{\emph{2021 IEEE International Conference on Data Mining (ICDM)}}. IEEE, \bibinfo{pages}{1421--1426}.
\newblock


\bibitem[Yannakakis(1981)]%
        {yannakakis1981algorithms}
\bibfield{author}{\bibinfo{person}{Mihalis Yannakakis}.} \bibinfo{year}{1981}\natexlab{}.
\newblock \showarticletitle{Algorithms for acyclic database schemes}. In \bibinfo{booktitle}{\emph{VLDB}}, Vol.~\bibinfo{volume}{81}. \bibinfo{pages}{82--94}.
\newblock


\bibitem[Ye et~al\mbox{.}(2020a)]%
        {ye2020lf}
\bibfield{author}{\bibinfo{person}{Qingqing Ye}, \bibinfo{person}{Haibo Hu}, \bibinfo{person}{Man~Ho Au}, \bibinfo{person}{Xiaofeng Meng}, {and} \bibinfo{person}{Xiaokui Xiao}.} \bibinfo{year}{2020}\natexlab{a}.
\newblock \showarticletitle{LF-GDPR: A framework for estimating graph metrics with local differential privacy}.
\newblock \bibinfo{journal}{\emph{IEEE Transactions on Knowledge and Data Engineering}} \bibinfo{volume}{34}, \bibinfo{number}{10} (\bibinfo{year}{2020}), \bibinfo{pages}{4905--4920}.
\newblock


\bibitem[Ye et~al\mbox{.}(2020b)]%
        {ye2020towards}
\bibfield{author}{\bibinfo{person}{Qingqing Ye}, \bibinfo{person}{Haibo Hu}, \bibinfo{person}{Man~Ho Au}, \bibinfo{person}{Xiaofeng Meng}, {and} \bibinfo{person}{Xiaokui Xiao}.} \bibinfo{year}{2020}\natexlab{b}.
\newblock \showarticletitle{Towards locally differentially private generic graph metric estimation}. In \bibinfo{booktitle}{\emph{2020 IEEE 36th International Conference on Data Engineering (ICDE)}}. IEEE, \bibinfo{pages}{1922--1925}.
\newblock


\bibitem[Zhang et~al\mbox{.}(2015)]%
        {zhang2015private}
\bibfield{author}{\bibinfo{person}{Jun Zhang}, \bibinfo{person}{Graham Cormode}, \bibinfo{person}{Cecilia~M Procopiuc}, \bibinfo{person}{Divesh Srivastava}, {and} \bibinfo{person}{Xiaokui Xiao}.} \bibinfo{year}{2015}\natexlab{}.
\newblock \showarticletitle{Private release of graph statistics using ladder functions}. In \bibinfo{booktitle}{\emph{Proceedings of the 2015 ACM SIGMOD international conference on management of data}}. \bibinfo{pages}{731--745}.
\newblock


\bibitem[Zhang et~al\mbox{.}(2022)]%
        {zhang2022inference}
\bibfield{author}{\bibinfo{person}{Zhikun Zhang}, \bibinfo{person}{Min Chen}, \bibinfo{person}{Michael Backes}, \bibinfo{person}{Yun Shen}, {and} \bibinfo{person}{Yang Zhang}.} \bibinfo{year}{2022}\natexlab{}.
\newblock \showarticletitle{Inference attacks against graph neural networks}. In \bibinfo{booktitle}{\emph{31st USENIX Security Symposium (USENIX Security 22)}}. \bibinfo{pages}{4543--4560}.
\newblock


\end{thebibliography}


\ifincludeappendix
  \newpage
  \appendix
  \section{Appendix}

\subsection{Redundant Count Removal}
\label{app:redundant-orientation}




\paragraph{$k$-line walk counting}

A natural approach to removing redundant counts is to divide the output of $\mathcal{M}_{\mathrm{k\text{-}wk}}$ by $2$. 
This applies when $k$ is odd, as each $k$-line walk is counted exactly twice by the mechanism. 
However, when $k$ is even, symmetric walks (e.g., $(v_1, v_2, v_1)$) exist and are counted only once, so dividing by $2$ would undercount.
Let $S_k$ denote the true number of symmetric $k$-line walks, and let $U_k$ denote the desired count of $k$-line walks without redundancy.
Then $U_k = (S_k + W_k)/2$. 
Since $\mathcal{M}_{\mathrm{k\text{-}wk}}$ produces an unbiased estimator of $W_k$, the task of estimating $U_k$ reduces to estimating $S_k$.
Intuitively, a symmetric $k$-line walk is fully determined by its first $k/2 {+}1$ nodes. 
Hence, $S_k = W_{k/2}$, which further reduces the problem to estimating $W_{k/2}$. 
As a result, we can modify $\mathcal{M}_{\mathrm{k\text{-}wk}}$ by introducing an additional step at round $k/2$, where the analyzer aggregates node values to obtain an unbiased estimator $\hat{S}_k$ of $S_k$. 
Using the example in Figure~\ref{fig:2-wk}, the analyzer aggregates the round $1$ outputs and obtains $\widehat{S}_k = 1.5$.
Finally, the analyzer outputs $(\mathcal{M}_\mathrm{k\text{-}wk}(G, \varepsilon) + \hat{S}_k)/2$, yielding an unbiased estimator of $U_k$.

Since the additional step only post-processes privatized results, this mechanism still satisfies $\varepsilon$-edge-LDP.
In terms of utility, the overall error is bounded by the sum of the individual errors for $\mathcal{M}_{\mathrm{k\text{-}wk}}(G, \varepsilon)$ and $\hat{S}_k$. 
Since the error of $\mathcal{M}_{\mathrm{k\text{-}wk}}(G, \varepsilon)$ asymptotically dominates, the overall utility bound remains $\tilde{O}(\sqrt{N}d(G)^{k-1})$.

\paragraph{$k$-line path counting}  
Unlike walks, paths do not allow repeated nodes, so symmetric paths do not exist.  
As a result, each $k$-line path is counted exactly twice in our mechanism, once for each orientation. 
To eliminate the count redundancy due to orientation, we divide $\mathcal{M}_{\mathrm{k\text{-}pt}}(G, \varepsilon)$ by $2$, 
yielding an unbiased estimator for the number of unoriented $k$-line paths while preserving the same privacy and utility guarantees.

\paragraph{$k$-edge acyclic pattern counting}

For a subgraph $H \subseteq G$ that matches the input pattern, multiple node mark assignments may lead to $H$ being counted. 
Consider a $k$-star pattern, where the pre-processed tree $\mathcal{T}$ consists of a root and $k$ leaves.
Any instance $H$ will be counted whenever its center is marked $k$ and its $k$ leaves receive distinct marks from $0$ to $k-1$, regardless of the specific ordering.
The sampling probability of $H$ is therefore $k!/(k+1)^{k+1}$, where $k!$ is the number of valid permutations of node marks in $H$. 
More generally, for an arbitrary acyclic pattern, the number of valid permutations equals the number of automorphisms~\cite{godsil2013algebraic}, denoted by $\sigma$.
Importantly, $\sigma$ depends only on the pattern and is independent of the particular tree formulation or vertex ordering.
To correct for this count redundancy, we adjust the rescaling factor to $(k+1)^{k+1}/\sigma$, yielding an unbiased estimator of the number of distinct acyclic pattern instances in $G$ while preserving the utility guarantees of Theorem~\ref{the:pattern_bound}.
Notably, this also applies to our one-round mechanism for $k$-star counting.

\subsection{Additional Experiment Results}
\label{app:exp}

This section presents the complete experimental results for Section~\ref{subsec:sample_dp}.
Table~\ref{tab:six_decomp} reports the decomposition of the total relative error into sampling error and DP noise for four selected queries, $Q_{\mathrm{4\text{-}pt}}$, $Q_{\mathrm{5\text{-}pt}}$, $Q_{\pi_1}$, and $Q_{\pi_2}$, across all six datasets.
Both error components are averaged over the same six runs used in the utility evaluation.
As shown in the table, the sampling error generally dominates the overall error, although its advantage over the DP noise is notably smaller than $d(G)$.
Figure~\ref{fig:full_nrep} presents the relative error decomposition on the \textbf{Epinions2} dataset for the same four queries under the LDP error optimization strategy, with $n_{\text{rep}}$ varying from $1$ to $10$.
For each query, the minimum total relative error over all values of $n_{\text{rep}}$ is highlighted in bold.
The results show that the optimization strategy is effective for all tested queries, yielding improvements of $2.8\times$, $1.8\times$, $4.6\times$, and $2.0\times$ in total LDP error at the optimal $n_{\text{rep}}$ values for $Q_{\mathrm{4\text{-}pt}}$, $Q_{\mathrm{5\text{-}pt}}$, $Q_{\pi_1}$, and $Q_{\pi_2}$, respectively, on this dataset.


\begin{table}[hbtp]
\centering
\small
\setlength{\tabcolsep}{4pt}
\resizebox{0.6\columnwidth}{!}{
\begin{tabular}{ccccc}
\toprule
\textbf{Dataset} & \textbf{Method} & \textbf{Rel. err (\%)} & \textbf{Sampling (\%)} & \textbf{DP (\%)} \\
\midrule
\multirow{4}{*}{\textbf{AstroPh}} 
 & $Q_{4\text{-pt}}$ & 5.67 & \cellcolor{gray!35}3.19 & 2.89 \\
 & $Q_{5\text{-pt}}$ & 11.77 & \cellcolor{gray!35}9.98 & 3.10 \\
 & $Q_{\pi_1}$       & 12.10   & \cellcolor{gray!35}14.97   & 3.26   \\
 & $Q_{\pi_2}$       & 11.53   & \cellcolor{gray!35}10.81   & 4.62   \\
\midrule
\multirow{4}{*}{\textbf{Enron} }
 & $Q_{4\text{-pt}}$ & 11.47 & \cellcolor{gray!35}10.42 & 2.09 \\
 & $Q_{5\text{-pt}}$ & 5.95  & \cellcolor{gray!35}6.35  & 4.78 \\
 & $Q_{\pi_1}$       & 21.90    & \cellcolor{gray!35}20.89    & 1.76   \\
 & $Q_{\pi_2}$       & 23.27   & \cellcolor{gray!35}20.24    & 13.16   \\
\midrule
\multirow{4}{*}{\textbf{Epinions1}} 
 & $Q_{4\text{-pt}}$ & 5.95 & \cellcolor{gray!35}6.48 & 0.68 \\
 & $Q_{5\text{-pt}}$ & 11.55 & \cellcolor{gray!35}10.52 & 2.16 \\
 & $Q_{\pi_1}$       & 11.71    & \cellcolor{gray!35}10.78    & 1.52   \\
 & $Q_{\pi_2}$       & 37.30    & \cellcolor{gray!35}36.79    & 4.26   \\
\midrule
\multirow{4}{*}{\textbf{Slashdot}} 
 & $Q_{4\text{-pt}}$ & 14.02 & \cellcolor{gray!35}13.40 & 2.35 \\
 & $Q_{5\text{-pt}}$ & 10.23  & \cellcolor{gray!35}8.07  & 3.78 \\
 & $Q_{\pi_1}$       & 27.54    & \cellcolor{gray!35}23.25    & 5.57   \\
 & $Q_{\pi_2}$       & 48.14    & \cellcolor{gray!35}28.47    & 23.64   \\
\midrule
\multirow{4}{*}{\textbf{Twitter}} 
 & $Q_{4\text{-pt}}$ & 8.91 & \cellcolor{gray!35}9.34 & 1.07 \\
 & $Q_{5\text{-pt}}$ & 4.13 & \cellcolor{gray!35}5.01 & 1.37 \\
 & $Q_{\pi_1}$       & 15.16   & \cellcolor{gray!35}14.16   & 1.54   \\
 & $Q_{\pi_2}$       & 18.64   & \cellcolor{gray!35}21.02   & 18.04   \\
\midrule
\multirow{4}{*}{\textbf{Epinions2}} 
 & $Q_{4\text{-pt}}$ & 8.79  & \cellcolor{gray!35}8.86  & 0.59 \\
 & $Q_{5\text{-pt}}$ & 9.55  & \cellcolor{gray!35}9.94  & 1.25 \\
 & $Q_{\pi_1}$       & 17.38 & \cellcolor{gray!35}17.21 & 0.58 \\
 & $Q_{\pi_2}$       & 17.24 & \cellcolor{gray!35}19.54 & 3.21 \\
\bottomrule
\end{tabular}
}
\caption{Relative error decomposition for four queries across six datasets.}
\label{tab:six_decomp}
\end{table}

\begin{figure}[!htbp]
    \centering
\includegraphics[width=0.62\linewidth]{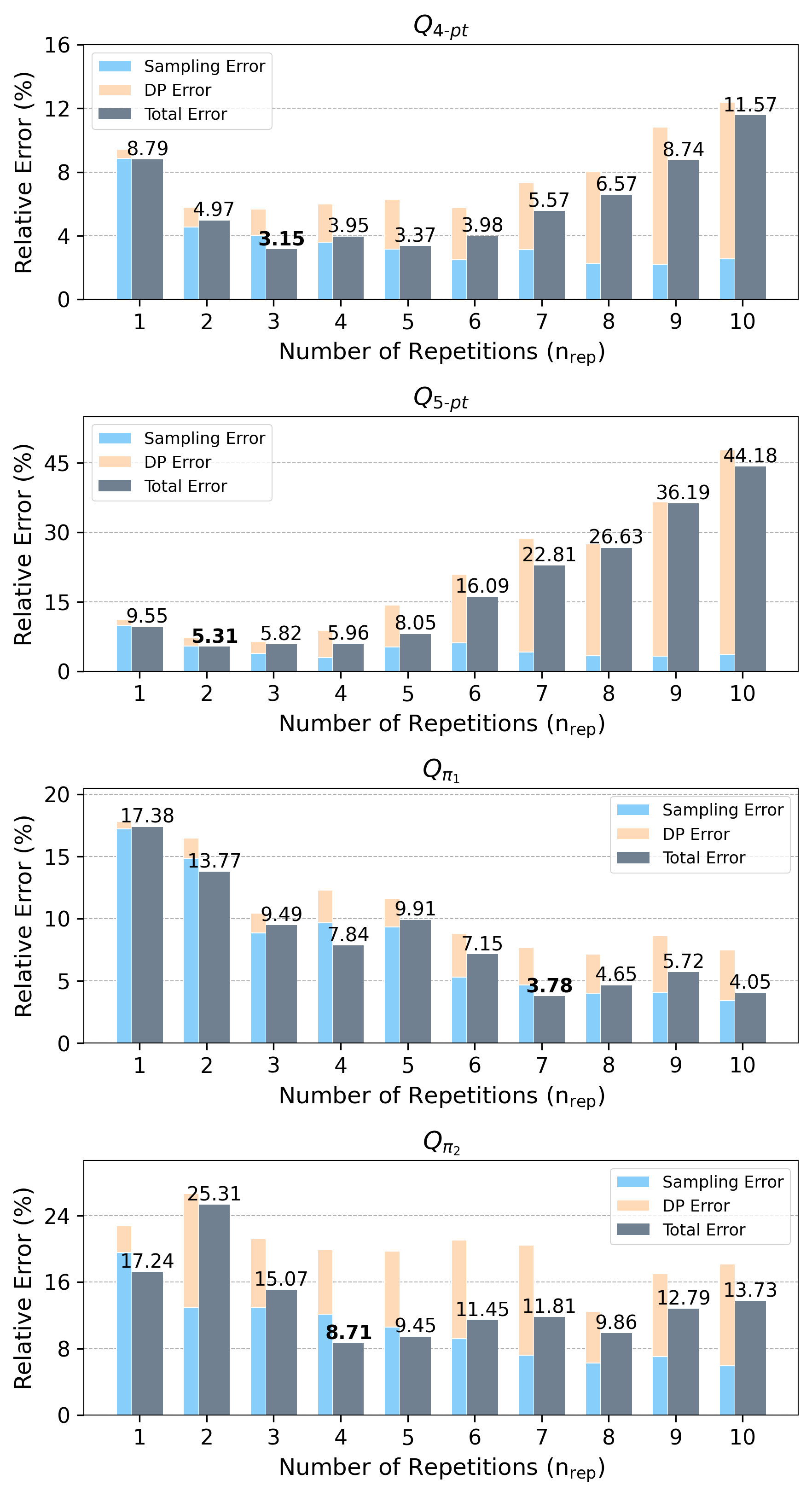}
    \caption{
    Relative error decomposition for four queries on the \textbf{Epinions2} dataset
        under varying number of repetitions ($n_{\text{rep}}$).
    }
    \label{fig:full_nrep}
\end{figure}
 
\subsection{Proofs}
\label{app:proof}


\begin{theorem}
For any graph $G$, walk length $k$, and privacy budget $\varepsilon$, $\mathcal{M}_{\mathrm{k\text{-}wk}}$ satisfies $\varepsilon$-edge-LDP.
\end{theorem}

\begin{proof}
We first show that the computation in each round $\ell \in \{1, \dots, k\}$ satisfies $\varepsilon/k$-edge-DP.
Let the differing edge between edge-neighboring graphs $G$ and $G'$ be $(v_i, v_j)$. 
In each round $\ell$, the computations at nodes $v_i$ and $v_j$ are affected by this edge: the presence of $(v_i, v_j)$ contributes an additional $\mathcal{X}^{(\ell-1)}_j$ to $v_i$'s result, and $\mathcal{X}^{(\ell-1)}_i$ to $v_j$'s result. 
Therefore, the per-node global sensitivity in round~$\ell$ is upper bounded by 
$\|\mathcal{X}^{(\ell-1)}\|_\infty$, which is the maximum of all nodes' absolute results at round $\ell{-}1$.
Applying the Laplace mechanism, adding a Laplace noise of scale $2k\|\mathcal{X}^{(\ell-1)}\|_\infty/\varepsilon$ as in Line~6 of the Algorithm~\ref{alg:k-wk_randomizer} ensures $\varepsilon/2k$-edge-DP for each node at each round. 
Since a single edge affects two nodes, parallel composition (Lemma~\ref{lem:pc}) implies that the collective message $\mathcal{X}^{(\ell)}$ satisfies $\varepsilon/k$-edge-DP at round $\ell$. 
Finally, applying basic composition over $k$ rounds (Lemma~\ref{lem:bc}) guarantees that the overall mechanism satisfies $\varepsilon$-edge-LDP.
\end{proof}


\begin{theorem}
For any graph $G$, walk length $k$, and parameters $\varepsilon$ and $\beta$, the $k$-round mechanism $\mathcal{M}_{\mathrm{k\text{-}wk}}$ satisfies
$\mathbb{E}[\mathcal{M}_{\mathrm{k\text{-}wk}}(G,\varepsilon)] = {W}_k$,
and with probability at least $1 - \beta$,
\begin{equation*}
|\mathcal{M}_{\mathrm{k\text{-}wk}}(G,\varepsilon) - {W}_k|
  \leq  k\gamma \sqrt{N} (d(G) + \gamma)^{k-1} = \tilde{O}(\sqrt{N}d(G)^{k-1}),
\end{equation*}
where $\gamma = 2k\sqrt{8\log(2kN/\beta)}/\varepsilon = \tilde{O}(1)$.  
\end{theorem}

\begin{proof}
At round $\ell \in \{1, \dots, k\}$, for any $i$,
\begin{equation}
\label{e:strawman_noisy_count}
    \mathcal{X}^{(\ell)}_i = \sum_{j \in \mathcal{N}(i)} \mathcal{X}^{(\ell-1)}_j + \mathrm{Lap}(\frac{2k}{\varepsilon}\|\mathcal{X}^{(\ell-1)}\|_\infty).
\end{equation}
It follows immediately that 
\begin{equation*}
    \bigg|\sum_{j \in \mathcal{N}(i)} \mathcal{X}^{(\ell-1)}_j \bigg| 
    \leq d(G)\, \|\mathcal{X}^{(\ell-1)}\|_\infty. 
\end{equation*}
Starting from the first round, we allocate a failure probability of $\beta'/N$ per node per round to upper-bound the added Laplace noise. 
Let 
$\gamma = 2k\sqrt{8\log(2N/\beta')}/\varepsilon$, applying the Laplace concentration bound (Lemma~\ref{lem:con_lap}), with probability at least $1- \beta'/N$, the Laplace noise added at any node in round $\ell$ satisfies
\begin{equation*}
   \bigg| \mathrm{Lap}(\frac{2k}{\varepsilon}\|\mathcal{X}^{(\ell-1)}\|_\infty) \bigg|
   \leq \gamma  \|\mathcal{X}^{(\ell-1)}\|_\infty.
\end{equation*}
Thus, $\|\mathcal{X}^{(1)}\|_\infty = d(G)+\gamma$, and we conclude by induction that
\begin{equation}
\label{e:round_max_wk}
    \|\mathcal{X}^{(\ell)}\|_\infty \leq (d(G) + \gamma)^{\ell}
\end{equation}
for any $\ell \leq k-1$ with probability at least $1 - (k-1)\beta'$.

In the round $k$, the analyzer $\mathcal{A}$ sums $\mathcal{X}^{(k)}_i$ over all $v_i \in V$. 
The Laplace noise added in round $\ell$ by node $v_i$ appears $W_{k-\ell}(i)$ times in the aggregated result, 
where $W_{k-\ell}(i)$ denotes the number of $(k-\ell)$-line walks starting from $v_i$. 
Thus, the total error introduced by the mechanism has
\begin{align*}
 \mathcal{M}_{\mathrm{k\text{-}wk}}(G,\varepsilon) - {W}_k
  &= \sum_{\ell=1}^{k} \sum_{i=1}^N 
  W_{k-\ell}(i) \cdot \mathrm{Lap}(\frac{2k}{\varepsilon} \|\mathcal{X}^{(\ell-1)}\|_\infty)\\
 &=\sum_{\ell=1}^{k} 
 \sum_{i=1}^N  
 \mathrm{Lap}(\frac{2k}{\varepsilon} \,
 \|\mathcal{X}^{(\ell-1)}\|_\infty \, W_{k-\ell}(i)).  
\end{align*}
Since the Laplace noises across nodes and rounds are independent and each has zero expectation, $\mathcal{M}_{\mathrm{k\text{-}wk}}(G,\varepsilon)$ is an unbiased estimator of $W_k$. 
Additionally, applying the Laplace concentration bound (Lemma~\ref{lem:con_lap}) on the Laplace noise added in the same round, for a fixed round $\ell$, with failure probability at least $1-\beta'/N$,
\begin{align*}
   \bigg| \sum_{i=1}^N  
 \mathrm{Lap}(\frac{2k}{\varepsilon} \, \|\mathcal{X}^{(\ell-1)}\|_\infty
 W_{k-\ell}(i)
 )
 \bigg|
    \leq 
    \gamma \|\mathcal{X}^{(\ell-1)}\|_\infty
     \sqrt{\sum_{i=1}^N W_{k-\ell}(i)^2}.
\end{align*}
Considering that $W_{k-\ell}(i) \leq d(G)^{k-\ell}$, the right-hand side is further upper-bounded by
\begin{equation*}
    \gamma (d(G) +\gamma)^{\ell-1} \sqrt{N} d(G)^{k-\ell}.
\end{equation*} 
Clearly, since $\gamma\geq0$, this expression is maximized when $\ell = k$, in which case the bound becomes 
$\gamma \sqrt{N}(d(G) + \gamma)^{k - 1}$.

Summing over all $k$ rounds, under probability at least $1 - k\beta'/N$, the total error is upper-bounded by $k\gamma \sqrt{N} (d(G) +\gamma)^{k-1}$.
By the union bound, the total failure probability is at most $(k{-}1)\beta' + k\beta'/N$. 
Since $k < N$ (otherwise no $k$-line walk exists in the graph), this sum is bounded by $k\beta'$. 
Therefore, setting $\beta' = \beta/k$ suffices, yielding $\gamma = 2k\sqrt{8\log(2Nk/\beta)}/\varepsilon$, and the error upper bound $k\gamma \sqrt{N} (d(G) +\gamma)^{k-1} = \tilde{O}(\sqrt{N}d(G)^{k-1})$.
\end{proof}


\begin{theorem}
For any graph $G$, walk length $k$, and privacy budget $\varepsilon$, the optimized mechanism $\mathcal{M}_{\mathrm{k\text{-}wk}}$ satisfies $\varepsilon$-edge-LDP.
\end{theorem}

\begin{proof}
To begin with, each of the first $k-2$ rounds guarantees $\varepsilon/k$-edge-DP, as established in the proof for the original mechanism. 
In the final round $k-1$, each node $v_i$ achieves $\varepsilon/k$-edge-DP by allocating a privacy budget of $\varepsilon/2k$ to privatize $\mathcal{X}^{(k-1)}_i$ before multiplication and the remaining $\varepsilon/2k$ to privatize $d(v_i)$, both via the Laplace mechanism.  
Because a differing edge $(v_i,v_j)$ affects only the computations for $v_i$ and $v_j$ in this round, parallel composition yields $2\varepsilon/k$-edge-DP. 
Hence, by basic composition across all $k-1$ rounds, $\mathcal{M}_{\mathrm{k\text{-}wk}}$ remains to achieve $\varepsilon$-edge-LDP. \qedhere

\end{proof}


\begin{theorem}
For any graph $G$, walk length $k$, and parameters $\varepsilon$ and $\beta$, the optimized $(k-1)$-round mechanism $\mathcal{M}_{\mathrm{k\text{-}wk}}$ satisfies
$
     \mathbb{E}[\mathcal{M}_{\mathrm{k\text{-}wk}}(G,\varepsilon)] = {W}_k,
$
and with probability at least $1 - \beta$,
\begin{equation*}
|\mathcal{M}_{\mathrm{k\text{-}wk}}(G,\varepsilon) - {W}_k|
  \leq  k\gamma \sqrt{N} (d(G) + \gamma)^{k-1} = \tilde{O}(\sqrt{N}d(G)^{k-1}),
\end{equation*}
where $\gamma = 2k\sqrt{8\log(2kN/\beta)}/\varepsilon = \tilde{O}(1)$.  
\end{theorem}

\begin{proof}
The computation of $\mathcal{X}^{(\ell)}_i$ in round $\ell \in \{1, \dots, k{-}1\}$ for each node $v_i \in V$ follows the same update rule as in \eqref{e:strawman_noisy_count}, where $\mathcal{X}^{(k-1)}_i$ refers to the value before multiplication by $d(v_i) + \mathrm{Lap}(2k/\varepsilon)$. 
Applying the same failure probability assignment, the bound in \eqref{e:round_max_wk} holds for all $\ell \leq k{-}1$ with probability at least $1 - (k{-}1)\beta'$, where $\gamma = 2k\sqrt{8\log(2N/\beta')}/\varepsilon$.

In the final round $k-1$, the analyzer $\mathcal{A}$ sums $\mathcal{X}^{(k-1)}_i \, (d(v_i) + \mathrm{Lap}(2k/\varepsilon))$ over all nodes $v_i \in V$ as the output.
We divide the total error of $|\mathcal{M}_{\mathrm{k\text{-}wk}}(G,\varepsilon) - {W}_k|$ into two parts.
The first part arises from the additive error in the term $\sum_{i = 1}^N \mathcal{X}^{(k-1)}_i \, d(v_i)$, which can be expressed as
\begin{equation*}
  \sum_{\ell=1}^{k-1} \sum_{i=1}^N  
  \mathrm{Lap}( 
\frac{2k}{\varepsilon}  \|\mathcal{X}^{(\ell-1)}\|_\infty \,
W_{k-\ell}(i)).
\end{equation*}
This part of the error matches that from the first $k-1$ rounds of the original mechanism and also has expectation $0$.
Similarly, allocating a failure probability of $\beta'/N$ for each round $\ell \in \{1, \dots, k-1\}$ yields an upper bound of $(k-1)\gamma \sqrt{N}(d(G)+\gamma)^{k-1}$.
The second part of the error is exactly 
\begin{equation*}
    \sum_{i=1}^N \mathcal{X}^{(k -1)}_i \, \mathrm{Lap}(\frac{2k}{\varepsilon}),
\end{equation*}
which also has an expectation of $0$.
By \eqref{e:round_max_wk}, $|\mathcal{X}^{(k-1)}_i| \leq (d(G) + \gamma)^{k-1}$. 
Thus, allocating a failure probability of $\beta'/N$, 
the Laplace concentration bound implies an upper bound of $\gamma\sqrt{N}(d(G) + \gamma)^{k-1}$ for this part of the error.

To conclude, the output of the optimized mechanism remains an unbiased estimator of $W_k$. 
By setting $\beta' = \beta/k$ and combining both sources of error, we obtain that, with probability at least $1 - \beta$, the total error is bounded by $k \gamma \sqrt{N} (d(G) + \gamma)^{k-1} = \tilde{O}(\sqrt{N}d(G)^{k-1})$ with the same $\gamma = 2k\sqrt{8\log(2kN/\beta)}/\varepsilon$. \qedhere

\end{proof}

\begin{lemma}
For any graph $G$, path length $k$, and parameter $\beta$, 
the $k$-round non-private algorithm $\widehat{\mathcal{M}}_{\mathrm{k\text{-}pt}}$ satisfies $
    \mathbb{E}[\widehat{\mathcal{M}}_{\mathrm{k\text{-}pt}}(G)] = P_k $,
and with probability at least $1 - \beta$,
\begin{equation*}
    |\widehat{\mathcal{M}}_{\mathrm{k\text{-}pt}}(G)-P_k| \leq
    (k+1) \sqrt{N\beta^{-1}}(d(G)+k+1)^k
    = O(\sqrt{N} d(G)^k).
\end{equation*}
\end{lemma}

\begin{proof}

We begin by associating a random variable $I$ with each $k$-line path in $G$, indexed arbitrarily. 
Notably, we consider paths in different orientations to be treated as distinct.
Let $I_q = 1$ if the $q$-th path is sampled by the algorithm, i.e., if the round indices of its nodes are exactly $0$ through $k$ in order; otherwise, set $I_q = 0$. 
Thus, $\mathbb{E}[I_q] = \Pr[I_q=1] = (k+1)^{-(k+1)}$.
Let $\rho = (k+1)^{-(k+1)}$, as the mechanism's output $\widehat{\mathcal{M}}_{\mathrm{k\text{-}pt}}(G) = \rho^{-1}\sum_q I_q$, we have
\begin{equation*}
    \mathbb{E}[\widehat{\mathcal{M}}_{\mathrm{k\text{-}pt}}(G)] = \rho^{-1}\sum_{q=1}^{P_k} \mathbb{E}[I_q] = P_k. 
\end{equation*}

Next, consider the variance of this output:
\begin{equation*}
    \operatorname{Var}(\rho^{-1}\sum_{q=1}^{P_k} I_q)
    =
    \rho^{-2}
     \sum_{q,{q'} =1}^{P_k} \operatorname{Cov}(I_q, I_{q'}).
\end{equation*}
Let $s$ be the number of the same nodes that occupy the same order in the $q$-th and the $q'$-th paths in $G$.
Then,
\begin{align*}
    \operatorname{Cov}(I_q, I_{q'}) 
    &= \mathbb{E}[I_qI_{q'}] - \mathbb{E}[I_q]\mathbb{E}[I_{q'}] \\
    &= \Pr[I_q = 1 \land I_{q'} = 1] - \rho^2\\
    &\leq \rho^2 \, ((k + 1)^{s}-1),
\end{align*}
where inequality in the last line arises from the case that two paths overlap on the same nodes but in different positions, then $\Pr[I_q = 1 \land I_{q'} = 1] = 0$.
When $s =0$, the covariance is at most $0$.
Now, fix the $q$-th oriented path, the number of oriented $k$-line paths that overlap it in exactly $s$ nodes is upper bounded by:
\begin{equation*}
    n(s, k) \leq \binom{k+1}{s}\, d(G)^{k+1-s}.
\end{equation*}
Given that $P_k\leq N d(G)^k$, the total output variance is bounded as:
\begin{align*}
  \mathrm{Var}(\rho^{-1}\sum_q I_q) &\leq 
   \rho^{-2} P_k \,
   \sum_{s=1}^{k+1} n(s, k) \, \rho^2 \, ((k + 1)^{s}-1)
   \\
   & \leq
   P_k \, \sum_{s=1}^{k+1} \binom{k+1}{s}
   d(G)^{k+1-s}\,{(k + 1)^s}.  
\end{align*}
Using binomial expansion, 
\begin{align*}
    \sum_{s=1}^{k+1} \binom{k+1}{s} d(G)^{k+1-s}\,{(k + 1)^s} 
    &= (d(G) +k+1)^{k+1} - d(G)^{k+1} \\
    &= (k+1)\,\sum_{i=0}^{k} (d(G) + k + 1)^{k-i} d(G)^{i}\\
    &\leq (k + 1)^{2} (d(G) + k + 1)^{k}.
\end{align*}

Finally, by Chebyshev's inequality (Lemma~\ref{lem:che_neq}), with probability at least $1- \beta$, 
\begin{align*}
     |\widehat{\mathcal{M}}_{\mathrm{k\text{-}pt}}(G) -P_k| 
     &\leq
  \sqrt{N d(G)^k ({k+1})^2(d(G)+k+1)^{k} \beta^{-1}} \\
  &\leq (k+1) \sqrt{N\beta^{-1}}(d(G)+k+1)^k\\
  &= {O}(\sqrt{N}d(G)^k). \qedhere
\end{align*}
\end{proof}


\begin{theorem}
    For any graph $G$, path length $k$, and privacy budget $\varepsilon$, the mechanism $\mathcal{M}_{\mathrm{k\text{-}pt}}$ satisfies $\varepsilon$-edge-LDP.
\end{theorem}

\begin{proof}
Let $G$ and $G'$ be two edge-neighboring graphs that differ by a single edge $(v_i, v_j)$.  
For each round $\ell \in \{1, \dots, k-1\}$, let $\mathcal{X}'^{(\ell)}$ denote the collective message produced by all $N$ nodes in $G'$.  
Let the node markings for $G$ and $G'$ be $\mathbf{r}= (r_1, \dots, r_N)$ and $\mathbf{r}' = (r'_1, \dots, r'_N)$, respectively.  
Since each node's mark is sampled independently of the graph's edge structure, the distributions of $\mathbf{r}$ and $\mathbf{r}'$ are identical and uniform over $\{0, 1, \dots, k\}^N$.  
Hence, releasing node markings does not incur any privacy loss.

Now, consider the case where both graphs have the same node markings, and the differing edge $(v_i, v_j)$ is sampled by the mechanism, i.e., nodes $v_i$ and $v_j$ receive consecutive marks.  
Without loss of generality, assume the computation at node $v_i$ is affected.  
Suppose $r_j, r_i < k - 1$ and $r_i = r_j + 1$. Then, the computations prior to round $r_i$ are identical for $G$ and $G'$, implying
\begin{equation}
\label{e:pt_privacy_1}
     \frac{
    \Pr[ 
    (\mathbf{r}, \mathcal{X}^{(1)}, \dots, \mathcal{X}^{(r_j)}) = \mathcal{H}^{(\leq r_j)}
    ]
  }{
    \Pr[ (\mathbf{r}', \mathcal{X}'^{(1)}, \dots, \mathcal{X}'^{(r_j)}) = \mathcal{H}^{(\leq r_j)}
    ]
  } = 1,
\end{equation}
where $\mathcal{H}^{(\ell)}$ denotes an arbitrary possible collective message at round $\ell$, and $\mathcal{H}^{(\leq\ell)}$ denotes the joint message through round $\ell$, including the random marking round.
We use $\mathcal{X}^{(\leq r_j)}$ to denote $(\mathbf{r}, \mathcal{X}^{(1)}, \dots, \mathcal{X}^{(r_j)})$ for $G$, and analogously for $G'$.
In round $r_i$, where the differing edge influences the computation, adding Laplace noise with scale $\|\mathcal{X}^{(r_i - 1)}\|_\infty/\varepsilon$ ensures:
\begin{equation}
\label{e:pt_privacy_2}
    \frac{\Pr[\mathcal{X}^{(r_i)} = \mathcal{H}^{(r_i)} \mid 
    \mathcal{X}^{(\leq r_j)} = \mathcal{H}^{(\leq r_j)}]}
    {\Pr[\mathcal{X}'^{(r_i)} = \mathcal{H}^{(r_i)} \mid 
    \mathcal{X}'^{(\leq r_j)} = \mathcal{H}^{(\leq r_j)}]} \leq e^\varepsilon.
\end{equation}


Then, we can further utilize the post-processing theorem to know that the remaining rounds of outputs do not lead to any privacy loss of edge $(v_i,v_j)$. Further combining \eqref{e:pt_privacy_1}, \eqref{e:pt_privacy_2}, we conclude that the full sequence of outputs satisfies $\varepsilon$-edge-LDP.  

In the case where $r_i = k-1$ and $r_j = k$, the joint outputs up to round $k-2$ are identically distributed across $G$ and $G'$, and the final round, where the differing edge affects node $v_i$'s output, satisfies $\varepsilon$-edge-DP. 
This also results in an overall $\varepsilon$-edge-LDP guarantee for the joint results of $k$ rounds.

Finally, if the differing edge is not sampled, then the outputs of all rounds are identical for $G$ and $G'$, resulting in zero privacy loss.  
Note that although the mechanism operates on a sampled subset of edges, sampling amplification does not apply here since node markings are released to the analyzer.
\end{proof}

\begin{lemma}
    For any graph $G$, path length $k$, parameters $\varepsilon$ and $\beta$, 
$\mathcal{M}_{\mathrm{k\text{-}pt}}$ satisfies
$
    \mathbb{E}[\mathcal{M}_{\mathrm{k\text{-}pt}}(G, \varepsilon)]
    =  
    \widehat{\mathcal{M}}_{\mathrm{k\text{-}pt}}(G)$
under the same node marking $(r_1, \dots, r_N)$, 
and with probability at least $1 - \beta$,
\begin{align*}
    |\mathcal{M}_{\mathrm{k\text{-}pt}}(G, \varepsilon) - 
    \widehat{\mathcal{M}}_{\mathrm{k\text{-}pt}}(G)| 
    &\leq 
    {k(k+1)} \,
    \gamma
    \sqrt{N}
    (\hat{d}(G) + \gamma \sqrt{\hat{d}(G)} + \gamma)^{k-1} 
    \\
    &=  \tilde{O}(\sqrt{N} d(G)^{k-1}),
\end{align*}
where $\gamma = (k+1)\sqrt{8\log(6N/\beta)}/\varepsilon = \tilde{O}(1) $ and $\hat{d}(G) = \max(d(G), \lceil\gamma^2\rceil)$.
\end{lemma}

\begin{proof}

As in the analysis for $k$-line walk counting under edge-LDP, each added Laplace noise has mean $0$.
Hence, $\mathbb{E}[\mathcal{M}_{\mathrm{k\text{-}pt}}(G, \varepsilon)] = \widehat{\mathcal{M}}_{\mathrm{k\text{-}pt}}(G)$ under the same node markings $(r_1,\dots,r_N)$.

To bound $|\mathcal{M}_{\mathrm{k\text{-}pt}}(G, \varepsilon) - 
    \widehat{\mathcal{M}}_{\mathrm{k\text{-}pt}}(G)|$, we introduce the notion of effective degree.
For a node $v_i$, its effective degree to round $c$ is defined as
\begin{equation}
    \label{e:eff_deg}
     d^{(c)}(v_i) := 
   \sum_{j\in \mathcal{N}(i)} \mathbf{1}\{r_j = c\}.
\end{equation} 
Since the markings $r_j$ are independent across $j \in \mathcal{N}(i)$,
\begin{equation*}
    d^{(c)}(v_i)  \;=\; \sum_{j\in \mathcal{N}(i)} B_j,
\qquad
B_j \sim \operatorname{Ber}(\frac{1}{k + 1}).\ \text{i.i.d.}
\end{equation*}
In our mechanism, for any node marked $r_i = \ell$, only $d^{(\ell\pm1)}(v_i)$ is relevant. 
Let $\gamma_1 = \sqrt{3(k + 1)\log(2N / \beta')}$ and $\hat{d}(G) = \max(d(G), \lceil \gamma_1^2 \rceil)$, then for every node $v_i$, $|\mathcal{N}(i)| \leq d(G) \leq \hat{d}(G)$.
Applying the multiplicative Chernoff bound (Lemma~\ref{lem:con_cher}), for each $v_i$ and $c = r_i \pm 1$, we have with probability at least $1 - \beta'$:
\begin{equation}
\label{e:d*_bound}
    d^{(c)}(v_i) \leq d^*(G) 
    = \frac{1}{k + 1} \left( \hat{d}(G) + \gamma_1 \sqrt{\hat{d}(G)} \right),
\end{equation}
where $\delta = \gamma_1 / \sqrt{\hat{d}(G)} \leq 1$.

Then, under the condition that \eqref{e:d*_bound} holds, we derive an upper bound on $\|\mathcal{X}^{(\ell)}\|_\infty$ for each $\ell \in \{1, \dots, k-1\}$, where $\mathcal{X}^{(k-1)}$ refers to the value before multiplication with $\sum_{j \in \mathcal{N}(i)} \mathbf{1}\{r_j = k\} + \mathrm{Lap}(1/\varepsilon)$.
Following the proof for \eqref{e:round_max_wk},
we allocate a failure probability of $\beta'/N$ to each node $v_i \in V$ to bound the Laplace noise added at round $r_i$ using the Laplace concentration bound (Lemma~\ref{lem:con_lap}). 
By the union bound, with probability at least $1 - \beta'$, for all $\ell \leq k-1$,
\begin{equation}
\label{e:max_sampling_mi}
    \|\mathcal{X}^{(\ell)}\|_\infty 
    \leq (d^*(G) + \frac{\gamma_2}{k+1})^{\ell}
     ,
\end{equation}
where $\gamma_2 = (k{+}1)\sqrt{8\log(2N/\beta')}/\varepsilon$.

In the final round $k{-}1$, the analyzer computes
\[
    \mathcal{M}_{\mathrm{k\text{-}pt}}(G, \varepsilon) =(k+1)^{k+1}
    \sum_{i=1}^N 
    \mathcal{X}^{(k-1)}_i
    \left(
        \sum_{j \in \mathcal{N}(i)} \mathbf{1}\{r_j = k\}
        + \mathrm{Lap}(\frac{1}{\varepsilon})
   \right).
\]
We divide the total error $|\mathcal{M}_{\mathrm{k\text{-}pt}}(G, \varepsilon) - 
    \widehat{\mathcal{M}}_{\mathrm{k\text{-}pt}}(G)|$ into two components. 
    The first arises from the additive error in 
\[
(k+1)^{k+1} \sum_{i=1}^N\mathcal{X}_i^{(k-1)} \sum_{j \in \mathcal{N}(i)}\mathbf{1}\{r_j = k\},
\]
and can be expressed as
\begin{equation*}
    (k+1)^{k+1} \, \sum_{\ell=1}^{k-1} \sum_{i=1}^N  
   \mathrm{Lap}(
   \frac{1}{\varepsilon} 
   \,
    \|\mathcal{X}^{(\ell-1)}\|_\infty   \,
   {P}_{k-\ell}(i)
   ),
\end{equation*}
where ${P}_{k-\ell}(i)$ denotes the number of $(k{-}\ell)$-line paths starting with node $v_i$, with marks from $\ell$ to $k$. 
Clearly, ${P}_{k-\ell}(i) \leq {d^*(G)}^{k-\ell}$.
Allocating a failure probability of $\beta'/N$ for each round $\ell \in \{1, \dots, k{-}1\}$ yields an upper bound of 
$
   (k+1)^{k}(k-1) \,  \gamma_2  \sqrt{N}(d^*(G) + {\gamma_2}/({k+1}))^{k-1}.
$
The second part of the error is
   \begin{align*}
 (k+1)^{k+1}
 \,
 \sum_{i=1}^N  \mathrm{Lap}(\frac{1}{\varepsilon} \, \mathcal{X}^{(k -1)}_i),
\end{align*}
which is upper bounded by $
     (k+1)^{k} \, \gamma_2 \sqrt{N}(d^*(G) + {\gamma_2}/{(k+1)})^{k-1}
$
using the Laplace concentration bound with a failure probability of $\beta'/N$.

Summing up both parts, the total error is bounded by
\begin{align*}
 &(k+1)^{k} k \, \gamma_2 \sqrt{N}(d^*(G) + \frac{\gamma_2}{k+1})^{k-1} \\
   &= {k(k+1)} \, \gamma_2 \sqrt{N}(\hat{d}(G) + \gamma_1\sqrt{\hat{d}(G)} + \gamma_2)^{k-1}
\end{align*}
with probability at least $1- \beta'-\beta' -k\beta'/N \geq 1-3\beta'$. 
Thus, setting $\beta' = \beta/3$ ensures the same error bound holds with probability at least $1-\beta$, with
$\gamma_1 = \sqrt{3(k+1)\log(6N/\beta)}$, $\gamma_2 = (k+1)\sqrt{8\log(6N/\beta)}/\varepsilon$, and $\hat{d}(G) = \max(d(G), \lceil \gamma_1^2 \rceil)$.
In typical settings where $\varepsilon \leq \sqrt{8(k+1)/3}$, we have $\gamma_2 \geq \gamma_1$, so we replace $\gamma_1$ with $\gamma_2$ in the error bound for simplicity.
Finally, since $\hat{d}(G) = O(d(G))$, the error bound is of 
$\tilde{O}(\sqrt{N}d(G)^{k-1})$. \qedhere

\end{proof}

\begin{theorem}
    For any graph $G$, path length $k$, parameters $\varepsilon$ and $\beta$, the $k$-round mechanism $\mathcal{M}_{\mathrm{k\text{-}pt}}$ satisfies $
    \mathbb{E}[\mathcal{M}_{\mathrm{k\text{-}pt}}(G, \varepsilon)] =  
    P_k$,
and with probability at least $1 - \beta$,
\begin{align*}
    |\mathcal{M}_{\mathrm{k\text{-}pt}}(G, \varepsilon) - 
    P_k|  \leq  (k+1) \sqrt{2N\beta^{-1}}(d(G)+k+1)^k + \\ 
          {k(k+1)} \, 
        \gamma
        \sqrt{N}
        (\hat{d}(G) + \gamma \sqrt{\hat{d}(G)} + \gamma)^{k-1} 
          =  \tilde{O}(\sqrt{N} d(G)^{k}),
\end{align*}
where $\gamma = (k+1)\sqrt{8\log(12N/\beta)}/\varepsilon = \tilde{O}(1) $ and $\hat{d}(G) = \max(d(G), \lceil\gamma^2\rceil)$.
\end{theorem}

\begin{proof}
We first establish that $
\mathbb{E}[\mathcal{M}_{\mathrm{k\text{-}pt}}(G, \varepsilon)] = P_k$.
By Lemma~\ref{lem:k-lp_sampling_bound}, $
\mathbb{E}[\widehat{\mathcal{M}}_{\mathrm{k\text{-}pt}}(G)] = P_k$.
Moreover, Lemma~\ref{lem:k-lp_additive_bound} implies that, for any fixed assignment of node markings $\mathbf{r} = (r_1, \dots, r_N)$,
$
\mathbb{E}[\mathcal{M}_{\mathrm{k\text{-}pt}}(G, \varepsilon)] = \widehat{\mathcal{M}}_{\mathrm{k\text{-}pt}}(G)
$.
For clarity, we use $\mathcal{M}_{\mathrm{k\text{-}pt}}(G, \varepsilon, \mathbf{r})$ and $\widehat{\mathcal{M}}_{\mathrm{k\text{-}pt}}(G, \mathbf{r})$ to denote the outputs of the respective mechanisms under fixed markings $\mathbf{r}$. Since the node markings are drawn uniformly from $\{0, 1, \dots, k\}^N$, we take the expectation over all such $\mathbf{r}$:
\begin{align*}
     \mathbb{E}[\mathcal{M}_{\mathrm{k\text{-}pt}}(G, \varepsilon)] 
     &= \sum_{\mathbf{r} \in \{0, 1, \dots, k\}^N} \Pr[\mathbf{r}] \cdot \mathbb{E}[\mathcal{M}_{\mathrm{k\text{-}pt}}(G, \varepsilon,  \mathbf{r})] \\
     &= \sum_{\mathbf{r}} \Pr[\mathbf{r}] \cdot \widehat{\mathcal{M}}_{\mathrm{k\text{-}pt}}(G, \mathbf{r}) \\
     &= \mathbb{E}[\widehat{\mathcal{M}}_{\mathrm{k\text{-}pt}}(G)] = P_k.
\end{align*}
This concludes the proof that $\mathcal{M}_{\mathrm{k\text{-}pt}}(G, \varepsilon)$ is an unbiased estimator of $P_k$.

Secondly, for the utility bound, consider an arbitrary marking $\mathbf{r} \in \{0, 1, \dots, k\}^N$. We apply the triangle inequality:
\begin{equation*}
|\mathcal{M}_{\mathrm{k\text{-}pt}}(G, \varepsilon) - 
    P_k| \leq |\mathcal{M}_{\mathrm{k\text{-}pt}}(G, \varepsilon) - 
    \widehat{\mathcal{M}}_{\mathrm{k\text{-}pt}}(G)| 
    + | \widehat{\mathcal{M}}_{\mathrm{k\text{-}pt}}(G) - P_k|.  
\end{equation*}
Combining the utility bound in Lemma~\ref{lem:k-lp_sampling_bound} and Lemma~\ref{lem:k-lp_additive_bound} with a failure probability of $\beta/2$ to each, the union bound implies that the result in the theorem holds with probability at least $1 - \beta$. \qedhere

\end{proof}

\begin{lemma}
    For any graph $G$, $k$-edge pattern $\mathcal{T}$, and parameter $\beta$, 
the $(k + 2 - |\mathcal{L}|)$-round non-private algorithm $\widehat{\mathcal{M}}_{\mathcal{T}}$ satisfies
$
    \mathbb{E}[\widehat{\mathcal{M}}_\mathcal{T}(G)] = Q_\mathcal{T}
$
, and with probability at least $1 - \beta$,
\begin{equation*}
    |  \widehat{\mathcal{M}}_\mathcal{T}(G)- Q_\mathcal{T}| \leq
    (k+1) \sqrt{N\beta^{-1}}(d(G)+k+1)^k
    = O(\sqrt{N} d(G)^k).
\end{equation*}

\end{lemma}

\begin{proof}

We associate a random variable $I_q$ with each pattern instance in $G$, indexed arbitrarily. 
Patterns with different valid node markings according to the given tree structure $\mathcal{T}$ are considered different.
Let $I_q = 1$ if the pattern is sampled in the algorithm with its specific node markings in the pattern, and $I_q = 0$ otherwise.
Since node mark assignments are independent, we have
$
    \rho = \Pr[I_q=1] = {(k+1)^{-(k+1)}},
$
and there are a total of $Q_\mathcal{T}$ variables.
As $\widehat{\mathcal{M}}_\mathcal{T}(G) = \rho^{-1}\sum_q I_q$, the variance of $\widehat{\mathcal{M}}_\mathcal{T}(G)$ can be represented as
\begin{equation*}
\operatorname{Var}(\rho^{-1}
\sum_{q=1}^{Q_\mathcal{T}} I_q)
   =\rho^{-2} 
   \sum_{q,{q'} =1}^{Q_\mathcal{T}} 
   \operatorname{Cov}(I_q, I_{q'}).
\end{equation*}
If marked pattern instances corresponding to $q$ and $q'$ share $s$ overlapping nodes, i.e., the same nodes assigned the same marks, then it holds that,
\begin{equation*}
    \operatorname{Cov}(I_q, I_{q'}) 
    \leq \rho^2 \, ((k + 1)^{s}-1).
\end{equation*}
The covariance is at most $0$ if $s = 0$.
When $s \geq 1$, fix the $q$-th pattern, the number of patterns in $G$ that overlap it in exactly $s$ nodes is bounded by 
\begin{equation*}
     n(s, k) =\binom{k+1}{s}\,d(G)^{k+1-s}.
\end{equation*}
Given that $Q_\mathcal{T} \leq Nd(G)^k$, similar to the proof of Lemma~\ref{lem:k-lp_sampling_bound}, the total variance has:
\begin{align*}
     \operatorname{Var}(\rho^{-1}\,
\sum_{q=1}^{Q_\mathcal{T}} I_q)
&\leq 
 \rho^{-2} Q_\mathcal{T} \,
   \sum_{s=1}^{k+1} n(s, k) \, \rho^2 \, ((k + 1)^{s}-1)
   \\
     &\leq  Nd(G)^k (k + 1)^{2} (d(G) + k + 1)^{k}. 
\end{align*}

Finally, by Lemma~\ref{lem:che_neq}, with probability of at least $1 - \beta$, 
\begin{align*}
   |\widehat{\mathcal{M}}_\mathcal{T}(G)- Q_\mathcal{T}| &\leq
    (k+1) \sqrt{N\beta^{-1}}(d(G)+k+1)^k\\
    &= O(\sqrt{N} d(G)^k)
 .  \qedhere
\end{align*}
 
\end{proof}

\begin{theorem}
For any graph $G$, $k$-edge pattern $\mathcal{T}$, and privacy budget $\varepsilon$, 
the mechanism $\mathcal{M}_\mathcal{T}$ satisfies $\varepsilon$-edge-LDP.
\end{theorem}

\begin{proof}
Similar to the proof of Theorem~\ref{the:k-lp_privacy}, consider any pair of edge-neighboring graphs $G$ and $G'$, which differ by a single edge, denoted $(v_i, v_j)$. 
This differing edge can affect the computation of at most one node.
Without loss of generality, assume the computation at node $v_i$ is affected. 
In this case, we have $r_j \in \mathcal{C}(r_i)$, so the contribution from $v_j$ appears in the values from round $r_j$ during the aggregation at $v_i$ in round $r_i$. Adding Laplace noise with scale $\|\mathcal{X}^{(r_j)}\|_\infty / \varepsilon$ ensures $\varepsilon$-edge-DP for the aggregated result $\mathcal{X}^{(r_i)}$ in round $r_i$.
All computations at other nodes are unaffected.
By parallel composition (Lemma~\ref{lem:pc}), the joint results across all rounds satisfy $\varepsilon$-edge-DP. 
Therefore, $\mathcal{M}_\mathcal{T}$ satisfies $\varepsilon$-edge-LDP.
\end{proof}

\begin{lemma}
    For any graph $G$, $k$-edge pattern $\mathcal{T}$, parameters $\varepsilon$ and $\beta$, 
$\mathcal{M}_{\mathcal{T}}$ satisfies
$\mathbb{E}[\mathcal{M}_{\mathcal{T}}(G, \varepsilon)] = 
    \widehat{\mathcal{M}}_{\mathcal{T}}(G)$ 
under the same node marking $(r_1, \dots, r_N)$, 
and with probability at least $1 - \beta$,
\begin{align*}
    |\mathcal{M}_{\mathcal{T}}(G, \varepsilon) - 
    \widehat{\mathcal{M}}_{\mathcal{T}}(G)| 
    &\leq 
    {k(k+1)} \,
    \gamma
    \sqrt{N}
    (\hat{d}(G) + \gamma \sqrt{\hat{d}(G)} + \gamma)^{k-1} 
    \\
    &=  \tilde{O}(\sqrt{N} d(G)^{k-1}),
\end{align*}
where $\gamma = (k{+}1)\sqrt{8\log(6kN/\beta)}/\varepsilon = \tilde{O}(1) $ and $\hat{d}(G) = \max(d(G), \lceil \gamma^2 \rceil)$.
\end{lemma}

\begin{proof}

We first show that $
    \mathbb{E}[\mathcal{M}_{\mathcal{T}}(G, \varepsilon)] = 
    \widehat{\mathcal{M}}_{\mathcal{T}}(G)
$. 
For any node $v_i$ and round $\ell = r_i$ where $\ell \notin \mathcal{L}$, we have:
\begin{align}
    \mathcal{X}^{(\ell)}_i 
&= \prod_{c \in \mathcal{C}(\ell)}  \mathcal{Y}^{(c)}_i \nonumber \\
&=
\label{e:x_pattern_update_rule}
\prod_{c \in \mathcal{C}(\ell)} 
\left(
\sum_{j \in \mathcal{N}(i)} 
           {\mathcal{X}}_j^{(c)} \,
           \mathbf{1}\{r_j = c\}
+ 
\mathrm{Lap}(\frac{1}{\varepsilon}\|\mathcal{X}^{(c)}\|_\infty )
\right)
\end{align}
Let $\widehat{\mathcal{X}}^{(\ell)}_i$ and $\widehat{\mathcal{Y}}^{(c)}_i$ denote the non-private values computed by Algorithm~\ref{alg:pattern_nondp} with the same pattern $\mathcal{T}$ and node markings.
We first consider nodes $v_i$ whose child rounds all lie in $\mathcal{L}$.
In that case, we have:
\begin{equation*}
    \mathcal{X}^{(\ell)}_i 
= \prod_{c \in \mathcal{C}(\ell)} 
\left(
\sum_{j \in \mathcal{N}(i)} 
           \mathbf{1}\{r_j = c\}
+ 
\mathrm{Lap}(\frac{1}{\varepsilon})
\right).
\end{equation*}
Since the Laplace noise has an expectation of $0$, for each $c \in \mathcal{C}(\ell)$, 
\begin{equation*}
    \mathbb{E}[\sum_{j \in \mathcal{N}(i)} 
           \mathbf{1}\{r_j = c\}
+ 
\mathrm{Lap}(\frac{1}{\varepsilon})] = \sum_{j \in \mathcal{N}(i)} 
           \mathbf{1}\{r_j = c\}. 
\end{equation*}
Because the Laplace noise added for each $c$ is sampled independently and $\sum_{j \in \mathcal{N}(i)} 
           \mathbf{1}\{r_j = c\}$ is deterministic for fixed node markings, $\mathcal{Y}^{(c)}_i$ across different $c \in \mathcal{C}(\ell)$ are independent. 
Thus,
\begin{equation*}
    \mathbb{E}[\mathcal{X}^{(\ell)}_i] = 
    \prod_{c \in \mathcal{C}(\ell)} 
    \sum_{j \in \mathcal{N}(i)} 
           \mathbf{1}\{r_j = c\} = \widehat{\mathcal{X}}_i^{(\ell)}.
\end{equation*}
By structural induction over the tree $\mathcal{T}$, we conclude that for every node $v_i$ with $r_i = \ell \notin \mathcal{L}$, $\mathbb{E}[\mathcal{X}^{(\ell)}_i] = \widehat{\mathcal{X}}^{(\ell)}_i$.
Therefore, $\mathbb{E}[\mathcal{M}_{\mathcal{T}}(G, \varepsilon)] = 
\widehat{\mathcal{M}}_{\mathcal{T}}(G)$ as claimed.

Then, we bound each node's effective degree, defined as in \eqref{e:eff_deg}.
For a given $\mathcal{T}$ and node $v_i$ with mark $\ell$, the effective degrees to the parent round $\mathcal{P}(\ell)$ and its child rounds $\mathcal{C}(\ell)$ are used by the mechanism.
As a result, there are at most $k$ such $c$ values that we need to bound for each $v_i$.
Following the proof of Lemma~\ref{lem:k-lp_additive_bound}, and applying the multiplicative Chernoff bound, for any node $v_i$ and relevant round $c \in \mathcal{C}(\ell) \cup \mathcal{P}(\ell)$, with probability at least $1 - \beta'$,
\begin{equation}
\label{e:pattern_d_bound}
    d^{(c)}(v_i) \leq d^*(G) = \frac{1}{k+1} \left( \hat{d}(G) + \gamma_1 \sqrt{\hat{d}(G)} \right),
\end{equation}
where $\gamma_1 = \sqrt{3(k{+}1)\log(kN/\beta')}$ and $\hat{d}(G) := \max(d(G), \lceil \gamma_1^2 \rceil)$.

We next derive an upper bound on $\|\mathcal{X}^{(\ell)}\|_\infty$ for $\ell \in \{1, \dots, k\}$, assuming that \eqref{e:pattern_d_bound} holds. 
For the first part of $\mathcal{Y}_i^{(c)}$, for any node $v_i$ with mark $\ell$ and $c \in \mathcal{C}(\ell)$,
\begin{equation*}
   \bigg| \sum_{j \in \mathcal{N}(i)} 
           {\mathcal{X}}_j^{(c)} \,
           \mathbf{1}\{r_j = c\} \bigg|
        \leq 
        d^*(G) \, \|\mathcal{X}^{(c)}\|_\infty.
\end{equation*}
For the second part of $\mathcal{Y}_i^{(c)}$, by the Laplace Concentration Bound with a failure probability of $\beta'/kN$ for each Laplace noise added at each node,
\begin{equation*}
    \bigg|\, \mathrm{Lap}(\frac{1}{\varepsilon}\|\mathcal{X}^{(c)}\|_\infty) \bigg|
    \leq \frac{\gamma_2}{k+1} \|\mathcal{X}^{(c)}\|_\infty,
\end{equation*}
with $\gamma_2 = (k+1)\sqrt{8\log(2kN/\beta')}/\varepsilon$. 
Combining both parts,
\begin{equation*}
     |\mathcal{X}^{(\ell)}_i|  
     \leq 
     \prod_{c \in \mathcal{C}(\ell)} 
     (d^*(G) + \frac{\gamma_2}{k+1}) \|\mathcal{X}^{(c)}\|_\infty.
\end{equation*}
Recall that for $c \in \mathcal{L}$, we have $\|\mathcal{X}^{(c)}\|_\infty = 1$. 
By induction, we can conclude that for any $\ell \leq k$, with probability at least $1 - \beta'$,
\begin{equation}
\label{e:pattern_round_max}
   \|\mathcal{X}^{(\ell)}\|_\infty 
   \leq 
   (d^*(G) + \frac{\gamma_2}{k+1})^{T_\ell},
\end{equation} 
where $T_\ell$ denotes the number of edges in the subtree $\mathcal{T}(u_\ell)$. 
Moreover, for any node $v_i$ with mark $\ell$ and $c \in \mathcal{C}(\ell)$,
\begin{equation}
\label{e:y_max}
 \mathcal{Y}^{(c)}_i \leq (d^*(G) + \frac{\gamma_2}{k+1})^{T_c + 1}.
\end{equation}

Finally, we bound the additive error introduced by Laplace noise.
For the Laplace noise of scale $\|\mathcal{X}^{(\ell)}\|_\infty/\varepsilon$ added in round $\mathcal{P}(\ell)$ at any active node, the corresponding multiplicative factor in the final output is, by \eqref{e:pattern_round_max} and \eqref{e:y_max}, upper bounded by
\begin{equation*}
    (k+1)^{k+1} (d^*(G) + \frac{\gamma_2}{k+1} )^{k -T_\ell-1}.
\end{equation*}
As a result, applying the Laplace concentration bound with failure probability $\beta'/kN$ to the sum of all such Laplace noises in a fixed round $\ell$, the error is bounded by
\begin{align*}
    &(k+1)^k \,\gamma_2 \sqrt{N} (d^*(G) + \frac{\gamma_2}{k+1} )^{k - 1} \\
    &\leq (k+1) \gamma_2 \sqrt{N} (\hat{d}(G) + \gamma_1 \sqrt{\hat{d}(G)} + \gamma_2 )^{k - 1}.
\end{align*}

Summing over all $k$ rounds, the total additive error is bounded by
\begin{equation*}
    |\mathcal{M}_{\mathcal{T}}(G, \varepsilon) - 
    \widehat{\mathcal{M}}_{\mathcal{T}}(G)| \leq (k+1)k \gamma_2 \sqrt{N} (\hat{d}(G) + \gamma_1 \sqrt{\hat{d}(G)} + \gamma_2 )^{k - 1},
\end{equation*}
with probability at least $1 - 2 \beta' - \beta'/N \geq 1 - 3\beta'$.
Thus, setting $\beta' = \beta/3$ ensures that the above error bound holds with probability at least $1 - \beta$, with
$\gamma_1 = \sqrt{3(k+1)\log(6kN/\beta)}$, $\gamma_2 = (k+1)\sqrt{8\log(6kN/\beta)}/\varepsilon$, and $\hat{d}(G) = \max(d(G), \lceil \gamma_1^2 \rceil)$.
In typical settings where $\varepsilon \leq \sqrt{8(k+1)/3}$, we have $\gamma_2 \geq \gamma_1$, so we replace $\gamma_1$ with $\gamma_2$ in the error bound for simplicity.
Furthermore, since $\hat{d}(G) = O(d(G))$, the error bound is $\tilde{O}(\sqrt{N} d(G)^{k - 1})$. \qedhere

\end{proof}

\begin{theorem}
For any graph $G$, $k$-edge pattern $\mathcal{T}$, parameters $\varepsilon$ and $\beta$,
the $(k + 2 - |\mathcal{L}|)$-round mechanism $\mathcal{M}_\mathcal{T}$ satisfies 
$\mathbb{E}[\mathcal{M}_\mathcal{T}(G, \varepsilon)] = Q_\mathcal{T}$,
and with probability at least $1 - \beta$,
\begin{align*}
    |\mathcal{M}_\mathcal{T}(G, \varepsilon) - Q_\mathcal{T}| 
    \leq
    (k+1) \sqrt{2N\beta^{-1}}(d(G)+k+1)^k +\\
      {k(k+1)} \,
    \gamma
    \sqrt{N}
    (\hat{d}(G) + \gamma \sqrt{\hat{d}(G)} + \gamma)^{k-1} 
    = \tilde{O}(\sqrt{N} d(G)^{k}),
\end{align*}
where $\gamma = (k{+}1)\sqrt{8\log(12kN/\beta)}/\varepsilon = \tilde{O}(1) $ and $\hat{d}(G) = \max(d(G), \lceil \gamma^2 \rceil)$.
\end{theorem}

\begin{proof}

Both the unbiasedness and the utility bound follow similarly to the proof of Theorem~\ref{the:k-lp_bound}, by combining Lemma~\ref{lem:pattern_sampling_error} and Lemma~\ref{lem:pattern_additive_bound}. \qedhere

\end{proof}

\begin{lemma}
 For any graph $G$, star edge count $k$, parameters $\varepsilon$ and $\beta$, the single-round $k$-star counting mechanism $\mathcal{M}_{k\star}(G, \varepsilon)$ satisfies that, with probability at least $1 - \beta$, 
 \begin{equation*}
    |\mathcal{M}_{k\star}(G, \varepsilon) - Q_{k\star }| = \tilde{O}(\sqrt{N}d(G)^{k-1}),
 \end{equation*}
where $Q_{k\star}$ denotes the true count of the number of $k$-stars in $G$.
\end{lemma}

\begin{proof}
To begin with, the noisy degree at node $v_i$ is
\begin{equation*}
    \tilde{d}(v_i) = d(v_i) + \eta_i, 
\end{equation*}
where $\eta_i = \mathrm{Lap}(2/\varepsilon)$.
Applying the Laplace concentration bound with failure probability $\beta/2N$ per node, with probability at least $1 - \beta/2$, 
$ | \eta_i | \leq \gamma$ for all $i$, where $\gamma = 2\sqrt{8\log(4N/\beta)}/{\varepsilon} = \tilde{O}(1) $. 

Let $(x)_k = x(x-1)\cdots(x-k+1)$, and $(x)_0 = 1$. 
Then,
\begin{equation*}
(x+y)_{k}=\sum_{s=0}^{k} \binom{k}{s}\,(x)_{k-s}\,(y)_{s}. 
\end{equation*}
Summing up all nodes' results, the total error is 
\begin{align*}
|\mathcal{M}_{k\star}(G, \varepsilon) - Q_{k\star }| &= \bigg| \sum_{i=1}^N (d(v_i))_k - (\tilde{d}(v_i))_k \bigg| \\
& = \bigg|\sum_{i=1}^N\sum_{s=1}^{k} \binom{k}{s}\,(d(v_i))_{k-s}\,(\eta_i)_{s} \bigg| \\
& = \bigg| \sum_{i=1}^N 
\eta_i
\cdot 
\sum_{s=1}^{k} \binom{k}{s}\,(d(v_i))_{k-s}(\eta_i-1)_{s-1} 
\bigg|.
\end{align*}
Finally, applying the Laplace concentration bound with failure probability $\beta/2N$, with probability at least $1- \beta/2 - \beta/2N \geq 1-\beta$,
\begin{align*}
    |\mathcal{M}_{k\star}(G, \varepsilon) - Q_{k\star }| 
    &\leq \gamma \sqrt{
    \sum_{i=1}^N
    \left(
    \sum_{s=1}^{k} \binom{k}{s}\,(d(v_i))_{k-s}(\eta_i-1)_{s-1}
    \right)^2 } \\
    &\leq \gamma \sqrt{N} \cdot \max_i 
    \left(
    \sum_{s=1}^{k} \binom{k}{s}\,(d(v_i))_{k-s}(\eta_i-1)_{s-1}
    \right) \\
    &\leq \gamma \sqrt{N} \cdot \max_i 
    \left(
    \sum_{s=1}^{k} \binom{k}{s}\,d(G)^{k-s}|\gamma+k|^{s-1}
    \right) \\
    &= \tilde{O}(\sqrt{N} d(G)^{k-1}). \qedhere
\end{align*}
\end{proof}


\begin{lemma}
For any graph $G$, walk length $k$, and privacy budget $\varepsilon$, the RR-based method introduced in Section~\ref{subsec:strawman} outputs an estimator $\hat{W}_k$ of the true count $W_k$ that $\mathbb{E}[\hat{W}_k] = W_k$,
and
$\mathrm{Var}(\hat{W}_k) = \Omega({N^{k+1}})$.
\end{lemma}

\begin{proof}
To begin with, since $W_k = \sum_w Z_w$ (over all $(k+1)$-node sequences) and $\mathbb{E}[\hat{Z}_w] = Z_w$, it follows that $\hat{W}_k = \sum_w \hat{Z}_w$ is an unbiased estimator of $W_k$.

The variance of $\hat{W}_k$ can be represented as
\begin{equation*}
    \mathrm{Var}(\hat{W}_k) = \sum_w \mathrm{Var}(\hat{Z}_{w}) 
    + \sum_{w \neq w'} \mathrm{Cov}(\hat{Z}_{w}, \hat{Z}_{w'}).
\end{equation*}
We show that $\mathrm{Cov}(\hat{Z}_{w}, \hat{Z}_{w'}) \geq 0$.
Let $\mathcal{E}_{\mathrm{dist}}(w)$ denote the set of distinct edges along the walk $w$, and define  
$\mathcal{E}_\triangle = \mathcal{E}_{\mathrm{dist}}(w) \,\Delta\, \mathcal{E}_{\mathrm{dist}}(w')$  
as the symmetric difference, and  
$\mathcal{E}_\cap = \mathcal{E}_{\mathrm{dist}}(w) \cap \mathcal{E}_{\mathrm{dist}}(w')$  
as the intersection of the distinct edges in the two $k$-line walks corresponding to $w$ and $w'$.
Recalling that $ \hat{Z}_{w} = \prod_{(i,j)\in \mathcal{E}_{\mathrm{dist}}(w)} \hat{a}_{i,j}$,
we have:
\begin{align*}
\mathrm{Cov}(\hat{Z}_{w},\hat{Z}_{w'}) 
&= \mathbb{E}[\hat{Z}_{w} \hat{Z}_{w'}]-\mathbb{E}[\hat{Z}_{w}]\mathbb{E}[\hat{Z}_{w'}] \\
&= \prod_{(i,j)\in \mathcal{E}_\triangle} \mathbb{E}[\hat a_{i,j}]
   \;\prod_{(i,j)\in \mathcal{E}_\cap} 
   \bigl(\mathbb{E}[\hat a_{i,j}^2]-\mathbb{E}[\hat a_{i,j}]^2\bigr) \\[2pt]
&= \prod_{(i,j)\in \mathcal{E}_\triangle} \mathbb{E}[\hat a_{i,j}]
   \;\prod_{(i,j)\in \mathcal{E}_\cap} \mathrm{Var}(\hat a_{i,j})
   \;\;\ge 0.
\end{align*}
Therefore, the total variance of $\hat{W}_k$ is at least $\sum_w \mathrm{Var}(\hat{Z}_{w})$.

For a fixed $w$, we have
\begin{equation*}
    \operatorname{Var}(\hat{Z}_{w})
= \mathbb{E}[ \hat{Z}_{w}^{2} ] - \mathbb{E}[\hat{Z}_{w}]^{2},
\end{equation*}
where $\mathbb{E}[\hat{Z}_{w}] \in \{0,1\}$, and 
\begin{align*}
    \mathbb{E}[\hat{Z}_{w}^{2}] 
    &= \prod_{(i,j)\in \mathcal{E}_{\mathrm{dist}}(w)}\mathbb{E}[\hat{a}_{i,j}^2], \\
    &= \prod_{(i,j)\in \mathcal{E}_{\mathrm{dist}}(w)} (\mathrm{Var}(\hat{a}_{i,j}) + \mathbb{E}[\hat{a}_{i,j}]^2).
\end{align*}
Because $\mathbb{E}[\hat{Z}_{w}] \leq \mathbb{E}[\hat{a}_{i,j}] \in \{0, 1\}$, we conclude that $\operatorname{Var}(\hat{Z}_{w}) \geq \mathrm{Var}(\hat{a}_{i,j})^{|\mathcal{E}_\mathrm{dist}(w)|} = (e^\varepsilon/(e^\varepsilon{-}1)^2)^{|\mathcal{E}_\mathrm{dist}(w)|}$. 
Consider the subset of all $w$ where all nodes are distinct; this implies $|\mathcal{E}_\mathrm{dist}(w)| = k$, i.e., all edges are also distinct. 
Since $k$ is constant, the number of such qualified $w$ is at least $N(N{-}1)\cdots(N{-}k) = \Omega(N^{k+1})$.


Therefore, by summing the variances over all oriented $k$-line walks consisting of distinct nodes (and hence distinct edges), we obtain a lower bound on the variance of $\hat{W}_k$ as follows:
\begin{align*}
    \mathrm{Var}(\hat{W}_k)
    &\geq \sum_{w: |\mathcal{E}_{\mathrm{dist}}(w)| = k} \mathrm{Var}(\hat{Z}_{w}) \\
    &\geq  \frac{e^{\varepsilon k}}{(e^\varepsilon-1)^{2k}} \cdot \Omega(N^{k+1}) \\
    &= \Omega(N^{k+1}). \qedhere
\end{align*}
\end{proof}

\fi

\end{document}